\documentclass{amsart}
\usepackage{amssymb,amsmath,amsfonts,stmaryrd}

\usepackage{mathtools}
\usepackage[utf8]{inputenc}
\usepackage[final,nopatch=footnote]{microtype}

\usepackage{leftidx}
\usepackage{comment}
\usepackage{booktabs}
\usepackage{yfonts}
\makeatletter

\usepackage[usenames, dvipsnames]{xcolor}
\usepackage[backref=page, colorlinks=true]{hyperref}
\hypersetup{
 linkcolor={red!50!black},
 citecolor={green!50!black},
 urlcolor={blue!80!black}
}
\usepackage{annotate-equations}

\usepackage[margin=1.3in]{geometry}
\usepackage[capitalize, noabbrev]{cleveref} % TODO: package options
\numberwithin{equation}{section}
\usepackage{amsthm}
\newtheorem*{thm*}{Theorem}
\newtheorem{conj}[equation]{Conjecture}
\newtheorem*{conj*}{Conjecture}

\newtheorem{thm}[equation]{Theorem}
\newtheorem{lem}[equation]{Lemma}
\newtheorem*{notation*}{Notation}
\newtheorem{prop}[equation]{Proposition}
\newtheorem{cor}[equation]{Corollary}

\newtheorem{question}[equation]{Question}
\theoremstyle{definition}
\newtheorem{defn}[equation]{Definition}
\newtheorem{exm}[equation]{Example}

\theoremstyle{remark}
\newtheorem{rem}[equation]{Remark}

\crefname{thm}{Theorem}{Theorems}
\crefname{prop}{Proposition}{Propositions}
\crefname{lem}{Lemma}{Lemmas}
\usepackage{mathtools}
\DeclarePairedDelimiter{\set}{\{}{\}}

\usepackage{adamsmacros}
\usepackage{spectralsequences}
\usepackage{tikz-cd}
\usepackage{dsfont}

\usepackage{booktabs}

\usetikzlibrary{shapes.geometric,calc}
\usetikzlibrary{shapes,arrows,chains}
\usetikzlibrary{decorations.markings}
\usetikzlibrary{decorations.pathmorphing}
\usetikzlibrary{shapes.multipart}
\tikzset{snake it/.style={decorate, decoration=snake}}
\tikzstyle{GraphNode}=[circle, draw=black, fill=black, inner sep=2pt, minimum size=5pt]
\tikzstyle{GraphEdge}=[black]
\usepackage{subcaption}
\usepackage{setspace}
\usetikzlibrary{calc}
\usetikzlibrary{tikzmark}
\usepackage[mathscr]{eucal}
\makeatletter
\def\instring#1#2{TT\fi\begingroup
  \edef\x{\endgroup\noexpand\in@{#1}{#2}}\x\ifin@}
\def\isuppercase#1{%
  \instring{#1}{ABCDEFGHIJKLMNOPQRSTUVWXYZ}%
}%
\newcommand{\C@lIfUpper}[1]{
 \if\isuppercase{#1}\mathscr{#1}%
 \else #1%
 \fi
}
\newcommand{\cat}[1]{\mathit{\@tfor\next:=#1\do{\C@lIfUpper{\next}}}}
\makeatother

\usepackage{arydshln}

\usepackage{xparse}
\newcommand{\Z}{\mathbb Z}

\newcommand{\cT}{\mathcal T}

\newcommand{\Spin}{\mathrm{Spin}}

\newcommand{\MTSpin}{\mathit{MTSpin}}

\newcommand{\ko}{\mathit{ko}}

\newcommand{\Sq}{\mathrm{Sq}}
\newcommand{\cP}{\mathcal{P}}
\newcommand{\cK}{\mathcal{K}}

\newcommand{\pt}{\mathrm{pt}}

\newcommand{\matt}[1]{{\bf \color{red} [MY: #1]}}
\newcommand{\arun}[1]{{\bf \color{blue} [AD: #1]}}

\newcommand{\SU}{\mathrm{SU}}

\newcommand{\SO}{\mathrm{SO}}

\newcommand{\GL}{\mathrm{GL}}

\renewcommand{\O}{\mathrm O}
\newcommand{\id}{\mathrm{id}}
\newcommand{\RP}{\mathbb{RP}}

\newcommand{\Pin}{\mathrm{Pin}}

\newcommand{\CP}{\mathbb{CP}}

\usepackage{xspace}
\newcommand{\pinp}{pin\textsuperscript{$+$}\xspace}
\newcommand{\pinm}{pin\textsuperscript{$-$}\xspace}
\newcommand{\spinc}{spin\textsuperscript{$c$}\xspace}
\newcommand{\pinc}{pin\textsuperscript{$c$}\xspace}

\newcommand{\term}{\emph}

\newcommand{\cA}{\mathcal A}

\newcommand{\Ext}{\mathrm{Ext}}

\newcommand{\MTSO}{\mathit{MTSO}}
\newcommand{\MTPin}{\mathit{MTPin}}

\definecolor{violet}{RGB}{148,0,211}
\definecolor{DarkGreen}{RGB}{0,150,0}
\definecolor{amber}{rgb}{1.0,0.75,0.0}
\newcommand\MAILTO[1]{\href{mailto:#1}{\nolinkurl{#1}}}

\DeclareDocumentCommand{\shortexact}{s O{} O{} mmmm}{
\IfBooleanTF{#1}{ % if star
\begin{tikzcd}[ampersand replacement=\&]
	{1} \& {#4} \& {#5} \& {#6} \& {1#7}
	\arrow[from=1-1, to=1-2]
	\arrow["#2", from=1-2, to=1-3]
	\arrow["#3", from=1-3, to=1-4]
	\arrow[from=1-4, to=1-5]
\end{tikzcd}
}{ % no star
\begin{tikzcd}[ampersand replacement=\&]
	{0} \& {#4} \& {#5} \& {#6} \& {0#7}
	\arrow[from=1-1, to=1-2]
	\arrow["#2", from=1-2, to=1-3]
	\arrow["#3", from=1-3, to=1-4]
	\arrow[from=1-4, to=1-5]
\end{tikzcd}
}}

\let\oldtocsection=\tocsection
\let\oldtocsubsection=\tocsubsection
\let\oldtocsubsubsection=\tocsubsubsection
\renewcommand{\tocsection}[2]{\hspace{0em}\bfseries\oldtocsection{#1}{#2}}
\renewcommand{\tocsubsection}[2]{\hspace{1em}\oldtocsubsection{#1}{#2}}
\renewcommand{\tocsubsubsection}[2]{\hspace{2em}\oldtocsubsubsection{#1}{#2}}

\usepackage[normalem]{ulem}

\newcommand{\mc}{\mathcal}

\definecolor{sanddune}{rgb}{0.59, 0.44, 0.09}
\definecolor{darkblue}{RGB}{0,0,102}
\definecolor{darkred}{rgb}{0.5,0.,0.}
\definecolor{BlueViolet}{RGB}{138,43,226}
\definecolor{SkyBlue}{RGB}{30,144,255}
\definecolor{DarkGreen}{RGB}{0,100,0}

\crefname{defn}{Definition}{Definitions}

\title{Global Structure in the Presence of a Topological Defect}

\author{Arun Debray}
\address{Department of Mathematics, University of Kentucky,
719 Patterson Office Tower,
Lexington, KY 40506-0027}
\email{\href{mailto:a.debray@uky.edu}{a.debray@uky.edu}}

\author{Weicheng Ye}
\address{Department of Physics and Astronomy, and Stewart Blusson Quantum Matter Institute,
University of British Columbia, Vancouver, BC, Canada V6T 1Z1}
\email{\href{mailto:victoryeofphysics@gmail.com}{victoryeofphysics@gmail.com}}

\author{Matthew Yu}
 \address{Mathematical Institute, University of Oxford, Woodstock Road, Oxford, OX2 6GG, UK}
 \email{\href{mailto:yumatthew70@gmail.com}{yumatthew70@gmail.com}}

\begin{document}

\raggedbottom
\vfuzz=3pt

\begin{abstract}
We investigate the global structure of topological defects which wrap a submanifold $F\subset M$ in a quantum field theory defined on a closed manifold $M$. The Pontryagin--Thom construction oversees the interplay between the global structure of $F$ and the global structure of $M$. We will employ this construction in two distinct mathematical frameworks with physical applications.
The first framework is the concept of a characteristic structure, consisting of the data of pairs of manifolds $(M,F)$ where $F$ is Poincaré dual to some characteristic class. 
This concept is discussed in the mathematics literature and shown here to have meaningful physical interpretations related to defects.
In our examples, we will mainly focus on the case where $M$ is 4-dimensional and $F$ has codimension 2. % where the topological defects we consider live. %when the characteristic classes are $\Z/2$-valued degree 2 classes. %The topological defects we consider will be codimension 2 and have support on $F$.
The second framework uses obstruction theory and the fact that spontaneously broken finite symmetries leave behind domain walls, to determine the conditions on which dimensions a higher-form finite symmetry can spontaneously break. We explicitly study the cases of higher-form $\Z/2$ symmetry, but the method can be generalized to other groups.
\end{abstract}

\thanks{
We would like to thank
Cameron Krulewski,
Yu Leon Liu,
Natalia Pacheco-Tallaj,
Ryan Thorngren,
and
Kevin Walker
for helpful discussions regarding this paper. WY was supported by the Natural Sciences and Engineering Research Council of Canada (NSERC) and the European Commission under the Grant Foundations of
Quantum Computational Advantage. 
MY is supported by the EPSRC Open Fellowship EP/X01276X/1. He would also like to thank the University of Kentucky, where part of this work was completed. The authors contribute equally and are listed in alphabetical order. No authors have competing
interests to declare that are relevant to the content of this article. }

\maketitle

\tableofcontents

\section{Introduction}
Topological defects are essential tools in understanding the properties of quantum field theory (QFT), and also connect to various phenomena in condensed matter and cosmology. Understanding properties regarding topological defects is a framework that has produced significant advancements in quantum field theory \cite{Oshikawa:1996ww,Oshikawa:1996dj,Frohlich:2004ef,Frohlich:2006ch,Antinucci:2024izg}, and continues to be an active and evolving area of research \cite{Cuomo:2021rkm,Roy:2021jus,Samanta:2023fvs,Wang:2022rmd,Zhou:2023fqu}. Moreover, modern perspectives on symmetries associate these topological defects with generalized symmetries \cite{Gaiotto:2014kfa}. The structure of symmetries is characterized by the fusion, braiding, etc, of these topological defects, which gives the \emph{algebraic} structure of topological defects.

On the other hand, to understand the properties of a quantum field theory, sometimes it is helpful to put the theory on a nontrivial manifold with nontrivial \emph{global} structure, which helps us understand many properties of the theory, such as its RG flows and anomalies (see, for example, \cite{Fre93,Lurie2009,Witten:2016cio,Tachikawa:2016cha,Komargodski:2017keh,Cordova:2019bsd,FH21InvertibleFT}). In the framework of topological field theory (TFT), this global structure is exactly the \emph{tangential} structure of the underlying manifold, such as orientation, spin structure, pin$^\pm$ structure, etc. Analyzing the tangential structure helps classify TFTs and the quantum phases with symmetries associated to these global structures \cite{FH21InvertibleFT,freed2019lectures,Luuk2021}.

However, it is less common to incorporate the tangential structure of a defect into this analysis. Traditionally, when discussing defects and their structures---such as fusion, braiding, and related properties---one typically does not even place the defect on a nontrivial manifold, but rather on flat space, suggesting a focus solely on their \emph{local} structure i.e.\ the space near the defect. More recently, there has been a great deal of research in the mathematical physics literature incorporating defects into the Atiyah--Segal approach to field theory, sometimes under the name \emph{defect TQFT}, including \cite{Fre93,Lurie2009,Davydov:2011kb,BCP14a,BCP14b,CR12,AFT17,CRS19,CRS20, CMRSS21,FT2022,CEG23,CM23,Fre24,FMM24,FMT24,Car23}. In these works, though, the tangential structure of the defect does not generally come into consideration.

It has been recently proposed that in the setting of spontaneous symmetry breaking (SSB), an algebro-topological map called the \term{Smith homomorphism} relates the anomaly on the symmetry defect with the anomaly in the bulk \cite{Hason:2020yqf,COSY20}. Later, the Smith homomorphism was expanded into a long exact sequence relating theories in different dimensions \cite{Debray:2023ior,Debray:2024wxm}, where a nontrivial interplay between the tangential structure of defects implementing SSB and the tangential structure of the bulk plays a central role. This has been applied in~\cite{DL23, DYY23,DDHM23, Deb23, DNT24, DK25} to study questions in both physics and mathematics.

In this light, we want to make a detailed investigation of the global structure of defects in QFTs from the point of view of algebraic topology. The purpose of this paper is to study:
\begin{question}\label{question:main1}
    What can we learn by taking into account the global structure of both the manifold and the defect in a QFT?
\end{question}
 Consider a QFT $\mathcal{T}$, which can be defined on a closed manifold $M$ with tangential structure $\xi$. For example, $\mathcal{T}$ may be a theory that involves local fermions and hence can only be defined on a spin manifold $M$ with some chosen spin structure. In the presence of non-interacting topological defects, we can consider putting the defect on a submanifold $F$ embedded in $M$, denoted as $F\hookrightarrow M$. However, now the QFT $\mathcal{T}$ is only defined on $M\backslash F$, the submanifold of $M$ after deleting $F$, and generally speaking we only need to choose the tangential structure $\xi$ on $M\backslash F$. In particular, we do not need to choose the tangential structure $\xi$ on $M$ itself, and $M$ may not even have a $\xi$-structure. This will be the starting point of our analysis. The Pontryagin--Thom construction provides a unified mathematical framework for establishing the existence of such $F$ given a manifold $M$ that does not necessarily have a $\xi$-structure. By leveraging this topological construction, we can exhibit new consequences of the global structure in physical theories. 

We will study two instances where \cref{question:main1} arises in a QFT. The first is an application of the Pontryagin--Thom construction which can be used to study global properties of external defects, or impurities, inserted in a theory. When $M$ is not necessarily compatible with $\xi$, we assume that the defects are supported on an \emph{embedded} %\footnote{In principle, we can also consider \emph{immersed} manifolds instead of \emph{embedded} manifolds of $M$, where the defect is supported. The situation seems to be more complicated. See e.g. \cite{crowley2024immersed} for recent development.} 
manifold $F$ which is Poincaré dual to some cohomology class, such that $\xi$ is still compatible with $M\backslash F$. Of particular interest in this work will be degree 2 cohomology classes with $\Z/2$-coefficients, such as the second Stiefel-Whitney class $w_2\in H^2(M;\Z/2)$. 

When considered as a pair $(M,F)$, i.e. a manifold with a particular defect wrapped on a submanifold $F$ with certain properties, one is led to the concept of \emph{characteristic pairs}. 
\begin{defn}\label{def:sign_charbord}(\cref{def:charbordism})
     Let $M$ be a manifold with $\xi$-tangential structure, and let $\mathcal P\in H^n(M; A)$. A $(\xi, \mathcal P)$-\textit{characteristic pair} consists of a pair $(M,F)$ where $F$ is a proper submanifold of $M$ Poincaré dual to a cohomology class $\mathcal{P}(M)$ of $M$, and the boundary of $M$ intersects $F$ precisely and transversely at the boundary of $F$. We will use the notation $\mathit{MTChar}(\xi,\mathcal P)$ to denote the Thom spectrum of manifolds with $(\xi,\mathcal P)$-characteristic structures on their tangent bundles.
\end{defn}
%As we will elaborate in \S\ref{section:preliminaries}, the manifold $M\backslash F$ will have a particular tangential structure that does not extended over $F$ and $\mathcal{T}$ must now be compatible with that structure.\footnote{Other applications of characteristic pairs have been studied by Walker see e.g. \cite{walk2017low}.} For a given $M$ and characteristic class $\mathcal{P}$ there is no guarantee that one can find a submanifold $F$ so that $F \hookrightarrow M$ in such a way that the structure on $M$ and $F$ are compatible. 
A manifold $M$ with an embedded submanifold $F$ of the above form is said to have a \textit{characteristic structure}. Hence we propose:
\begin{conj}
    QFTs with
defects, or impurities, inserted along particular submanifolds that are Poincaré dual to a characteristic class of $M$ are described by characteristic structures.
\end{conj}

Kirby--Taylor~\cite[\S 2]{KT90} showed how to use the Pontryagin--Thom construction to compute the group of characteristic pairs considered up to bordism, known as a \textit{characteristic bordism group}, using standard techniques in algebraic topology (see also Crowley-Grant~\cite[\S 1.1]{crowley2024immersed}). By work of Freed--Hopkins~\cite{FH21InvertibleFT} and Grady~\cite{Gra23}, these bordism groups determine the groups of deformation classes of reflection-positive invertible field theories (IFTs) on manifolds with characteristic structures. These IFTs provide a mathematical model for 't Hooft anomalies of unitary QFTs, so characteristic bordism calculations reveal abstract properties about anomalies associated to theories with defects. This seems to connect to the notion of \emph{defect anomaly} recently proposed in \cite{Antinucci:2024izg}. It would be interesting to connect it with the physical considerations from field theory in a more straightforward manner, similar to the connection between ordinary bordism groups and 't Hooft anomalies in e.g.\ \cite{Dai:1994kq,Witten:2016cio}.%\matt{mention local symmetries i.e. smith and char sequences. maybe in the results section, or after the next paragraph so we dont distrupt the narrative here}

The second context in which \cref{question:main1} arises, and which we will examine using the Pontryagin--Thom construction, is in the framework of spontaneous symmetry breaking of a finite $n$-form global symmetry, first introduced in \cite{Gaiotto:2014kfa} with applications explored in \cite{Hofman:2018lfz,Zhao:2020vdn,Qi:2020jrf,Rayhaun:2021ocs,Pace:2023ccj,Zhang:2024fpf,Liu:2023eml,Liu:2026tcl}. In particular, like the spontaneous breaking of zero-form symmetries, probing the spontaneous breaking of higher-form symmetries requires introducing a domain wall or a defect supported on a lower-codimensional manifold where the symmetry's order parameter vanishes. Consequently, the bulk manifold $M$ and the lower-codimensional manifold $F$ naturally form a characteristic pair. However, unlike the codimension-1 domain walls of zero-form symmetries, which always exist, such a lower-codimensional manifold $F$ may not exist for higher-form symmetries.

\begin{defn}\label{def:SSB}
     For a quantum field theory defined on a manifold $M$, we define the spontaneous symmetry breaking of an $n$-form $A$-symmetry on the manifold $M$ as the existence of a characteristic pair $(M,F)$, where $F$ is Poincaré dual to the cohomology class $\mathcal P(M) \in H^{n+1}(M;A)$, such that the quantum field theory can be extended to the characteristic pair.
\end{defn}

On closed manifolds, the aforementioned domain walls will be supported on the submanifold $F$, and a goal of this work will be to explain how to compute the obstructions to finding the Poincaré dual submanifolds. This definition we provide here is not mathematically precise, since we do not attempt to give a precise meaning of the ``extension''. Still, we conjecture that the obstructions to finding the Poincaré dual submanifolds are exactly the obstructions to the extension. 
Moreover, if the obstruction does not vanish for a manifold $M$, then a theory with a spontaneously broken $n$-form symmetry cannot exist on $M$. This topological obstruction is analogous to the obstruction preventing the definition of fermionic theories on non-spin manifolds.

In this work, we also conjecture the existence of a
\textit{characteristic long exact sequence} of bordism groups, analogously to the Smith long exact sequence in~\cite{Debray:2024wxm} and generalizing some constructions in Kirby--Taylor~\cite[\S 6]{KT90}.
\begin{conj*}(\cref{char_conj})
Given a pair $(\xi, \mathcal P)$ as in \cref{def:sign_charbord}, and $\xi'$ a tangential structure for a submanifold $F$, there is a map of spectra $\mathcal R\colon\mathit{MTChar}(\xi,\mathcal P)\to \Sigma^n \mathit{MT\xi}'$ such that the map $\mathcal R$ induces on $\pi_k$ is the
map $R\colon \Omega_k^{(\xi,\mathcal P)}\to\Omega_{k-n}^{\xi'}$ sending a characteristic pair $(M, F)\mapsto F$.
\end{conj*}
The key takeaway of this conjecture is that the homotopical data in $\mathcal R$ automatically implies some useful algebraic corollaries. Notably, \cref{char_conj} implies that the map $(M,F)\mapsto F$ on bordism extends to one of the three maps in a long exact sequence, which we call the characteristic long exact sequence. Its form is discussed in \S\ref{subsection:smithLESintro}. 

We give formal justifications on why we anticipate such a sequence to exist; however, we have not provided an explicit description. To address this, we approximate the characteristic long exact sequence using the Smith long exact sequence. Specifically, we propose that there exists a map from the Smith long exact sequence to the characteristic long exact sequence. Moreover, it is known that each term in the Anderson dual of the Smith long exact sequence has physical interpretations as the anomaly of the bulk, the anomaly of the symmetry defect and the family anomaly of a one-parameter space of bulk theories \cite{Debray:2023ior}. %mapping the anomaly of the bulk given the anomaly on a defect
%which is associated with spontaneously breaking a symmetry.
Inspired by this we believe it would be important to in the future investigate:
\begin{question}\label{q:CharandSmith}
    How does a characteristic long exact sequence generalize the Smith long exact sequence, and what does it reveal about obstructions to symmetry breaking?
\end{question}
More specifically, even though homotopy theory tells us that it is always possible, we have not determined the third term in the anticipated long exact sequence in a computationally effective manner. Nevertheless, in this paper we lay the groundwork for studying this question by looking at characteristic structures that should naturally arise in fermionic systems, and (assuming \cref{char_conj}) compute the relevant  bordism groups in the characteristic long exact sequence that we have access to.

\subsection{Summary of Main Results}

The following results pertain to the  characteristic pairs whose names are given in Table \ref{tab:Charbordism}, which display the structures on $M$ and $F$ as well as the characteristic class that $F$ is Poincaré dual to.

\begin{itemize}
  \item We use the Pontryagin--Thom construction to study the characteristic bordism groups for each type of characteristic pair. These characteristic bordism groups treat two characteristic pairs $(M,F)$ and $(M',F')$ as equivalent if there exists a characteristic pair $(X,Y)$ where $X$ is a $(k+1)$-dimensional manifold with the same tangential structure as $M$ and $Y$ is a Poincaré dual submanifold of $X$ such that $\partial (X,Y) = (M,F)\sqcup -(M',F')$. We describe each characteristic structure as a twisted spin structure, summarized below.
   
\begin{table}[htb]
\begin{center}
\begin{tabular}{ c c c} 
\toprule
Name & Twisted Spin Structure & Prop.\\
\midrule
Freedman--Kirby (FK)& spin$^c$  & \ref{prop:FKpairs} \\
%\hline
Freedman--Kirby$^\mathrm{O}$(FK$^{\mathrm{O}
}$)
& pin$^c$ &\ref{prop:FKOcharpair} \\ %\hline 
Guillou--Marin (GM)
& $(M\O_2, 0, U)$ & \ref{prop:GMcharpairs}\\
%\hline
Kirby--Taylor$^-$ (KT$^{-}$) & $(B\O_1\times M\O_2,w_1,U)$ & \ref{prop:KT-bord}
\\ 
%\hline
Kirby--Taylor$^+$ (KT$^{+}$)
& $(B\O_1\times M\O_2,w_1,w_1^2 + U)$ & \ref{prop:KT+structures}\\
\bottomrule
\end{tabular}
\end{center}
\caption{This table summarizes which twisted spin structure each characteristic structure is equivalent to, and the proposition that proves the equivalence. We note that spin$^c$ and pin$^c$ structures are both twisted spin structures e.g.\ a spin$^c$ structure is a $(B\SO_2,L)$-twisted spin structure, where $L$ is the tautological bundle.}
\label{tab:results1}
\end{table}

    \item Using this translation, we compute these characteristic bordism groups up to degree 4 using the Adams spectral sequence. We give the results of these homotopy groups in Table \ref{tab:results2} below:
    \begin{table}[htb]
\begin{center}
\begin{tabular}{ c c c c c c c} 
\toprule
Name & $\Omega_0$ & $\Omega_1$ & $\Omega_2$ & $\Omega_3$ & $\Omega_4$ &Prop.\\
\midrule
Freedman--Kirby (FK)&  $\Z$ &0  & $\Z$ & 0 & $\Z^{\oplus 2}$ & - \\
Freedman--Kirby$^\mathrm{O}$(FK$^{\mathrm{O}
}$)
& $\Z/2$ &0  & $\Z/4$ & 0 & $\Z/8 \oplus \Z/2$ & -\\ 
Guillou--Marin (GM)
& $\Z$ &0  & 0 & 0 & $\Z^{\oplus 2}$ & \ref{prop:bordcompGM}\\
%\hline
Kirby--Taylor$^-$ (KT$^{-}$) & $\Z/2$  & 0  & $\Z/2$ & 0  & $\Z/2^{\oplus3}$& \ref{prop:bordcompKT-}
\\ 
%\hline
Kirby--Taylor$^+$ (KT$^{+}$)
& $\Z/2$  & 0& $\Z/2$& 0& $\Z/2^{\oplus 3}$& \ref{prop:bordcompKT+}\\
\bottomrule
\end{tabular}
\end{center}
\caption{This table summarizes the result of the corresponding bordism groups in degrees 0-4. For FK and FK$^{\mathrm{O}
}$, whose associated tangential structures are spin$^c$ and pin$^c$ respectively, these bordism groups were computed by Anderson--Brown--Peterson \cite{ABP67}, resp.\ Bahri-Gilkey~\cite{BG87a, BG87b}.
For the other tangential structures, we also list the propositions where we perform the calculations explicitly.}
\label{tab:results2}
\end{table}
    \item We discuss the long exact sequences in characteristic bordism in \S\ref{subsection:smithLESintro} and use the Smith long exact sequence to approximate these sequences. Specifically, the Smith long exact sequence reflects the symmetry structure that emerges locally around the defect.
    We show in Proposition \ref{prop:FKapprox} and \ref{prop:FKOapprox} that certain Smith long exact sequences involving spin$^c$ and pin$^c$ bordism are equivalent to the characteristic long exact sequences for FK and FK$^\mathrm{O}$ pairs. We furthermore show that there are maps from certain Smith long exact sequences associated to spin$\text{-}\O_2$ and pin$^\pm\text{-}\mathrm{O}_2$ structures to the characteristic long exact sequences of GM and KT$^\pm$.
    \item For each of the characteristic pairs, we give physical interpretations for the global structure of QFTs in the presence of a defect wrapped on a submanifold $F$. 
\end{itemize}
Specifically, assuming existence of the characteristic long exact sequence (\cref{char_conj}), we derive the following physics predictions from our computations:
\begin{itemize}
    \item \Cref{GM_physics}: for a $k$-dimensional field theory $Z$ defined on manifolds with Guillou--Marin structure, i.e.\ an oriented theory with a defect Poincaré dual to $w_2$, the anomaly of $Z$ cannot be compatible with a nonzero anomaly on the defect for $k = 0$, $1$, or $3$; for $k = 2$, there is a $\Z$-valued obstruction to compatibility of the bulk and defect anomalies.
    \item \Cref{KT_minus_physics}: if $Z$ has a $\mathrm{KT}^-$-structure (which we have interpreted as a time-reversal symmetry and a defect Poincaré dual to $w_2 + w_1^2$), the anomaly of $Z$ cannot be compatible with a nonzero anomaly on the defect for $k \le 2$; for $k = 3$ there is an obstruction for the anomalies in the bulk and defect to be compatible.
    \item \Cref{KT_plus_physics}: 
if $Z$ has a $\mathrm{KT}^+$-structure (which we have interpreted as a time-reversal symmetry  and a defect Poincaré dual to $w_2$), the anomaly of $Z$ cannot be compatible with a nonzero anomaly on the defect for $k = 0,2$; for $k = 3$ the defect anomaly is not compatible with the bulk anomaly, but for $k = 1$ the anomalies of the bulk and defect are always compatible.
%it is always possible to  the anomaly in the bulk and the anomaly of the defect.
\end{itemize}
It would be interesting to make these obstructions explicit in examples.

With regard to the second application involving spontaneous symmetry breaking we show the following:
\begin{itemize}
  \item We use the Pontryagin--Thom construction to give Definition \ref{def:obstructionpoly} for an obstruction to spontaneously breaking a finite higher-form symmetry. We then identify the cohomology class corresponding to such an obstruction in Proposition \ref{prop:obstr}. 
    \item  For the case of a $\Z/2$ valued 1-form symmetry, we show in  \Cref{cor:obstruction1form} that the primary obstruction to spontaneously breaking the symmetry is in terms of a degree 5 $\Z/2$-valued cohomology class. This means that if the class is trivialized when pulled back to some  5-manifold $M$, then spontaneously breaking the $\Z/2$ 1-form symmetry on $M$ is not topologically forbidden. As a consequence this means that 1-form symmetries of Yang-Mills theory in 4d are not obstructed at all by topology, and can spontaneously break e.g. in the Coulomb phase. The example of Yang-Mills is explained in \cref{ex:4dYM}. In \cref{ex:5dtheory} we discuss a supersymmetric 5d theory with a spontaneously broken 1-form symmetry, which also falls into the scope of what is allowed by \Cref{cor:obstruction1form}.
    \item For the case of a $\Z/2$ valued 2-form symmetry we show in  \Cref{prop:breaking2form} that the dimension that one might expect to encounter the first  topological obstruction to spontaneously breaking this symmetry is not naively what one would expect by inspecting the case of 1-form symmetries. In particular, one has to go higher than 6 to identify the first obstruction. This feature also holds for $n$-form symmetries where $n>2$. We give in \cref{ex:6d2form} a discussion of a 6d SCFT which has a spontaneously broken 2-form symmetry, that is unobstructed by \Cref{prop:breaking2form}.
\end{itemize}

\begin{rem}
     Our results support the fact that in the examples we present,  spontaneous symmetry breaking is not obstructed. It would be especially interesting to find an example of a theory where the obstruction to spontaneously breaking a higher form discrete symmetry is nontrivial. Hence, special care would need to be taken in specifying the background manifold, if one were to study the effects of spontaneous symmetry breaking.
\end{rem}

This paper proceeds as follows: In \S\ref{section:preliminaries} we review the topological background necessary for implementing the constructions in the rest of the paper. In particular we review spectra and Thom spaces, and in \S\ref{subsection:duals} we review the properties of Poincar\'e duality and the Pontryagin--Thom construction. We devote \S\ref{section:Defects} to discussing applications of five types of characteristic structures in the context of studying the global structure of spacetime with defects. In \S\ref{subsubsection:approxGM} and \S\ref{subsubsection:approxKT} we explain how the Smith long exact sequence approximates the characteristic long exact sequences, in line with Question \ref{q:CharandSmith}.
In \S\ref{section:SSB} we explore obstructions to breaking higher-form symmetries using the Pontryagin--Thom construction. 

\section{Topological Preliminaries}
\label{section:preliminaries}
%%%%%%%%%%%%%%%%%%%%%%%%%%%%%%%%%%%%%
We begin by outlining the mathematical framework that will be applied in subsequent sections. First, in \cref{subsec:tangential}, we review tangential structures and discuss how they are modified by vector bundle twists. Following this, we introduce Thom spaces and Thom spectra, emphasizing their role in calculating bordism groups via the Pontryagin--Thom theorem. In \cref{subsection:duals}, we cover Poincaré duality as a foundation for the Pontryagin--Thom construction. We then explore characteristic structures and the concept of characteristic bordism. While characteristic bordism is intrinsically geometric, our main goal is to reinterpret it through the Thom spaces introduced earlier. By leveraging this framework alongside the Pontryagin--Thom construction, we simplify the calculation of characteristic bordism groups to a process relying primarily on algebraic methods. Finally, in \cref{subsection:smithLESintro}, we conjecture the existence of a certain long exact sequence associated to characteristic bordism and discuss its implications.

\subsection{Tangential structure}\label{subsec:tangential}

In this subsection, we review the definition of tangential structures along with several related concepts, including Thom spaces, Thom spectra, and the Pontryagin--Thom theorem. None of the material in this subsection is original; we include it here to refresh the reader and standardize notation.

\begin{defn}[Lashof~\cite{Las63}]
\label{def:tangential}
 Given a map $\xi\colon B\rightarrow B\O$ of spaces,  a \textit{$\xi$-structure} on a vector bundle $V\to X$ is a lift of the classifying map $f_V\colon X\to B\O$ of $V$ to $B$, i.e.\ it is a map $\widetilde f_V\colon X\to B$ such that $f_V \simeq \xi\circ\widetilde f_V$. 
\end{defn}
For example, if $\xi$ is the map $B\SO\to B\O$, a $\xi$-structure is equivalent to a reduction of structure group for the principal $\O_n$-bundle of frames to a principal $\SO_n$-bundle -- in other words, an orientation. A $\xi$-structure for $\xi = \id\colon B\O\to B\O$ is no data.

When we say that a manifold $M$ has a $\xi$-structure, that means a $\xi$-structure on $TM$.
\begin{defn}[{\cite[\S 4.1]{Hason:2020yqf}}]
\label{def:twisted}
     Let $V\to X$ be a vector bundle. An \emph{$(X, V)$-twisted $\xi$-structure} on a vector bundle $E\to M$ is data of a map $f\colon M\to X$ and a $\xi$-structure on $E\oplus f^*(V)$. %\arun{Allow virtual bundles, so \pinp is $1-\sigma$ rather than $3\sigma-3$?}
\end{defn}
This work will mostly be concerned with spin structures and their twistings.
\begin{rem}
\Cref{def:twisted} makes sense \textit{mutatis mutandis} when $V$ is a \term{virtual vector bundle}, which is a formal difference of two vector bundles. In particular, $B\O$ can be thought of as classifying virtual vector bundles up to stable equivalence, or rank-zero virtual vector bundles.
\end{rem}
The Whitney sum formula shows that an $(X,V)$-twisted spin structure is equivalent data to the map $f$ and trivializations of
\begin{subequations}
\begin{gather}
    w_1(E) + f^*(w_1(V))\\
    w_2(E) + w_1(E) f^*(w_1(V)) + f^*(w_2(V)).
\end{gather}
\end{subequations}
In other words, $V$ is much more data than we needed, suggesting the following variant of \cref{def:twisted} originating in work of B.L. Wang~\cite[Definition 8.2]{Wan08}.
\begin{defn}\label{nonVB}
Let $X$ be a space, $a\in H^1(X;\Z/2)$, and $b\in H^2(X;\Z/2)$. An \emph{$(X, a, b)$-twisted spin structure} on a vector bundle $E\to M$ is data of a map $f\colon M\to X$ and trivializations of $w_1(E) + f^*(a)$ and $w_2(E) + w_1(E)^2 + f^*(b)$.
\end{defn}

In this definition we do not assume that there is a vector bundle $V$ such that $w_1(V) = a$ and $w_2(V) = b$. Frequently no such $V$ can exist, as discussed in~\cite{GKT89, RWG14, JF19, Kuh20, Spe22, DY:2023tdd, DY22}.% Still, we can use \emph{virtual} vector bundles such that the techniques for real vector bundles can applied with little to no changes. Recall that a virtual vector bundle is a formal difference of real vector bundles, and  $B\O$ is the classifying space of virtual vector bundles. It can also be considered as the classifying space for rank-zero virtual vector bundles by subtracting trivial summands. 

\begin{exm}\label{ex:spinc}
    Let $TM$ be the tangent bundle on an orientable manifold $M$. Let $L\rightarrow M$ be a complex line bundle over $M$. The data of $L$ is equivalent to a map $M \rightarrow B\mathrm{U}_1$ such that $L$ is equivalent to the pullback of the tautological complex line bundle over $B\mathrm{U}_1$. A \spinc structure on $TM$ (see \cref{defn:spinc}) is equivalent to a spin structure on $TM\oplus L \rightarrow M$. By the Whitney sum formula this is equivalent to an identification $w_2(TM)=w_2(L)$. Therefore a spin$^c$ structure is a $(B\mathrm{U}_1,L)$-twisted spin structure (in the language of \cref{def:twisted}), or a $(B\mathrm U_1, 0, w_2)$-twisted spin structure (in the language of \cref{nonVB}).
\end{exm}

Let $V\to X$ be a \emph{real} vector bundle, with a choice of Euclidean metric. Let $D(V)$ be the disk bundle of vectors in $V$ with norm less than or equal to 1, and let $S(V)$ be the sphere bundle of vectors of norm equal to 1.

\begin{defn}
     For $V\to X$  a real vector bundle with Euclidean metric, the \term{Thom space} $\mathrm{Th}(X,V)$ is the quotient $D(V)/S(V)$.
\end{defn}
As the space of Euclidean metrics on a bundle is contractible, the homotopy type of the Thom spectrum does not depend on the choice of metric. However, we choose a metric on the tautological bundle over $B\O_2$, to specify $M\O_2$ on the nose rather than up to homotopy. We will also use the pullback metric on the tautological bundles over $B\SO_2$ and $B\O_1$.
\begin{exm}
    Let $V\rightarrow B\O_n$ be the tautological bundle coming from pulling back the tautological bundle $\mc{V}\rightarrow B\O$, where $\O:= \lim_{n \rightarrow \infty}\O_{n}$ is the  infinite-dimensional orthogonal
group, along the map $B\O_n\to B\O$. The unit sphere bundle $S(V)$ of $V$ is equivalent to $B\O_{n-1}$ and the Thom space $M\O_n = B\O_n/ B\O_{n-1}$. We can do a similar construction and show that $M\SO_n = B\SO_n/B\SO_{n-1}$.
\end{exm}
For much of our discussion, the main object of interest are bordism groups of manifolds with a certain tangential structure and principal bundle. These bordism groups are computable from the point of view of a specific type of Thom spectrum, which we now begin to introduce. 

\begin{defn}
    The \term{Thom spectrum} $X^V$ of a vector bundle $V\rightarrow X$ is the suspension spectrum of the Thom space $\Sigma^\infty \mathrm{Th}(X,V)$.
\end{defn}
Especially, though a virtual vector bundle $V\to X$ does not have a well-defined notion of Thom space, it is still possible to define a Thom spectrum $X^V$ by filtering $X$ by its compact subspaces $K$, and assembling $X^V$ from suspension spectra of Thom spaces over $K$. This was originally done by Lewis~\cite[\S IX.3]{LMSM86}; see there for the details.

 The cohomology of the Thom space can be related to the cohomology of the base space. 
\begin{defn}
    Let $\pi: V\to X$ be an oriented vector bundle and let $A$ be an abelian group. The  Thom class is a class $U \in H^n(\mathrm{Th}(X,V);A)$, which gives rise to an isomorphism 
    \begin{equation}
        H^*(X;A) \cong H^{*+n}(\mathrm{Th}(X,V);A)
    \end{equation}
    of $H^*(X;A)$-modules by cupping with $U$.
\end{defn}
%The Euler class is the pullback of the Thom class along the zero section.
The zero section of any vector bundle $V\to X$ defines an inclusion $X\hookrightarrow V$, which after quotienting by $S(V)$, defines a
map $z\colon X \to \mathrm{Th}(X,V)$ also called the zero section. The pullback $z^*U$ is the Euler class $e(V)\in H^n(X;A)$. We will see a generalization of the Euler class to generalized cohomology in the next section. 

\begin{defn}
Let $\xi: B\rightarrow B\O$ be a tangential structure and let $V \rightarrow B\O$ denote the tautological stable vector bundle. The \term{Madsen--Tillmann spectrum} $MT\xi$ associated to $\xi$ is the Thom spectrum of $\xi^*(-V)\rightarrow B$.
\end{defn}
 Madsen--Tillmann spectra are useful for computing bordism groups by the following theorem.%\arun{Here we should include the standard disclaimer of why we use $MT$ and not $M$.}
\begin{thm}[Pontryagin--Thom]
    There is an isomorphism $\pi_n(MT\xi) \cong \Omega^\xi_n$.\footnote{The usual formulation of the Pontryagin--Thom theorem uses the Thom spectrum of $\xi^*(V)\to B$, not $-V$, and identifies its homotopy groups with the bordism groups of manifolds with a $\xi$-structure on the stable \emph{normal} bundle $\nu\to M$, rather than the stable tangent bundle $TM\to M$. This is equivalent, as stably $TM\simeq -\nu$; we work tangentially because this is how $\xi$-structures come to us most naturally in physics applications.}
\end{thm}

We now aim to turn twisted tangential structures into stable homotopy groups and give them an interpretation from Thom space. This is done through a shearing construction \cite{Debray:2024wxm,DDHM23}.

\begin{lem}[Shearing]\label{lem:shearing}
    Let $S \rightarrow B\O$ be the tautological rank-zero virtual vector bundle and $\eta: B \times X \to B\O$ a tangential structure classified by the virtual vector bundle $\xi^*(S)\boxplus (V-r_V)$ where $r_V$ denotes the rank of $V$. Then an $\eta$-structure is equivalent to a $(X,V)$-twisted $\xi$-structure.
\end{lem}

\begin{cor}\label{shearcor}
    The bordism groups of manifolds with a $(X,V)$-twisted $\xi$-structure are naturally isomorphic to the homotopy groups of the Thom spectrum $MT\xi \wedge X^{V-r_W}$, and the corresponding bordism groups are denoted by $\Omega^{\xi}_*(X^{V-r_W})$.
\end{cor}
The shift by $r_W$ appears because, in order to invoke the Pontryagin--Thom theorem without an extra degree shift, we must normalize to a degree-$0$ virtual bundle.

%%%%%%%%%%%%%%%%%%%%%%%%%%%%%%%%%%%%%%%%%%%%%%%
\subsection{Pontryagin--Thom construction}\label{subsection:duals}
%%%%%%%%%%%%%%%%%%%%%%%%%%%%%%%%%%%%%%%%%%%%%%%

In this subsection, we revisit Poincaré duality and present the Pontryagin--Thom construction, a framework that parametrizes Poincaré dual submanifolds corresponding to a given cohomology class through the analysis of maps into a Thom space. This construction serves as a cornerstone for the two primary applications in this paper. Once again, this subsection does not include new material, but reviews and elaborates on a point of view taken up by Kirby--Taylor~\cite{KT90}.
%First, we state the definition of Poincaré duality.

\begin{thm}[Poincaré duality]
    Let $A$ be a commutative ring, $M$ be a closed oriented $d$-dimensional manifold, and let $[M]\in H_*(M; A)$ be its fundamental class. The cap product
    \begin{equation}
        (-) \cap [M]: H^k(M; A) \xlongrightarrow{\cong}H_{d-k}(M; A)
    \end{equation}
    is an isomorphism for all $k$.
\end{thm}
%\arun{This is a theorem, not a definition\dots? Also, we should make explicit that we need an orientation, or $A = \Z/2$.}\matt{changed}

Choose an element $\omega\in H^k(M;A)$ with $k > 1$. Since $M$ is oriented, $\omega$ has a unique Poincaré dual $W\in H_{n-k}(M;A)$. We would like to find a closed, oriented submanifold $N$ of $M$ such that the inclusion $i\colon N\hookrightarrow M$ sends the fundamental class of $N$ to $W$ -- then we can think of $N$ as the defect associated to the background field $\omega$ of the higher-form symmetry.

As a consequence of the Steenrod realization problem, it is not always possible to find such a submanifold $N$. Moreover, if one wants to make the choice of $N$ ``local'' in the sense of field theory (i.e.\ encode the data of $N$ in the tangential structure), there are additional obstructions coming from the Pontryagin--Thom construction, as noted by Kirby--Taylor~\cite{KT90}. This is starkly different from the case $k = 1$, for which the theory of the Smith homomorphism has been applied to study symmetry breaking~\cite{COSY20, Hason:2020yqf, Debray:2023ior, Debray:2024wxm}: when $k = 1$, it is always possible to make a consistent local choice of the defect $N$, and the choice is well-defined up to bordism. This is related to the fact that not every cohomology class is the Euler class of a vector bundle, though this is true for classes in $H^1(\text{--};\Z/2)$. See also Crowley--Grant~\cite{crowley2024immersed}, who study when an integral homology class can be represented by an immersed submanifold, rather than an embedded submanifold.

%Given an element $\omega \in H^k(M; A)$, its Poincaré dual is an element in $H_{d-k}(M)$, but may not be represented \arun{We should be crystal clear what this means} by a submanifold embedded in $M$. The Pontryagin--Thom construction exactly tells us when a submanifold is Poincaré dual to $\omega$. 
%\arun{Since it's a key aspect of our paper, I think we should do our best to explain why this is not always possible -- it's not exactly the same as the Steenrod problem if I remember right. We should maybe also say something about how it's not automatic that a coh class is the Euler class of a VB, so the Smith story we're used to doesn't work. I knew more than one person who was surprised by each of these facts, so it does us well to mention them}

From the embedding of a submanifold $i\colon N\hookrightarrow M$, the tubular neighborhood theorem states that it is always possible to form a tubular neighborhood of $N$ which is isomorphic to the normal bundle $\nu_N\rightarrow N$, with a map $J\colon \nu_N \rightarrow M$ that chooses a tubular neighborhood of the embedding $i$.
%\arun{The tubular nbhd theorem gives us $J$ for free. There's choices but they're all equivalent, so that the PT collapse map is well-defined up to homotopy}
%given by a vector bundle $\pi: E \rightarrow N$. In particular, the normal bundle $\nu_N \rightarrow N$ is a tubular neighborhood. 
%with a map  $J: E\rightarrow M$ that chooses a tabular neighborhood of the embedding $i$. 
\begin{defn}[Pontryagin--Thom collapse] Given an embedding $i\colon N\hookrightarrow M$ and a map $J:\nu_N \to M$, that maps a tubular neighborhood of $N$ into $M$, the map:
\begin{equation}
    c_i\colon M \rightarrow M/(M\backslash J(\nu_N))\,,
\end{equation}
that collapses the manifold outside of $\nu_N$ to a point is the \emph{Pontryagin--Thom collapse map}.
\end{defn}

This gives rise to the sphere bundle $S(\nu_N)$, which is filled in by the disk bundle $D(\nu_N)$. Therefore, whenever we have an embedded submanifold $N$, we can construct a Thom space via $\nu_N$. 

\begin{lem}\label{lem:excision}
    The map of the pair $(D(\nu_N),S(\nu_N))\to (M,M\backslash N)$ is an equivalence on homology.
\end{lem}
\begin{proof}
    The excision theorem for homology says that for $U \subseteq A\subseteq X$, there is an equivalence $H_n(X\backslash U,A\backslash U)\cong H_n(X,A)$. Let $X = M$, $A=M\backslash N$ and $U = M\backslash D(\nu_N)$. Then $X\backslash U= D(\nu_N)$ and $A\backslash U = S(\nu_N)$ and plugging into the excision theorem proves the equivalence.
\end{proof}

When $N$ is Poincaré dual to an element $a\in H^n(M;\Z/2)$, we have the following property.

\begin{lem}\label{lem:restricta}
   Let $i\colon N \hookrightarrow M$ be a submanifold that is Poincaré dual to $a \in H^n(M;\Z/2)$. The pullback of $a$ to $M\backslash N$ vanishes.
\end{lem}
\begin{proof}
 Since $N$ is obtained by transversality, the normal bundle $\nu_N\rightarrow N$ is identified with the pullback of the canonical bundle over $B\O_n$. By dualizing Lemma \ref{lem:excision}, the relative cohomology $H^n(M,M\backslash N;\Z/2)$ is isomorphic to 
$H^n(D(\nu_N),S(\nu_N);\Z/2)$. The Thom space $\mathrm{Th}(N,\nu_N)$ is equivalent to  $D(\nu_N)/S(\nu_N)$ and hence $H^n(D(\nu_N),S(\nu_N);\Z/2)\cong H^n(\mathrm{Th}(N,\nu_N);\Z/2)$. Invoking the Thom isomorphism gives the equivalence
\begin{equation}
    H^n(\mathrm{Th}(N,\nu_N);\Z/2) \cong H^0(N;\Z/2) = (\Z/2)^{\pi_0(N)}\,,
\end{equation}
 for which the Thom class $U\in H^n(\mathrm{Th}(N,\nu_N);\Z/2)$ restricts to the product of generators in each component. But since the Thom class pulls back along the embedding $i$ to $a$, this means that $a$ vanishes when restricted to $M\backslash N$. 
\end{proof}
%\matt{transition, and explain better} 
%We will frequently make use of the fact that the tangent bundle $TM$ at $N$ decomposes as $TM|_N=TN \oplus \nu_N$.

The submanifold $N$ was realized as the Poincaré dual of $a$, but another way of viewing $N$ is through its normal bundle. In particular, $N$ can be constructed from taking the Poincaré dual of $e(\nu_N)$. Since $a$ restricts trivially to $M\backslash N$ we recover a well-known fact:
\begin{cor}\label{cor:atoN}
   For any class $a\in H^n(M;\Z/2)$ the restriction $i^*(a)$ onto $N$ is $e(\nu_N) \bmod 2$.
\end{cor}
In the examples we will study, the Poincaré dual manifold $N$ will be taken with respect to the class $a$ realized as a combination of $w_2(TM)$ and $w_1(TM)$. A particularly useful way to view $N$ is as an obstruction to extending a tangential structure, that is induced by the trivialization of $a$ on $M\backslash N$, to all of $M$. The cases in the following proposition will be most interesting for our applications to physical systems that involve fermions. Any defect that is placed along $N$ then has the property of acting as a \textit{hard boundary} for the theory on $M\backslash N$. Notably, due to the differing tangential structures, the two theories cannot be made compatible and therefore the degrees of freedom on $M\backslash N$ do not always extend onto the defect.

\begin{prop}[{Kirby--Taylor~\cite[Lemma 2.2, Theorem 2.4]{KT90}}]\label{prop:notextending}
     Let $M$ be a closed manifold without boundary. A submanifold $N$ of $M$ is Poincaré dual to  $w_1(TM)$, $w_2(TM)+w_1(TM)^2$, or $w_2(TM)$ if and only if $M\backslash N$ has an orientation or resp. pin$^-$ or resp. pin$^+$ structure, that does not extend across any component of $N$.
\end{prop}
%\arun{We could put this in a ``Proof sketch'' environment. Also, if this is in Kirby--Taylor for $w_2$ and $w_2 + w_1^2$ we should cite that and also maybe just not include the proof. Up to you}
\begin{proof}
We sketch out the argument in proof of \cite[Lemma 2.2]{KT90} that establishes this proposition in the case when $N$ is Poincaré dual to $w_1(TM)$, and when $N$ has only a single component. The proofs for the two other cohomology classes and when $N$ has multiple components are similar and given in \cite[Theorem 2.4]{KT90}. Begin by picking a class $\lambda \in H^1(M,M\backslash N;\Z/2) \cong  H^1(D(\nu_N),S(\nu_N);\Z/2)= \Z/2$ which evaluates to the nontrivial element if the orientation of the sphere bundle does not extend to the disk bundle, and the trivial element when the orientation does extend. This is equivalent to detecting if the orientation on $M\backslash N$ fails to extend across $N$. The image of $\lambda$ in $H^1(M;\Z/2)$,  denoted $i^*(\lambda)$, is equal to $w_1(TM)$. To see this, consider an embedded $S^1$ in $M$ which is transverse to $N$ subject to the additional condition that if $S^1$ intersects $N$ at a point then it intersects in the subsequent way.  Let the $S^1$ first enter $S(\nu_N)$ at one point, continue down a fiber hitting $N$, and then continuing down the same fiber until it exits the other side of $S(\nu_N)$. The tangent bundle of $M$ restricted to this $S^1$ is oriented if $\langle i^*(\lambda), j^*[S^1]\rangle =1$, where $j^*[S^1]$ is the image of the fundamental class of
the $S^1$ in $H_1(M;\Z/2)$. Since $w_1(TM)$ also shares the same property when paired with $j^*[S^1]$, we see that $i^*(\lambda) = w_1(TM)$.

If we assume $M\backslash N$ has an orientation that does not extend across $N$, then $\lambda$ is nontrivial in $\Z/2$, and is in the image of the Thom class. Hence $N$ is Poincaré dual to $w_1(TM)$. On the other hand if we assume that $N$ is Poincaré dual to $w_1(TM)$, then $M\backslash F$ is oriented, and a fixed orientation on $M\backslash N$ can be modified by a class in $H^0(M\backslash N;\Z/2)$. If we fix a particular orientation on $M\backslash N$ we can consider the class $\lambda$ to probe if this orientation extends across $N$. The idea is to consider if $\lambda$ is
or is not already in the image of the Thom class, i.e. if it equals 1 or $-1$. If $\lambda$ is already in the image of the Thom class then the orientation we fixed does not extend across $V$ and we are done. If $\lambda$ is not in the image of the Thom class then it is always possible to modify $\lambda$ by a class $\delta^*(c)$, where $\delta^*:H^0(M\backslash N; \Z/2) \rightarrow H^1(M,M\backslash N;\Z/2)$, so that $\lambda + \delta^*(c)$ is in the image of the Thom class. This class $c$ changes the orientation that was initially fixed on $M\backslash N$ and therefore means the orientation that is modified by $c$ does not extend across $N$.
\end{proof}
We now state a construction due to Pontryagin and Thom, which allows one to parametrize Poincaré dual submanifolds to a cohomology class by studying maps into a Thom space. It serves as the cornerstone for the two main applications in this paper.

\begin{thm}[Pontryagin--Thom]\label{construction:PT}
Let $M$ be a closed manifold.
\begin{enumerate}
    \item Let $A$ be a finite abelian group and  $\omega \in H^n(M;A)$. There exists an oriented submanifold $F$ that is Poincaré dual to $\omega$ if and only if the map $\omega\colon M \rightarrow K(A,n)$ lifts across the Thom class map $U\colon M\SO_n\to K(A, n)$ to a map $M\rightarrow M\SO_n$.
    \item Let $\omega \in H^n(M;\Z/2)$. There exists a (not necessarily oriented) submanifold that is Poincaré dual to $\omega$ if and only if the map $\omega\colon M \rightarrow K(\Z/2,n)$ lifts across the Thom class map $U\colon M\O_n\to K(\Z/2, n)$ to a map $M\rightarrow M\O_n$.
\end{enumerate}
\end{thm}
This follows from the Pontryagin--Thom theorem, as explained by Kirby--Taylor~\cite[\S 2]{KT90}.

%%%%%%%%%%%%%%%%%%%%%%%%%%%%%%%%%%%%%%%%%%%%%%%
\subsubsection{Characteristic bordism}\label{subsubsection:structureextensions}
%%%%%%%%%%%%%%%%%%%%%%%%%%%%%%%%%%%%%%%%%%%%%%%
The ideas of characteristic bordism originated with the work of Freedman--Kirby (FK) and Kirby--Taylor (KT) in \cite{FK,KT90}, where the authors defined a bordism group generated by a pair consisting of a 4-dimensional manifold $M$ and a 2-dimensional submanifold $F$ that is Poincaré dual to  either the characteristic classes $w_2(TM)$ or $w_2(TM)+w_1(TM)^2$. %\vic{$w_1(TM)$?}
\begin{defn}
\label{def:charbordism}
%\arun{In this definition, what is $\mathcal P$?}
Choose an abelian group $A$ and tangential structure $\xi\colon B\to B\O$, where either $A = \Z/2$, or $\xi$ factors through $B\SO$. Let $M$ be a manifold with $\xi$-tangential structure, and let $\mathcal P\in H^n(B; A)$. A $(\xi, \mathcal P)$-\textit{characteristic pair} consists of a pair $(M,F)$ where $F$ is a proper submanifold of $M$ Poincaré dual to $\mathcal{P}(M)$ and the boundary of $M$ intersects $F$ precisely and transversely at the boundary of $F$. 
\end{defn}

\begin{rem}
The requirement on a pair $(M, F)$ to have a characteristic structure means that a $(\xi, \mathcal P)$-characteristic structure on $(M, F)$ induces one on $(\partial M, \partial F)$. Thus there are notions of bordisms of $(\xi, \mathcal P)$-characteristic manifolds, bordism groups, etc.
\end{rem}

%Just as in \cref{construction:PT}, the manifold $F$ can be oriented or unoriented depending on the group $A$. It noteworthy to mention that this definition does not require making an explicit choice of such a tangential structure on $F$. Specifying such a choice will lead to different notions of characteristic bordism, and we will further remark on this in \S\ref{section:Defects}. 
If $F$ is Poincaré dual to $\mathcal{P}$ then
$M \backslash F$ has a $\psi$-tangential structure, where $\psi$ is the tangential structure which is a $\xi$-structure and a trivialization of $\mathcal P$. For example, if $\xi$ is an orientation and $\cP = w_2(TM)$ then $\psi$ is a spin structure. All of the characteristic structures we will discuss in \S\ref{section:Defects} will be \textit{characterized}, meaning that the $\psi$-tangential structure on $M \backslash F$ does not extend across $F$. Equivalently, it means that there is no $\psi$-structure on $M$ that restricts to the $\psi$-structure on $M\backslash F$. There are immediate physical consequences that can be gleaned as a consequence of this fact. Most conspicuously, a theory when placed on $M\backslash F$ with the data of a choice of $\psi$-structure, cannot be compatibly extended across $F$.

%\begin{defn}
%    The \textit{characteristic bordism group} is the group of  $\mathcal{P}$-characteristic pairs $(M,F)$ where two pairs $(M,F)$ and $(M',F')$ are $\mathcal{P}$-characteristically equivalent if there exists a $\mathcal{P}$-characteristic pair $(L,E)$ where $L$ is a $(k+1)$-dimensional manifold with the same tangential structure as $M$ and $E$ is a submanifold of $L$ that is Poincaré dual to $\mathcal{P}$, such that $\partial (L,E) = (M,F)\sqcup -(M',F')$.
%\end{defn}

The concept of characteristic pairs has received limited attention in the literature; however, a few notable examples do exist, which we will examine in this work.

\begin{defn}[\cite{FK}]\label{def:FK}
    A \term{Freedman--Kirby (FK) characteristic pair} is a $w_2(TM)$-characteristic pair $(M, F)$ where $M$ and $F$ are closed, oriented manifolds.
\end{defn}
\begin{defn}\label{defn:FKO}
A \term{$\mathrm{FK}^\O$ characteristic pair} is a $w_2(TM)$-characteristic pair $(M, F)$ where $F$ is oriented, but $M$ is not necessarily oriented.
\end{defn}

\begin{defn}[\cite{GM}]\label{def:GM}
A \term{Guillou--Marin (GM) characteristic pair} is a $w_2(TM)$-characteristic pair $(M, F)$ in which $M$ is oriented, but $F$ is not necessarily oriented.
\end{defn}

\begin{defn}[{\cite[\S 6]{KT90}}]\label{def:KT}
A \term{Kirby--Taylor minus (KT$^-$)} resp.\ \term{Kirby--Taylor plus (KT$^+$) characteristic pair} is a $w_2(TM)+w_1(TM)^2$ resp.\ $w_2(TM)$ characteristic pair, with no orientation data specified on $M$ or $F$.
\end{defn}
This is slightly different from the data Kirby--Taylor use: see \cref{rem:differentfromKT}. Kirby--Taylor also consider a variant of \cref{def:FK,def:GM,def:KT} for topological manifolds (\textit{ibid.}, \S 9), and Klug~\cite{Klu21} introduced variants of these structures for manifolds with boundary.

\begin{notation*}
    We will refer to the characteristic bordism groups for the characteristic pairs defined above by the authors' names.  e.g. $\Omega^{\mathrm{GM}}_k$ is the group generated by characteristic pairs $(M,F)$ of GM-type where $M$ is in dimension $k$.
\end{notation*}

For the manifolds $M$, $M\backslash F$, and $F$, we will follow a line of reasoning where we do not specify a choice of tangential structure for $F$ when we conduct the computations of characteristic bordism. We will see how to apply this specific notion of characteristic pairs in \S\ref{section:Defects}.

%%%%%%%%%%%%%%%%%%%%%%%%%%%%%%%%%%%%%%%%%%%%%%%%%%%%%%%
\subsection{Long exact sequences and approximate symmetries}\label{subsection:smithLESintro}
%%%%%%%%%%%%%%%%%%%%%%%%%%%%%%%%%%%%%%%%%%%%%%%%%%%%%%%

%\arun{What is $\mathcal P$? What is $\xi'$? I think we should introduce notation for a general characteristic structure, so that we can refer to it unambiguously here}

In this subsection, we conjecture the existence of a
\textit{characteristic long exact sequence}, which forms a long exact sequence in the context of characteristic bordism groups.  Such a long exact sequence captures more than just symmetries, but rather the topological aspects of the symmetry operator when supported along a submanifold. It can also be used to compute obstructions to symmetry breaking; see~\cite{Debray:2023ior} for applications of a related ``symmetry breaking long exact sequence'' in QFT.

\begin{conj}
\label{char_conj}
Given $(\xi, \mathcal P)$ as in \cref{def:charbordism}, and $\xi'$ a tangential structure for a submanifold $F$, there is a map of spectra $\mathcal R\colon\mathit{MTChar}(\xi,\mathcal P)\to \Sigma^n \mathit{MT\xi}'$ such that the map $\mathcal R$ induces on $\pi_k$ is the
map $R\colon \Omega_k^{(\xi,\mathcal P)}\to\Omega_{k-n}^{\xi'}$ sending a characteristic pair $(M, F)\mapsto F$.
\end{conj}
Here $\mathit{MTChar}(\xi, \mathcal P)$ is the Thom spectrum for tangential $(\xi, \mathcal P)$-characteristic structures, as in \Cref{def:sign_charbord}.
%\TODO: remark: module structure. \TODO: reader does not know what MT is yet. \matt{a remark along these lines is made in footnote 1} \vic{I think that we can just cite Definition 2.11}

There are several pieces of heuristic evidence for \cref{char_conj}.
\begin{enumerate}
    \item In general, natural, geometrically defined homomorphisms between bordism groups typically lift to maps of spectra. Examples include the Smith homomorphism (see~\cite[Proposition 3.17]{Debray:2024wxm}) and including bordism of closed manifolds into a notion of bordism of compact manifolds, possibly with boundary (see~\cite[\S 5.1]{Debray:2024wxm}).\footnote{See also the recent work of Hoekzema-Merling-Murray-Rovi-Semikina~\cite{HMMRS22}, Hoekzema-Rovi-Semikina~\cite{HRS22}, Campbell-Kuijper-Merling-Zakharevich~\cite{CKMZ23}, and Merling-Ng-Semikina-Sendón Blanco-Williams~\cite{MNSBW25} on lifting invariants of SK-groups, a variant of bordism groups due to Karras-Kreck-Neumann-Ossa~\cite{KKNO73}, to maps of spectra.} We will present some details on  the Smith homomorphism in the latter part of this section.
    \item \Cref{char_conj} for the characteristic structures in \cref{def:FK,,def:GM,,def:KT} is implicit in Kirby--Taylor~\cite[Corollary 6.12, Remarks 6.15 and 6.16]{KT90} (though see \cref{KT_correction}).
    \item If we assume in addition that the normal bundle of $F\hookrightarrow M$ extends to a vector bundle on $M$, then the analogue of \cref{char_conj} is true, as we describe below, and in some examples including Freedman--Kirby characteristic bordism this assumption always holds.
\end{enumerate}
Later in this paper, we will assume \cref{char_conj} and derive consequences we expect to see in symmetry breaking in field theories with defects. These applications mostly make use of a single key consequence of \cref{char_conj}: if the map $\mathcal R$ of spectra exists, then the maps $R$ on bordism participate in a long exact sequence:
\begin{equation}\label{eq:schematic}
    \dotsb \longrightarrow \Omega_k^? \overset{}{\longrightarrow} \Omega_k^{\mathcal{P}}
	\overset{R}{\longrightarrow} \Omega_{k-n}^{\xi'}\longrightarrow
	\Omega_{k-1}^?\longrightarrow \dotsb\,.
\end{equation}
Heuristically, the idea is that lifting $R$ to a map of spectra gives a point-set model for it; given a map $A\to X$ of spaces, one can construct its cofiber $X/A$ and a long exact sequence on generalized homology. Likewise, given a point-set model of $R$, we can construct its cofiber to obtain the long exact sequence~\eqref{eq:schematic}.

While such a characteristic long exact sequence is expected to exist, explicitly describing it in a way that is computationally effective is challenging. However, we can approximate the characteristic long exact sequence using the Smith long exact sequence.
In particular, a lift of $B\SO \xrightarrow{\omega} K(\Z/2,n) $ to $B\O_{d-n}$ is less restrictive than lifting to $M\O_{d-n}$. The normal bundle to $F$ provides a map to $B\O_{d-n}$ everywhere locally on $F$, but to have this bundle be globally well defined means giving $F$ the structure of a Poincaré dual to $\omega$. 
We will explicitly show how good of an approximation the Smith long exact sequence is in the examples of \S\ref{subsection:GMcharBord} and \S\ref{subsection:KTcharBord}. As is often the case in physical situations, one might only view the defect from a ``local perspective”, by which we mean the physical observables and symmetries are only considered in  a neighborhood of the defect. It is then reasonable to expect symmetries to emerge locally around the defect, and this is indeed the case as can be seen by the normal directions of $F$ in $M$. However, this symmetry may not persist globally if the lift to $M\O_{d-n}$ is obstructed.
In this vein, one can view the Smith long exact sequence as a special instance of the characteristic long exact sequence when the global structure of the defect is not in question.

We now give a brief overview of the Smith long exact sequence, which has a physical interpretation that arises when we have a global symmetry and we consider its spontaneous breaking pattern.
Spontaneously breaking a $p$-form symmetry given by an abelian group $E$ results in a $(p+1)$-codimensional domain wall \cite{Lake:2018dqm,Brennan:2023mmt}. We will take all of our theories to live on closed ambient manifolds $M$, so that the domain wall has support on a submanifold $F$ which is Poincaré dual to a class in $H^{p+1}(M, E)$. In the case of 0-form symmetries, the Smith long exact sequence \cite{Debray:2023ior,Debray:2024wxm} is a powerful tool that maps the anomaly from a domain wall that emerges from spontaneous symmetry breaking, to the anomaly of the bulk theory with certain symmetry. In the first examples in the literature, the authors in \cite{Hason:2020yqf,COSY20} %\arun{We should also cite Córdova-Ohmori-Shao-Yan} 
considered a spontaneously broken unitary $\Z/2$-symmetry in $(d+1)$-dimensions, which led to a non-unitary $\Z/2$-symmetry on the codimension 1 domain wall upon taking into account CPT symmetry. 
The generalization of this starts with a Thom space $\mathrm{Th}(X,V)$ and a symmetry breaking order parameter which transforms in some representation $\rho$ of $X$; in particular, the data of the order parameter can be presented as an associated bundle $W$ to $\rho$. 

Let $V,W\rightarrow X$ be two vector bundles of rank $r_V$ and $r_W$ over a space $X$, and let $p\colon S(W)\to X$ be the sphere bundle of $W$.% We have the following lemma for the case where $X$ is a 1-type:
\begin{lem}[{\cite[Lemma 3.5]{Debray:2024wxm}}]
    Let $M$ be a closed $k$-dimensional manifold with $(X,V)$-twisted $\xi$-structure and let $i:N \hookrightarrow M$ be a closed $(k-r_W)$-dimensional submanifold of $M$ such that the mod 2 fundamental class $i_*(N)\in H_{k-r_W}(M;\Z/2)$ is Poincaré dual to the mod 2 Euler class $e(W)$. The $(X,V)$-twisted $\xi$-structure on $M$ induces a  $(X,V\oplus W)$-twisted $\xi$-structure on $N$.
\end{lem}
\begin{proof}
    Let  $TM|_N$ be the restriction of the tangent bundle of $M$ onto $N$, and let $\nu$ be the normal bundle to $N$, so that $TM|_N\cong TN \oplus \nu$. Since the fundamental class of $N$ is Poincaré dual to $e(W)$, then $\nu$ is isomorphic to $W$. If $TM\oplus V$ has a $\xi$-structure, then $TN\oplus V \oplus W$ has a $\xi$-structure and thus $N$ has a $(X,V\oplus W)$-twisted $\xi$-structure.
\end{proof}
Let $\Omega_\xi^*$ denote the generalized cohomology theory whose corresponding generalized homology theory is $\xi$-bordism $\Omega_*^\xi$. $\Omega_\xi^*$ is sometimes called ``$\xi$-cobordism,'' e.g.\ when it was first constructed by Atiyah~\cite{Ati61}. This is not the same theory as the Anderson or Pontryagin duals of $\xi$-bordism.

For any tangential structure $\xi$ and any rank-$r$ vector bundle $V\to X$, there is a twisted Euler class $e^\xi(V)\in \Omega_\xi^r(X^{V-r})$, defined to be the image of the universal twisted Euler class $e^\mathrm{fr}(V)$ in twisted framed cobordism, constructed by Becker~\cite{Bec70}, under the map forgetting from a stable framing to a $\xi$-structure.

The Smith map is defined to be the cap product with $e^\xi(V)$:\footnote{Strictly speaking, in order for this cap product to be defined, we need additional structure: that $\xi$-cobordism is a \emph{multiplicative} generalized cohomology theory. This is true for $\xi = \Spin$, but \emph{not} for $\xi = \Pin^\pm$. If $\xi$ lacks multiplicativity, one may instead use the fact that $\xi$-bordism is a module over stably framed cobordism, which is multiplicative, and so take the cap product with $e^{\mathrm{fr}}(V)$.}
%\arun{At this point, is it clear to the reader that these are twisted $\xi$-bordism groups? Also, does the reader know why we have these $-r_V$ and $-r-W$ terms?}
\begin{equation}
    S_W  = e^\xi(V)\frown\colon \Omega^{\xi}_k(X^{V-r_V}) \rightarrow \Omega^{\xi}_{k-r_W}(X^{V\oplus W-r_V-r_W}).
\end{equation}
%
%which arises from capping with the $\xi$-cobordism Euler class $e^{\xi}(W)\in \Omega^{r_W}_\xi(X^{W-r_W})$ \arun{I think we should explain briefly what the deal is here, like we did in the anomaly indicators paper, since this is a pretty confusing aspect of the Smith hom. We should also cite the OG papers defining twited gen coh Euler classes. I can take care of this.}. This is defined to be an element in degree $r_W$ $(-W)$-twisted $\xi$-cohomology of $X$.
From the Smith map, we get a long exact sequence, which by \cref{shearcor} consists of certain twisted $\xi$-bordism groups:
\begin{equation}\label{eq:smithLES}
	\dotsb \longrightarrow \Omega_k^\xi(S(W)^{p^*V-r'}) \overset{p_*}{\longrightarrow} \Omega_k^\xi(X^{V-r'})
	\overset{S_W}{\longrightarrow} \Omega_{k-r}^\xi(X^{V\oplus W -r-r'})\longrightarrow
	\Omega_{k-1}^\xi(S(W)^{p^*V-r'})\longrightarrow \dotsb
\end{equation}
where we should think of the group $\Omega_{k-r}^\xi(X^{V\oplus W -r-r'})$ as generated by the submanifold which supports the defect that arises when the order parameter gets an expectation value. One way to see the connecting map is to realize this long exact sequence as arising from taking the colimit over $X$ of the cofiber sequence of pointed spaces:
\begin{equation}
        S(W)_{+} \wedge S^V\rightarrow D(W)_{+}\wedge S^V\rightarrow S^W \wedge S^V\,.
    \end{equation}
It is well-known in homotopy theory that the above cofiber sequence of spaces results in the cofiber sequence of spectra
\begin{equation}\label{eq:cofiberseq}
    S_X(W)^V \rightarrow X^V \rightarrow X^{W\oplus V}\,,
\end{equation}
and~\cite[Theorem 5.1]{Debray:2024wxm} shows that, by applying $\pi_*$ to this cofiber sequence, we obtain the long exact sequence~\eqref{eq:smithLES}.

Let $\mho_\xi^k$ denote the generalized cohomology theory obtained by applying \term{Anderson duality} to $\Omega_{k+1}^\xi$~\cite{And69, Yos75}; then, by work of Freed--Hopkins~\cite{FH21InvertibleFT} and Grady~\cite{Gra23}, $\mho_\xi^k(X)$ is naturally isomorphic to the abelian group of deformation classes of $k$-dimensional reflection-positive invertible field theories on manifolds with a $\xi$-structure and a map to $X$. As these classify anomalies of field theories in one dimension lower, we will apply Anderson duality to the characteristic long exact sequence to obtain information on anomalies:
%We take the Anderson dual of the Smith long exact sequence in bordism and arrive at the corresponding long exact sequence of invertible field theories:
\begin{equation}
\begin{aligned}
    \ldots \longleftarrow \mho^k_\xi(S(W)^{p^*V-r'})
    &\longleftarrow \mho^k_\xi(X^{V-r'}) \xlongleftarrow{D} \mho^{k-r}_\xi(X^{V\oplus W -r-r'})\\
    &\xlongleftarrow{R} \mho^{k-1}_{\xi}(S(W)^{p^*V-r'})\longleftarrow \dotsb\,.
\end{aligned}
\end{equation}
This dualized sequence then has immediate applications in the context of invertible field theories and anomaly inflow.
This group $\mho^k_\xi(X^V)$ gives the deformation classes of invertible field theories in dimension $k$, where the symmetry and its twisting of the $\xi$-structure is encapsulated by $X^V$. There is a defect anomaly map  $D:\mho_{\xi}^{k-r}(X^{V\oplus W-r-r'}) \to \mho_{\xi}^{k}(X^{V-r'})$ which describes the anomaly of the bulk given the anomaly on the defect which is associated with spontaneously breaking a symmetry. The map $R: \mho^k_\xi(X^{V-r'})\rightarrow \mho^k_\xi(S(W)^{p^*V-r'})$ measures the obstruction to gapping the $(k-1)$-dimensional theory over $S(W)^{p^*V-r'}$.

%%%%%%%%%%%%%%%%%%%%%%%%%%%%%%%%%%%%%
%\section{Warmup: Codimension 1-operators}
\label{section:codim1}
%%%%%%%%%%%%%%%%%%%%%%%%%%%%%%%%%%%%%
%Before diving into the more interesting scenario with codimensional-2 defects, let us first have a look at the simplest scenario with codimensional-1 defects. It turns out that the story aligns well with our general understanding of smith homomorphism for defects. 

%First consider the characteristic structure corresponding to codimensional-1 defects. We can call this the time-reversal characteristic bordism \vic{?}. It contains a paire $(M, F)$, where $F$ is Poincaré dual to $w_1(TM)\in H^1(M, \Z/2)$. We have the smith homomorphism
%\begin{equation}
%    some simth homo
%\end{equation}

%The characteristic pairs generating $\Omega_*^()$ in each degree are
%\begin{itemize}
%    \item 
%\end{itemize}

%%%%%%%%%%%%%%%%%%%%%%%%%%%%%%%%%%%%%
\section{Defects and Characteristic Structures}
\label{section:Defects}
%%%%%%%%%%%%%%%%%%%%%%%%%%%%%%%%%%%%%
We consider defects in a 4-dimensional setting, i.e.\ (3+1)-spacetime dimensions, and give a general prescription of what tangential structures one can expect for defects inserted along the Poincaré dual to a degree $2$ $\Z/2$-valued cohomology class. This 
conveniently integrates in examples studied by \cite{FK,GM,KT90}.

We will focus on reframing the examples of characteristic pairs introduced in \S\ref{subsection:duals} into the language that only requires a twisted tangential structure given in Definition \ref{def:twisted}. Establishing a dictionary between the two formulations of characteristic structures enables us to translate a topological description into a computationally accessible form, but also gives insights into the underlying physics—such as anomalies that emerge from the global properties of these defects. More specifically we will consider the examples tabulated in Table \ref{tab:Charbordism}.

\begin{table}[htb]
\begin{center}
\begin{tabular}{ c c c c c} 
\toprule
Name & $\xi$ & $\mc{P}$ & $F$ & Ref.\\
\midrule
Freedman--Kirby (FK) & $\SO$ & $w_2(TM)$ &  Oriented & \S\ref{subsection:FKcharbordism} \\
%\hline
Freedman--Kirby$^\mathrm{O}$(FK$^{\mathrm{O}
}$)
& $\O$ &  $w_2(TM)$& Oriented & \S\ref{subsubsection:pincKT} \\ %\hline 
Guillou--Marin (GM)
& $\SO$ &  $w_2(TM)$& Not Oriented & \S\ref{subsection:GMcharBord}\\
%\hline
Kirby--Taylor$^-$ (KT$^{-}$) & $\O$ & $w_2(TM)+w_1(TM)^2$ & Not Oriented & \S\ref{subsection:KTcharBord}
\\ 
%\hline
Kirby--Taylor$^+$ (KT$^{+}$)
& $\O$ &  $w_2(TM)$ & Not Oriented & \S\ref{subsection:pin+KT}\\
\bottomrule
\end{tabular}
\end{center}
\caption{This table summarizes the choices of structure that we study in this paper for $M$ and $F$, as well as the cohomology class that $F$ is Poincaré dual to, following \cref{def:charbordism}. These examples follow closely those studied in \cite{KT90}.}
\label{tab:Charbordism}
\end{table}

The main mathematical content of this section can be summarized as follows: given the structures on $M$ and $F$, and the cohomology class $F$ is dual to, we calculate the characteristic bordism groups and the approximation of the characteristic long exact sequence by the Smith long exact sequence.

%%%%%%%%%%%%%%%%%%%%%%%%%%%%%%%%%%%%%%%%%%%%%%%%%%
\subsection{Freedman--Kirby Characteristic Bordism}\label{subsection:FKcharbordism}
%%%%%%%%%%%%%%%%%%%%%%%%%%%%%%%%%%%%%%%%%%%%%%%%%

In the case of FK characteristic bordism given in Definition \ref{def:FK}, Freedman--Kirby~\cite{FK} considered $4$-dimensional closed, oriented manifolds $M$ equipped with an oriented submanifold dual to $w_2(TM)$.
An important aspect of FK characteristic pairs that already aligns with our goals of this paper is that FK structures are equivalent to a commonly studied twisted spin structure called a \term{\spinc structure}.
%it reflects a commonly used twisted spin structure called spin$^c$. 

\begin{defn}\label{defn:spinc}
   Let $V\rightarrow X$ be an oriented vector bundle. A \textit{spin$^c$ structure} on $V$ is data of a complex line bundle $L$ with $w_2(V)= w_2(L)$.
\end{defn}
%\arun{Since the reader already knows what twisted spin structures are, we could just define \spinc structures that way. Likewise with \pinc structures in the next subsubsection}

\noindent $L$ is referred to as the {determinant line bundle} of the spin$^c$ structure. An equivalent definition for a spin$^c$ structure relies on a class $c_1\in H^2(X;\Z)$ and an identification of $c_1\bmod 2 = w_2(V)$, and we will see an analogue of this definition used in the next section. We will see in \cref{prop:FKapprox} that (1) \cref{char_conj} is true for FK characteristic structures, and (2) the induced characteristic long exact sequence, which appears in Kirby--Taylor~\cite[Corollary 6.12]{KT90}, is isomorphic to the \spinc Smith long exact sequence constructed in~\cite[Example 7.26]{Debray:2024wxm}.
%find that the characteristic long exact sequence for FK structures is completely determined by the Smith long exact sequence, i.e. the approximation is exact in this case. 

In light of the Pontryagin--Thom construction, the obstruction to finding a characteristic pair $(M,F)$ is given as an obstruction to the existence of a map $f$ in 

\begin{equation}
    \begin{tikzcd}
        && M\SO_2 \arrow[d]\\
        B\SO \arrow[rr,"w_2",swap ] \arrow[urr,dotted,"f"] & &K(\Z/2,2)\,.
    \end{tikzcd}
\end{equation}
In this case we can make a simplification:
\begin{thm}[Smith isomorphism]
\label{cpx_smith_isom}
There is a homotopy equivalence $M\SO_2\simeq B\SO_2$.
\end{thm}
This is a standard fact in homotopy theory; see~\cite[Example 2.1]{Ada74} for a reference. \Cref{cpx_smith_isom} is the infinite-dimensional version of the fact that the Thom space of the tautological bundle over $\CP^n$ is homotopy equivalent to $\CP^{n+1}$.

One helpful consequence of \cref{cpx_smith_isom} is that a map to $M\SO_2$ is the data of a rank 2-vector bundle $V$. In particular, the approximation to the long exact sequence in characteristic bordism by the Smith long exact sequence, discussed in \S\ref{subsection:smithLESintro}, is an isomorphism in this case.

\begin{rem}
   The cases of $M\O_1$, $M\SO_2 = M\mathrm U_1$, and 
   $M\mathrm{Sp}_1$ are special in that they are the only Thom spectra that lead to Smith isomorphisms in the sense of \cref{cpx_smith_isom}. We will see in the later sections how the isomorphism fails for different Thom spaces.
\end{rem}

The pullback square for the space of Freedman--Kirby pairs is therefore given by 
\begin{equation}\label{eq:FKpullback}
    \begin{tikzcd}[column sep=2cm, row sep=2cm]
 B\mathrm{FK} \arrow[r,"L"] \arrow[d,"V"] \arrow[dr, phantom, "\lrcorner", very near start] & M\SO_2 \arrow[d, "w_2(L)"] \\
B\SO \arrow[r, "w_2(V)"] & K(\Z/2,2)
\end{tikzcd}
\end{equation}
which implies that the data of a FK pair is the data of a rank 2 oriented vector bundle $L$, a map to $B\SO$ given by $V$ and an identification $w_2(V) = w_2(L)$. This agrees with Definition \ref{defn:spinc}: rank $2$ oriented vector bundles are equivalent data to complex line bundles thanks to the isomorphism of the respective structure groups $\SO_2\cong\mathrm U_1$.% if we interpret $L$ as a complex line bundle. 

\begin{prop}[{\cite[Remark 6.14]{KT90}}]\label{prop:FKpairs}
    The characteristic structure of Freedman--Kirby is equivalent to a spin$^c$ structure.
\end{prop}
%\arun{The ideas in the proof are great but would it be OK if I reworded it some? There's a few small things, e.g.\ $\nu_F$ is not the tautological bundle over $B\SO_2$: instead, it's the pullback of the tautological bundle to $F$. Explaining that this is about an $\SO_2$ symmetry is kind of an unusual way to phrase it.}
\begin{proof}
By definition, a Freedman--Kirby characteristic structure on an oriented vector bundle $E\to M$ is data of a map $f\colon M\to M\SO_2$ and an identification of $w_2(E)$ and $f^*(U)$ in $H^2(M;\Z/2)$ -- i.e., by \cref{nonVB}, a $(M\SO_2, 0, U)$-twisted spin structure. The Smith isomorphism $B\SO_2\to M\SO_2$ pulls $U$ back to $w_2$, so a Freedman--Kirby characteristic structure is equivalent to a $(B\SO_2, w_2)$-twisted spin structure, which is a \spinc structure.
%   Let $\nu_F$ be the normal bundle to $F$. Due to the Pontryagin--Thom construction $\nu_{F}$ is identified with the tautological bundle $V\rightarrow B\SO_2$, as the normal directions to $F$ have a $\SO_2$-symmetry. As $M\SO_2 \cong B\SO_2$, this normal bundle is enough data to specify $F$ in the FK characteristic pair.   Under the map $i:F\hookrightarrow M$ we have $i^*(w_2(TM))=w_2(\nu_F)$, as $F$ is Poincaré dual to $w_2(TM)$ by Corollary \ref{cor:atoN}. We see that the data of $M$ along with the normal bundle $\nu_F$ provided by $F$ gives a spin$^c$ structure, and thus the characteristic pair $(M,F)$ of type Freedman--Kirby  is equivalent to spin$^c$ structure.
\end{proof}
\begin{prop}\label{prop:FKapprox}
\hfill
\begin{enumerate}
    \item \Cref{char_conj} holds for Freedman--Kirby characteristic bordism: there is a map of spectra $\mathcal R\colon \mathit{MTFK}\to\Sigma^2 \MTSpin\wedge (B\SO_2)_+$ which on homotopy groups is the characteristic map $(M,F)\mapsto F$.
    \item The induced characteristic long exact sequence~\eqref{eq:schematic}, which has the form
    \begin{equation}\label{eq:FKLES}
        \ldots \longrightarrow \Omega^{\Spin}_k \xlongrightarrow{i} \Omega^{\mathrm{FK}}_k \xlongrightarrow{R} \Omega^{\Spin}_{k-2}(B\SO_2) \xlongrightarrow{a} \Omega^{\Spin}_{k-1}\longrightarrow \ldots\,,
    \end{equation}
    is isomorphic to the Smith long exact sequence associated to the Smith homomorphism $\mathrm{sm}_L\colon\MTSpin^c\to\Sigma^2\MTSpin\wedge (B\SO_2)_+$, where $L\to B\SO_2$ is the tautological complex line bundle.
\end{enumerate}
\end{prop}
Kirby--Taylor~\cite[Corollary 6.12]{KT90} construct \eqref{eq:FKLES}.
\begin{proof}
It suffices to show that the identification of FK-structures and \spinc structures \cref{prop:FKpairs} intertwines the characteristic map sending a pair $(M,F)\mapsto F$ and the Smith homomorphism sending a \spinc manifold $M$ to a smooth representative of the Poincaré dual $N$ of the Euler class of the determinant line bundle $L_M\to M$. This is because, once we have established this, we can define $\mathcal R =\mathrm{sm}_L$; thus the fibers of these maps are also identified, giving isomorphic long exact sequences of homotopy groups.

So we want to show that the characteristic map $(M,F)\mapsto F$ equals the Smith homomorphism $M\mapsto N$, where $N$ is a Poincaré dual of $e(L_M)$ (and we may use any such $N$). Since $M$ is \spinc, $w_2(L_M) = w_2(M)$, and $e(L_M)\bmod 2 = w_2(L_M)$, so we may choose $N = F$ to show that $\mathrm{sm}_L$ is also the characteristic map.
\end{proof}
 %   The Smith long exact sequence for spin$^c$ is isomorphic to the characteristic long exact sequence of FK pairs.
%\end{prop}
%The FK characteristic bordism in degree $k$, which we will denote as $\Omega^{\mathrm{FK}}_k$, fits into the following long exact sequence given in \cite[Corollary 6.12]{KT90}:

By its relation to spin$^c$, the sequence for $\mathrm{FK}$ characteristic bordism is the Smith long exact sequence for spin$^c$ bordism, as described in~\cite[Example 7.26]{Debray:2024wxm}. However the two sequences have different interpretations for their entries. We first give the interpretation coming from the Smith
long exact sequence, which puts the symmetries at the forefront. In particular, it is most applicable in fermionic systems where the fermionic particles are charged under a $\mathrm{U}_1$ symmetry.
\begin{enumerate}
    \item An element in the group $\Omega^{\Spin}_k$ is a choice of $k$-dimensional manifold $M$ with spin structure.
    \item Since FK and \spinc structures are equivalent by \cref{prop:FKpairs}, an element in $\Omega^{\mathrm{FK}}_k$ is a manifold $M$ with a principal $\SO_2$-bundle. In particular, the bundle is part of the data of a $(B\SO_2,L)$-twisted spin structure for $M$. The map $i$ sends a spin manifold to the \spinc structure with the trivial $\SO_2$-bundle.
    \item An element in $\Omega^{\Spin}_{k-2}(B\SO_2)$ is a spin manifold with a principal $\SO_2$-bundle, where the spin structure is not twisted. The map $R$ represents taking the Poincaré dual with respect to the Euler class of $L$.
    \item The final map labeled $a$ takes the sphere bundle of the bundle $L$.
\end{enumerate}
The second interpretation is more topological, and comes from unpacking how the sequence arises from the characteristic pairs:
\begin{enumerate}
    \item An element in the group $\Omega^{\Spin}_k$ is a choice of $k$-dimensional manifold $M$ with spin structure. This has the same interpretation as in the Smith long exact sequence. 
    \item An element of $\Omega^{\mathrm{FK}}_k$ is now represented by a pair $(M,F)$ where a submanifold $F\subset M$ which is Poincaré dual to $w_2(TM)$, and a spin structure on $M\backslash F$ that is different from the spin structure on $M$.\footnote{We can think of this last condition as the defect having some nontrivial action on the spin structure of $M$ upon its insertion.}  The map $i$ sends a spin manifold to the characteristic structure where we take $F = \varnothing$. 
    %Starting from the spin$^c$ interpretation, we can take the zero set of a section of the $B\SO_2$-bundle. This will give rise to the submanifold $F$. \matt{describe also the map up}
    \item The map $R$ restricts $(M,F)\rightarrow F$, with the normal bundle $\nu_F\rightarrow F$ remembering how $F$ was embedded in $M$. Along $F$ the tangent bundle of $M$ decomposes as $TF \oplus \nu_F$. Using the fact that the restriction of $w_2(TM)$ to $F$ is equal to $w_2(\nu_F)$, this means $w_2(TF)=0$ and $F$ is spin. The $\SO_2$-bundle arises from the symmetry in the normal direction, and $F$ therefore represents an element in $\Omega^{\Spin}_{k-2}(B\SO_2)$. 
    %While $F$ has a spin structure, the fact that it is Poincaré dual to $w_2(TM)$ means that the spin structure on $M\backslash F$ does not extend across $F$. 
    \item The map $a$ can be interpreted as taking the boundary of a tubular neighborhood of $F$. This is equivalent to the interpretation of $a$ in the Smith long exact sequence.
\end{enumerate}

We now elaborate on the physical context in which characteristic pairs can appear.
The setup begins with a four-manifold $M$ with a chosen spin structure that can compatibly host a QFT $\mathcal{T}$ that depends on a spin structure. We consider inserting a defect with the global property that it wraps a codimension 2 submanifold $F$, which is Poincaré dual to $w_2(TM)$. For a chosen spin structure on $M\backslash F$, this structure does not extend across $F$ by Proposition \ref{prop:notextending}. In this sense the manifold $F$ acts as a hard codimension 2 boundary for $M$, and the theory $\cT$ which is defined on $M\backslash F$ is not defined on $F$. The theory $\mathcal{T}$ does however have an enhanced  $\mathrm{U}_1$ global symmetry in this case thanks to the defect $F$. This is a nontrivial fact, relying critically on the equivalence $M\SO_2 \cong B\SO_2$ to ensure that the symmetry is not merely local to the defect but extends consistently throughout the bulk. For a concrete lattice model that implements a Freedman-Kirby structure to construct an invertible TQFT, see \cite{Kobayashi:2026hij}.

We emphasize that, as stated in \S\ref{subsubsection:structureextensions}, we do not specify a choice of tangential structure for $F$, and this is in line with keeping $F$ a defect that is inserted externally.
 A priori, specifying choices of tangential structure on different subsets of $M$, $M\backslash F$, and $F$ can potentially lead to different results for the characteristic bordism groups. 
 See also Remark \ref{rem:differentfromKT} for additional details on how specifying a choice of a given tangential structure on $F$ will lead to manifestly different results in computations.
 
 Taking the Anderson dual of the sequence in \eqref{eq:FKLES}, as described in \S\ref{subsection:smithLESintro}, gives an anomaly interpretation of characteristic bordism. Due to the situation being exactly a Smith long exact sequence, the interpretation can be given by the same analysis as in \cite{COSY20,Hason:2020yqf,Debray:2023ior}. It would in general make sense to take the Anderson dual of such a sequence given in \eqref{eq:schematic}, and we will compute examples of the groups $\Omega^{\mathcal{P}}_k$ in the subsequent subsections. Exactness then lends some leverage in computing what the groups $\Omega^?_k$ are. In Theorem \ref{GM_LES_thm} we compute these groups for the case of GM characteristic bordism, and in \cref{KT-_LES_thm} and \cref{KT+_LES_thm} for KT$^-$ and KT$^+$ respectively. Additionally the Anderson dual of $\Omega^?_k$ should in principle include the obstruction to symmetry breaking, however the full interpretation of how the obstruction arises is unclear.

\begin{comment}
\begin{rem}
    In situations that concern theories of quantum gravity, the theory itself must have a boundary to be defined. This because the only observables that one should be able to
measure are due to scattering processes, and this requires a notion of a boundary to be precise. In this spirit, boundary conditions for the metric at infinity are usually taken. Taking a closed manifold and cutting
out a submanifold is also a way of introducing a boundary, and the left over manifold is a priori a fine background for quantum gravity. However we see that the situation is rather subtle because we gain an extra local symmetry around the defect, due to the normal bundle. \matt{make this better or remove} \arun{In general I found this confusing -- is ``boundary'' in the sense of a boundary theory to a QFT? Or a boundary of a manifold? Then again, I'm likely not the intended audience for this remark}
\end{rem}
\end{comment}

%%%%%%%%%%%%%%%%%%%%%%%%%%%%%%%%%%%%%%%%%%%%%%%%%%%%%%%
\subsection{\texorpdfstring{Freedman--Kirby$^\mathrm{O}$ Characteristic Bordism}{}}\label{subsubsection:pincKT}
%%%%%%%%%%%%%%%%%%%%%%%%%%%%%%%%%%%%%%%%%%%%%%%%%%%%%%

We now turn to the pairs of FK$^\mathrm{O}$ type in Definition \ref{defn:FKO}, when the ambient manifold $M$ is unoriented, but $F$ is co-oriented i.e. its normal bundle is oriented and Poincaré dual to $w_2(TM)$.\footnote{If $M$ is oriented that implies that $F$ is co-oriented by applying the Whitney sum formula to the decomposition $TM|_F = TF \oplus \nu_F$.} The classifying space $B\mathrm{FK}^{\mathrm{O}}$ of FK$^\mathrm{O}$ characteristic pairs lives in the following pullback:
\begin{equation}\label{eq:FKcpullback}
    \begin{tikzcd}[column sep=2cm, row sep=2cm]
 B\mathrm{FK}^{\mathrm{O}} \arrow[r] \arrow[d,"V"] \arrow[dr, phantom, "\lrcorner", very near start] & M\SO_2 \arrow[d, "U"] \\
B\O \arrow[r, "\mathcal{P}"] & K(\Z/2,2)\,.
\end{tikzcd}
\end{equation}

We will subsequently show that it turns out that the tangential structure for $(M,F)$ does not depend on whether $F$ is Poincaré dual to $\mathcal{P}=w_2(V)$ or $\mathcal{P}=w_2(V)+w_1(V)^2$, however the tangential structure on $M\backslash F$ does.
\begin{defn}
    Let $V\rightarrow X$ be a vector bundle. A \term{pin$^c$ structure} on $V$ is a trivialization of $\Box_{\Z}(w_2(V))$ where $\Box_{\Z}\colon H^2(X;\Z/2) \rightarrow H^3(X;\Z)$ is the Bockstein operator.
\end{defn}

\begin{lem}[\v{C}adek~\cite{Cad99}]
    Let $V\to X$ be a vector bundle and $\widetilde e\in H^1(B\O_1;\Z_{w_1})$ denote the twisted first Euler class. Then
    \begin{equation}
        \widetilde e(\det(V))^2\bmod 2 = w_1(V)^2.
    \end{equation}
\end{lem}
Because the product of two twisted cohomology classes is untwisted (see~\cite{Cad99}), $\widetilde e{}^2$ is an \emph{un}twisted class in $H^2(B\O_1;\Z)$, in fact the unique nonzero class (\textit{ibid.}, Theorem 1).
    
 %   integral lift of $w_1(V)^2$ always exists and is given by $\tilde{e}^2 \bmod 2$ where $\tilde{e} \in H^1(X;\Z_{w_1(\det V)})$ is the twisted Euler class of $\det V\to X$.
%\end{lem}
%\arun{I'm a little confused by the statement of this lemma. Is the point that $w_1(V)^2$ always has a lift, and this lift is $\widetilde e^2$?}\matt{yep}
\begin{cor}
The conditions $\Box_\Z(w_2(V)) = 0$ and $\Box_\Z(w_2(V) + w_1(V)^2) = 0$ are equivalent.
\end{cor}
So there is no difference between the \pinm and \pinp versions of a \pinc structure.\footnote{There are other, similar-in-spirit tangential structures where the \pinp and \pinm analogues are different, including Freed--Hopkins' pin\textsuperscript{$\tilde c+$} and pin\textsuperscript{$\tilde c-$} structures~\cite[Proposition 9.4]{FH21InvertibleFT} and Nakamura's spin\textsuperscript{$c+$} and spin\textsuperscript{$c-$} structures~\cite[\S 3]{Nak13}.}
In addition, if $\mathcal{P}=w_2(V)$ has an integral lift i.e.\ if $V$ is pin$^c$, then $\cP=w_2(V)+w_1(V)^2$ automatically has an integral lift. However, depending on whether $F$ is Poincaré dual to $w_2(TM)$ or $w_2(TM)+w_1(TM)^2$, then $M\backslash F$ has either a pin$^+$ resp.\ pin$^-$ structure by \cref{lem:restricta}. 
Following the same proof as in \cref{prop:FKpairs,prop:FKapprox} we have the following:
\begin{prop}\label{prop:FKOcharpair}
    The structure of $\mathrm{FK}^\mathrm{O}$ characteristic pairs is equivalent to a pin$^c$-structure.
\end{prop}
\begin{prop}\label{prop:FKOapprox}
\hfill
\begin{enumerate}
    \item \Cref{char_conj} holds for $\mathrm{FK}^\O$ characteristic bordism: there are maps of spectra $\mathcal R^\pm\colon \mathit{MTFK}^\O\to\Sigma^2 \MTPin^\pm\wedge (B\SO_2)_+$ which on homotopy groups is the characteristic map $(M,F)\mapsto F$.
    \item The induced characteristic long exact sequence~\eqref{eq:schematic}, which has the form
    \begin{equation}\label{pinc_LES}
         \ldots \longrightarrow \Omega^{\Pin^{\pm}}_k \xlongrightarrow{i} \Omega^{\mathrm{FK}^{\mathrm{O}}}_k \xlongrightarrow{R} \Omega^{\Pin^{\pm}}_{k-2}(B\SO_2) \xlongrightarrow{a} \Omega^{\Pin^{\pm}}_{k-1}\longrightarrow \ldots\,.
    \end{equation}
    are isomorphic to the Smith long exact sequences associated to the Smith homomorphism $\mathrm{sm}_L\colon\MTPin^c\to\Sigma^2\MTPin^\pm\wedge (B\SO_2)_+$.
\end{enumerate}
\end{prop}
The reason for the pair of long exact sequences is because \pinc structures can be described as either twisted \pinp or twisted \pinm structures. The choice of defect $F$ depends on this description.
%$We thus see that there are two interpretations of pin$^c$ on characteristic bordism which depend on the choice of $\cP$,  which gives rise to two sequences in characteristic bordism
 %   \begin{equation}
 %  
 %   \end{equation}
% pin$^c$ structure for $V \rightarrow X$ can equivalently be presented as the data of a complex line bundle $L$ with $w_2(L)= w_2(V)$, and no condition on $w_1(V)$. 
% As in the case of the characteristic long exact sequence involving FK characteristic pairs, we find:
%\begin{prop}\label{prop:FKOapprox}
 %   The Smith long exact sequence for pin$^c$ is isomorphic to the characteristic long exact sequence for FK$^\mathrm{O}$ pairs.
%\end{prop}
%\noindent
%Specifically, using the description of pin$^c$ as a $(B\mathrm{U}_1,L)$-twisted pin$^+$ structure on $V$, we have the following sequence:
 %      \begin{equation}
  %     \label{pinc_LES}
   %     \ldots \longrightarrow \Omega^{\Pin^{\pm}}_k \xlongrightarrow{p^*} \Omega^{\Pin^c}_k \xlongrightarrow{S_L} \Omega^{\Pin^{\pm}}_{k-2}(B\SO_2) \xlongrightarrow{} \Omega^{\Pin^{\pm}}_{k-1}\longrightarrow \ldots\,.
   % \end{equation}
Hambleton--Kreck--Teichner~\cite[\S 2]{HKT94} work out the analogue of~\eqref{pinc_LES} for the bordism groups of \emph{topological} manifolds with topological analogues of \pinp, \pinm, and \pinc structures. Moreover, as a consequence of \cref{cpx_smith_isom}, these long exact sequences split:
\begin{prop}
For all $k$, $\Omega_k^{\Pin^\pm}(B\SO_2)\cong \Omega_k^{\Pin^{\pm}}\oplus \Omega_{k-2}^{\Pin^c}$.
\end{prop}
This result extends work of Hertl~\cite[Theorem 3.1.1]{Her17} and Davighi--Lohitsiri~\cite[\S A.1]{Davighi:2020uab}, who calculated $\Omega_*^{\Pin^-}(B\SO_2)$, resp.\ $\Omega_*^{\Pin^+}(B\SO_2)$, in low degrees.
\begin{proof}
This is a consequence of \cref{cpx_smith_isom} and the equivalences $\mathit{MTPin}^\pm\simeq \MTSpin\wedge (B\O_1)^{\pm(1-\sigma)}$ (see~\cite[\S 7]{Pet68} and~\cite[\S 8]{Sto88}) and $\mathit{MTPin}^c\simeq \MTSpin\wedge (B\O_1)^{\pm(1-\sigma)}\wedge \Sigma^{-2}M\SO_2$ (see~\cite[\S 4C]{HS13} or~\cite[(10.3)]{FH21InvertibleFT}).
\end{proof}
   
We now shift focus to interpreting $\mathrm{FK}^\mathrm{O}$ in the context of introducing a defect along $F$. Much of this interpretation parallels the approach taken in the FK case, with the following exception. We see that if one started off with a theory on $M$ that is unoriented, then inserting the defect $F$ gives the theory that lives on $M\backslash F$ a specific tangential structure depending on whether $F$ was Poincaré dual to $w_2(TM)$ or $w_2(TM)+w_1(TM)^2$. This is an instance of how a choice of defect can explicitly change the global structure of the theory i.e.\ determines what time-reversal structure it has.

%%%%%%%%%%%%%%%%%%%%%%%%%%%%%%%%%%%%%%%%%
\subsection{Guillou--Marin Characteristic Bordism}\label{subsection:GMcharBord}
%%%%%%%%%%%%%%%%%%%%%%%%%%%%%%%%%%%%%%%%%
In this section, we study the tangential structure consisting of an orientation and a Poincaré dual submanifold $F$ of $w_2$, but unlike for Freedman--Kirby structures, we do not require $F$ to be oriented. Guillou--Marin~\cite{GM} first considered this structure and computed its degree-$0$ and $4$ bordism groups.
%the implication when the defect is unoriented, but the ambient manifold is oriented. The groups are computed by Guillou--Marin \cite{GM} in degree 0 and degree 4.
Using our methods we recover their results and are able to extend the computation to other degrees. 

The obstruction to finding a characteristic pair where the unoriented submanifold $i\colon F\hookrightarrow M$ is Poincaré dual to $w_2(TM)$ is given as an obstruction to a map $f$ in
\begin{equation}
    \begin{tikzcd}
        && M\O_2 \arrow[d]\\
        B\SO \arrow[rr,"w_2",swap ] \arrow[urr,dotted,"f"] & &K(\Z/2,2)\,.
    \end{tikzcd}
\end{equation}
The classifying space $B\mathrm{GM}$ of Guillou--Marin characteristic pairs, given in Definition \ref{def:GM}, therefore fits into the following pullback square:
\begin{equation}\label{eq:GMpullback}
    \begin{tikzcd}[column sep=2cm, row sep=2cm]
B\mathrm{GM} \arrow[r] \arrow[d] \arrow[dr, phantom, "\lrcorner", very near start] & M\O_2 \arrow[d, "U"] \\
B\SO \arrow[r, "w_2(TM)",swap] & K(\Z/2,2)
\end{tikzcd}
\end{equation}
which implies that a map to $B\mathrm{GM}$ is the data of a map to $B\SO$ which picks the orientation of $M$, a map to $M\O_2$, and an identification of $w_2(TM)$ and $U$.  
\begin{prop}\label{prop:GMcharpairs}
    The structure of GM characteristic pairs is equivalent to a $(M\O_2, 0, U)$-twisted spin structure.
\end{prop}
Unlike in previous sections, it is not immediate whether there is a vector bundle $W\to M\O_2$ with $w_1(W) = 0$ and $w_2(W) = U$. Rather than worry about whether such a bundle $W$ exists, we follow~\cite{DY:2023tdd} and calculate with $(0, U)$ directly. We briefly review the key steps.
\begin{enumerate}
    \sloppy
    \item Given an $E_\infty$-ring spectrum $R$, May~\cite[\S III.2]{MQRT77} constructs a space $B\GL_1(R)$. In the framework of Ando-Blumberg-Gepner~\cite{ABG18}, twists of $R$-homology over a space $X$ are classified by maps $f\colon X\to B\GL_1(R)$: specifically, the $f$-twisted $R$-homology groups of $X$ are the homotopy groups of the \emph{$R$-module Thom spectrum} $M^R f$ constructed by Ando-Blumberg-Gepner-Hopkins-Rezk~\cite{ABGHR14a, ABGHR14b}.
    \item For $R = \MTSpin$, a theorem of Beardsley~\cite[Theorem 1]{Bea17} implies there is a map $T\colon B\O/B\Spin\to B\GL_1(\MTSpin)$~\cite[\S 1.2.3]{DY:2023tdd}, and there is a canonical homotopy equivalence $B\O/B\Spin\simeq K(\Z/2, 1)\times K(\Z/2, 2)$ (\textit{ibid.}, Proposition 1.33).\footnote{See also Beardsley-Luecke-Morava~\cite[\S 5.2]{BLM23} for some related perspectives and results constructing twists of spin bordism.}
    \item Given $a\in H^1(X;\Z/2)$ and $b\in H^2(X;\Z/2)$, we therefore get a map $f_{a,b}\colon X\to\allowbreak B\GL_1(\MTSpin)$ defined to be the composition
    \begin{equation}
        f_{a,b}\colon X \overset{(a,b)}{\longrightarrow} K(\Z/2, 1)\times K(\Z/2, 2)\simeq B\O/B\Spin \overset T\longrightarrow B\GL_1(\MTSpin).
    \end{equation}
    Hebestreit-Joachim~\cite[Corollary 3.3.8]{HJ20} show that the homotopy groups of the Thom spectrum $M^\MTSpin f_{a,b}$ are naturally isomorphic to the bordism groups of manifolds with $(X, a,b)$-twisted spin structures (\cref{nonVB}).
    \item \label{item:approx} As explained in~\cite[Corollary 2.37]{DY:2023tdd}, because the Atiyah--Bott--\allowbreak Shapiro map $\widehat A\colon \MTSpin\to\ko$~\cite{ABS64, Joa04} is $7$-connected~\cite{ABP67}, in order to compute $(X, a,b)$-twisted spin bordism groups in degrees $7$ and below, one may replace $M^\MTSpin f_{a,b}$ with $M^\ko(\widehat A\circ f_{a,b})$. 
\end{enumerate}
Thus we want to compute $\pi_k(M^\ko (\widehat A\circ f_{0, U}))$ for $k\le 5$, as these are isomorphic to the bordism groups of $k$-dimensional $(M\O_2, 0, U)$-twisted spin manifolds.

%\arun{So we need to introduce non-vector-bundle twisted spin structures! This also happens in \S 4.3}
%\matt{i agree. should we do it right here, or in section 2? 12/3: we decided to do it in section 2 and reference it here}
%
%\noindent By the shearing construction in Lemma \ref{lem:shearing}, we can view this twisted spin structure as the spectrum \arun{This is not literally true: if a twist is not given by a vector bundle, it may be that the Thom spectrum does not factor as $\MTSpin\smash X$ for any $X$. This problem happened for $\Spin\text{-}\SU_8$ bordism. In this example, we haven't found a vector bundle, and we haven't prove that no such bundle exists, I think. In any case, the notation $(M\O_2)^U$ is not great. I can find a better way to notate this and the other examples in \S 4}
%\begin{equation}\label{eq:GMspectrum}
%    \MTSpin \wedge (M\O_2)^U\,.
%\end{equation}

We now turn to the tangential structure on $F$. Using the fact that $TM|_F = TF \oplus \nu_F$, and $TM$ is oriented thanks to $M$ being part of a GM pair, we directly see that $w_1(TF) = w_1(\nu_F)$ and $w_2(TF) + w_1(TF)^2 = 0$. 
This is equivalent to a $(B\O_2,\sigma)$-twisted spin structure, where $\sigma\to B\O_n$ always denotes the determinant bundle of the tautological bundle. The map $(M,F) \to F$ which restricts onto the submanifold induces a map on characteristic bordism
\begin{equation}\label{eq:GM}
   R\colon \Omega^{\mathrm{GM}}_4 \rightarrow \Omega^{\Spin}_2((B\O_2)^{\sigma-1})\,.
\end{equation}
Unlike the sequences for Freedman--Kirby characteristic bordism in \S\ref{subsection:FKcharbordism} we are unable to extend the map in \eqref{eq:GM} to a long exact sequence unless we assume \cref{char_conj}. We address possible approximations to a long exact sequence in \S\ref{subsubsection:approxGM}. 

Now we compute the low-dimensional $(M\O_2, 0, U)$-twisted spin bordism groups. At $p = 2$, we will use a variant of the Adams spectral sequence due to Baker--Lazarev~\cite{BL01}, but before that we take care of the odd-primary case, which is simpler.
\begin{lem}\label{GM_odd}
The groups $\pi_k(M^\MTSpin f_{0,U})\otimes\Z[1/2]$ are isomorphic to $\Z[1/2]$ for $k =0$, $\Z[1/2]^{\oplus 2}$ for $k = 4$, and vanish for all other $k<8$.
\end{lem}
\begin{proof}
Forgetting from a spin structure to an orientation defines a homomorphism of $E_\infty$-ring spectra $\phi\colon \MTSpin\to\MTSO$, which is an isomorphism after inverting $2$; thus, the induced map $M^{\MTSpin}f_{0,U}\to M^{\MTSO}(\phi\circ f_{0,U})$ is a homotopy equivalence after inverting $2$. Therefore we will focus on $M^{\MTSO}(\phi\circ f_{0,U})$.

The key fact is that, after base-changing to $\MTSO$, the twist vanishes. Naturality of the map $T\colon B\O/BG\to B\GL_1(\mathit{MTG})$ in $G$ (see~\cite[Theorem 1]{Bea17} and~\cite[\S 1.2]{DY:2023tdd}) implies that $\phi\circ f_{0,U}$ factors through the map $B\O/B\Spin\to B\O/B\SO$ which is the projection $K(\Z/2, 1)\times K(\Z/2, 2)\to K(\Z/2, 1)$ (\textit{ibid.}, Proposition 1.33) -- in other words, sending $(a, b)\mapsto a$. Thus $(0, U)$ maps to the zero twist, whose Thom spectrum is $\MTSO\wedge (B\O_2)_+$ (\textit{ibid.}, Example 1.7). That is, there is a map $\pi_k(M^\MTSpin f_{0,U})\to \Omega_k^\SO(B\O_2)$ which is an isomorphism after tensoring with $\Z[1/2]$; the rest of the proof amounts to setting up the Atiyah--Hirzebruch spectral sequence
\begin{equation}
    E^2_{p,q} = H_p(B\O_2; \Omega_q^\SO\otimes\Z[1/2]) \Longrightarrow \Omega_{p+q}^\SO(B\O_2)\otimes\Z[1/2]
\end{equation}
and observing that it collapses without extension problems to imply the lemma statement.
\end{proof}
This tells us everything except $2$-torsion, which requires a different technique.
\begin{defn}
Let $M$ be a $\ko$-module; then, $H_\ko^*(M)\coloneqq\pi_{-*}\mathrm{Map}_\ko(M, H\Z/2)$.
\end{defn}
That is: $H^*(M; \Z/2)$ is the homotopy groups of the spectrum of maps from $M$ to the Eilenberg-Mac Lane spectrum $H\Z/2$; by instead asking for $\ko$-module maps, we get $H_\ko^*(M)$. The $\ko$-module structure on $H\Z/2$ is the usual one, induced from the (unique) $E_\infty$-$\ko$-algebra structure given by Postnikov truncation $\ko\to H\Z$ followed by mod $2$ reduction $H\Z\to H\Z/2$.
\begin{lem}[{Baker~\cite[Theorem 5.1]{Bak20}}]
There is a canonical isomorphism $H_\ko^*(H\Z/2)\cong \cA(1)$, where $\cA(1)\subset\cA$ is the subalgebra generated by $\Sq^1$ and $\Sq^2$.
\end{lem}
Thus, for any $\ko$-module $M$, $H_\ko^*(M)$ is naturally an $\cA(1)$-module by postcomposition of a class in $\cA(1)$, thought of as a $\ko$-linear map $H\Z/2\to \Sigma^t H\Z/2$.
\begin{thm}[Baker--Lazarev~\cite{BL01}]
\label{BL_thm}
Let $M$ be a finite type, bounded-below $\ko$-module. Then there is a strongly convergent spectral sequence of Adams type
\begin{equation}\label{BLASS}
    E_2^{s,t} = \Ext_{\cA(1)}^{s,t}(H_\ko^*(M), \Z/2) \Longrightarrow \pi_{t-s}(M)_2^\wedge.
\end{equation}
\end{thm}
We will apply this for $M = M^\ko (\widehat A\circ f_{0,U})$, which satisfies the hypotheses of \cref{BL_thm}. The first step is to determine $H_\ko^*(M^\ko (\widehat A\circ f_{0,U}))$.
\begin{thm}[{\cite[Theorem 2.28(3)]{DY:2023tdd}}]
\label{DY23_main_thm}
Let $X$ be a space, $a\in H^1(X;\Z/2)$ and $b\in H^2(X;\Z/2)$, and let $V_\ko(a,b,X)$ denote the $\cA(1)$-module which is a free $H^*(X;\Z/2)$-module on a single generator $Q$ and with $\cA(1)$-action specified by the formulas
\begin{subequations}
\label{fake_Sq12}
\begin{align}
    \Sq^1(Qx) &= Q(ax + \Sq^1(x))\\
    \Sq^2(Qx) &= Q(bx + a\Sq^1(x) + \Sq^2(x))
\end{align}
\end{subequations}
Then, there is a natural isomorphism of $\cA(1)$-modules
\begin{equation}
    H_\ko^*(M^\ko(\widehat A\circ f_{a,b})) \overset\cong\longrightarrow V_\ko(a,b,X).
\end{equation}
\end{thm}
Part of the content of \cref{DY23_main_thm} is that the formulas in~\eqref{fake_Sq12} indeed satisfy the Adem relations and therefore define an $\cA(1)$-action. If there is a vector bundle $V\to X$ with $w_1(V) = a$ and $w_2(V) = b$, then~\eqref{fake_Sq12} is the Wu formula, with $Q$ equal to the Thom class; thus one way to think of \cref{DY23_main_thm} is that the Wu formula holds whether or not $a$ and $b$ come from a vector bundle!
\begin{defn}\label{M1defn}
Let $M_0\coloneqq\cA(1)\otimes_{\cA(0)}\Z/2$ and $M_1$ denote the unique nontrivial $\cA(1)$-module extension of $\Sigma^4 M_0$ by $M_0$. Pictures of $M_0$ and $M_1$ can be found in~\cite[Figures 14, resp.\ 17]{BC18}.
\end{defn}
\begin{defn}
Let $M$ and $N$ be $\cA(1)$-modules, and for any $\cA(1)$-module $L$, let $L_{>n}$ denote the sub-$\cA(1)$-module of $L$ spanned by homogeneous elements of degree strictly greater than $n$. We say $M$ and $N$ are \emph{isomorphic up to degree $n$}, written $M\cong_{\le n} N$, if there is an $\cA(1)$-module isomorphism $M/M_{>n}\cong N/N_{>n}$.
\end{defn}
By using \cref{DY23_main_thm} and the known $\cA(1)$-module structure on $H^*(M\O_2;\Z/2)$ (see for example~\cite[Figure 6.6]{Cam17}), one can prove the following proposition.
\begin{prop}\label{GMA1}
There is an isomorphism up to degree $6$% \matt{[TODO for AD! Deg 6 OK?]}
\begin{equation}
    H_\ko^*(M^\ko (\widehat A\circ f_{0,U})) \cong_{\le 6}
        \textcolor{BrickRed}{M_1} \oplus
        \textcolor{MidnightBlue}{\Sigma^4 M_0} \oplus
        \textcolor{Green}{\Sigma^6 \Z/2}.
\end{equation}
\end{prop}
See \cref{fig:GM}, left, for a picture of \cref{GMA1}. $\Ext_{\cA(1)}(\textcolor{MidnightBlue}{M_0}, \Z/2)$ and $\Ext_{\cA(1)}(\textcolor{BrickRed}{M_1}, \Z/2)$ can be found in~\cite[Figures 15, resp.\ 18]{BC18}. Quotienting above degree $6$ induces an isomorphism in Ext in degrees $t-s\le 5$, so
we can draw the $E_2$-page of Baker--Lazarev's Adams spectral sequence in \cref{fig:GM}, right. In the range we are interested in, there are no differentials for degree reasons, and all extensions are resolved by the $h_0$-action on $E_\infty$, so we discover that there is no $2$-torsion in degrees $5$ and below. Assembling this and \cref{GM_odd}, we have the following.
\begin{prop}\label{prop:bordcompGM}
There are isomorphisms
\begin{equation}
\begin{aligned}
    \Omega_0^{\mathrm{GM}} &\cong \Z\\
    \Omega_1^{\mathrm{GM}} &\cong 0\\
    \Omega_2^{\mathrm{GM}} &\cong 0\\
    \Omega_3^{\mathrm{GM}} &\cong 0\\
    \Omega_4^{\mathrm{GM}} &\cong \Z^2\\
    \Omega_5^{\mathrm{GM}} &\cong 0.
\end{aligned}
 \end{equation}
  %  The bordism groups $\Omega^{\mathrm{GM}}_* = \{\Z, 0,0,\Z\oplus \Z \}$ for $*= 0, \ldots,4$.
\end{prop}

As a check, we see that our results agree with the computations of \cite{GM}, who computed these bordism groups in degrees $0$ and $4$. 

\begin{figure}[!ht]
\begin{subfigure}[c]{0.5\textwidth}
\hspace{16mm}
	\begin{tikzpicture}[scale=0.8, every node/.style = {font=\tiny}]
		\foreach \y in {0, 1, ...,6} {
			\node at (-2, \y) {$\y$};
		}
        \begin{scope}
            \clip (-1, -0.2) rectangle (5, 6.6);
        
		\begin{scope}[BrickRed]
			\tikzpt{0}{0}{$Q$}{};
			\tikzpt{0}{2}{$UQ$}{};
            \tikzpt{0}{3}{$w_1UQ$}{};
			\tikzpt{0}{4}{}{};
            \tikzpt{0}{5}{}{};
            \tikzpt{0}{6}{}{};
            \sqone(0, 2);
            \sqone(0, 4);
             \sqtwoR(0, 0);
              \sqtwoR(0, 3);
              \sqtwoL(0, 4);
              \sqone(0, 6);
			%\tikzpt{0}{4.5}{}{};
			%\threebock(0, 4);
		\end{scope}
  \begin{scope}[MidnightBlue]
			\tikzpt{2}{4}{$w_2 UQ$}{};
            \tikzpt{2}{6}{$\alpha Q$}{};
             \sqtwoR(2, 4);
             \sqone(2, 6);
			%\tikzpt{0}{4.5}{}{};
			%\threebock(0, 4);
		\end{scope}
        \end{scope}
        \begin{scope}[Green]
            \clip (3, 5) rectangle (5, 6.6);
            \tikzpt{4}{6}{$w_1^2w_2UQ$}{};
            \sqtwoR(4, 6);
        \end{scope}
	\end{tikzpicture}
\end{subfigure}
\qquad
\begin{subfigure}[!ht]{0.4\textwidth}
\includegraphics[width=43mm,scale=0.3]{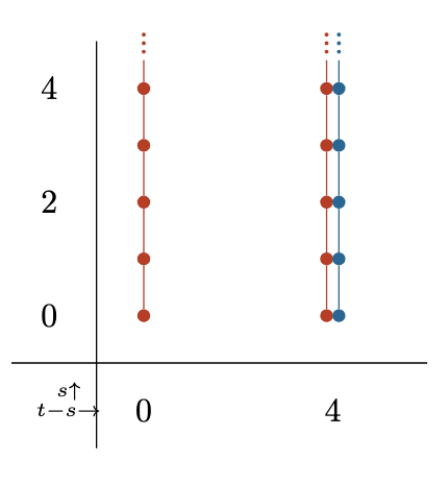}
%\begin{sseqdata}[name=hettmf, Adams grading, classes=fill, xrange={0}{5}, yrange={0}{4},xscale=0.5, yscale=0.6, x tick
%step=4, y tick step=2, class labels = { left = 0.07em, font=\small },
%  x label = {$\displaystyle{\substack{s\uparrow \\ t-s\rightarrow}}$},
 % x label style = {font = \small, xshift = -13.5ex, yshift=1.7ex}]]
%\begin{scope}[BrickRed]
	%\class(0, 0)\AdamsTower{}
	%\class(4, 0)\AdamsTower{}
%\end{scope}
%\begin{scope}[MidnightBlue]
	%\class(4, 0)\AdamsTower{}
%\end{scope}
%\end{sseqdata}
%\printpage[name=hettmf]
\end{subfigure}
\caption{Left: the $\cA(1)$-module structure on $H_\ko^*(M^\ko(\widehat A\circ f_{0,U}))$ in low degrees. The pictured module contains all classes in degrees $6$ and below. As will be the case in all the figures, curved lines that join points separated by two degrees denote actions by $\Sq^2$, and straight lines joining points separated by one degree denote $\Sq^1$ actions. %\arun{I always like to show the modules continuing upward a little bit if they continue past the top of the diagram. I can take care of this}
Right: The $E_2$-page of the Baker--Lazarev Adams spectral sequence computing $2$-completed $(M\O_2, 0, U)$-twisted spin bordism groups. We use this spectral sequence to prove \cref{prop:bordcompGM}.}
    \label{fig:GM}
\end{figure}
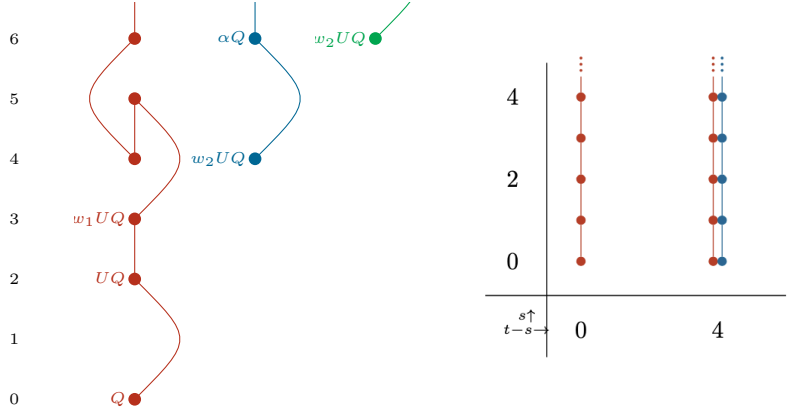

\begin{rem}
    The map $M\O_2\to K(\Z/2, 2)$ is $3$-connected, so in dimensions $3$ and below, a GM-structure is equivalent to a $(K(\Z/2,2),0, T)$-twisted spin structure, where $T \in H^2(K(\Z/2,2);\Z/2)$ is the canonical class. But this is equivalent to an orientation! See~\cite[\S 3]{Bea17} or~\cite[Example 1.28(2)]{DY:2023tdd}. Therefore $\Omega^{\mathrm{GM}}_* \cong \Omega^{\SO}_*$ for $* = 0,1,2,3$, as we saw in \cref{prop:bordcompGM}.
\end{rem}

We now interpret the physical consequences of the GM characteristic pairs. The manifold  $M\backslash F$ is oriented so we do not demand that the theory placed on the manifold to have time-reversal symmetry. Because $F$ is unoriented, a theory defined on the defect supported at $F$ can have a time-reversal symmetry. Due to the fact that the structure on $M\backslash F$ does not extend across $F$, it is possible to define a 4d theory and a 2d theory within the same background manifold, such that the 4d theory does not necessarily have a time-reversal symmetry but the 2d theory does. Our characteristic bordism computation implies that there are two bordism classes of pairs $(M,F)$ in dimension 4, suggesting that there are two classes of QFTs which host these time-reversal invariant topological defects. This example shows how the global structure of the defect $F$ can lead to different symmetry structures for theories that live on $F$ and outside of it.

%%%%%%%%%%%%%%%%%%%%%%%%%%%%%%%%%%%%%%%%%
\subsubsection{Symmetry Approximations to GM}\label{subsubsection:approxGM}
%%%%%%%%%%%%%%%%%%%%%%%%%%%%%%%%%%%%%%%%%
We give the first example of how the Smith long exact sequences can be used to approximate the sequences that arise from characteristic bordism, in the spirit of \S\ref{subsection:smithLESintro}. We will show that a spin-$
\O_2$ structure approximates the structure of GM-characteristic pairs. In our notation, spin-$\O_2$ is the group $(\Spin \times \O_2)/(\langle -1,x \rangle)$ where $(\langle -1,x \rangle)$ generates the diagonal $\Z/2$-subgroup. The classifying space of spin-$\O_2$ structures fits into the following pullback square:
    \begin{equation}
    \begin{tikzcd}[column sep=2cm, row sep=2cm]
B(\Spin\!\text{-}\O_2) \arrow[r,"V"] \arrow[d] \arrow[dr, phantom, "\lrcorner", very near start] & B\O_2\arrow[d, "w_2(V)"] \\
 B\SO \arrow[r, "w_2(TM)",swap] & K(\Z/2,2)\,.
\end{tikzcd}
\end{equation}
The map $B\O_2 \rightarrow K(\Z/2,2)$ is the mod 2 Euler class of $V$, which is $w_2(V)$.
The pullback then implies that the data of a spin-$\O_2$ structure is that of an orientation on $TM$, and an identification $w_2(TM) = w_2(V)$. We note that due to the map $z:B\O_2\rightarrow M\O_2$, i.e. taking the zero section, we have a map $B(\Spin\!\text{-}\O_2) \rightarrow B\mathrm{GM}$.
\begin{lem}
\label{the_spin_O2_twist}
    Let $V\to B\O_2$ be the tautological rank-$2$ vector bundle. A $\Spin\!\text{-}\O_2$ structure on $M$ is equivalent to a $(B\O_2, V\oplus 3 \det(V))$-twisted spin structure.
\end{lem}
\begin{proof}
    Since a spin structure induces an orientation, we must have $w_1(TM\oplus V \oplus 3\det(V))=0$. It follows from a straightforward application of the Whitney sum formula that $w_1(TM\oplus V \oplus 3\det(V))=0$ implies that $w_1(TM)=0$ since $w_1(V)=w_1(3\det(V))$. Furthermore  $w_2(TM\oplus V \oplus 3\det(V))= 0$ implies
\begin{subequations}
\begin{align}
    w_2(TM)+w_2(V \oplus 3 \det(V)) +w_1(TM)w_1(TM \oplus 3 \det(V))&= 0\,,\\
     w_2(TM) + w_2(V) +w_2(3\det(V))+w_1(V)w_1(3\det(V))&= 0\,,
\end{align}
\end{subequations}
and thus $w_2(TM)= w_2(V)$. 
\end{proof}
See~\cite[Corollary 14.16]{DDHM23} for an analogue of \cref{the_spin_O2_twist} with $\O_2$ replaced with the dihedral group $D_8$.

We compute the bordism groups for spin-$\O_2$ using the shearing construction in Lemma \ref{lem:shearing} to write this twisted spin structure as the spectrum 
\begin{equation}
    \MTSpin \wedge (B\O_2)^{V\oplus 3\det(V) - 5}\,,
\end{equation}
There is no odd-primary torsion, which can be proven in essentially the same manner as in \cref{GM_odd}. At $p = 2$, we use the Adams spectral sequence. The Wu formula tells us the $\cA(1)$-action on $H^*((B\O_2)^{V + 3\det(V) - 5};\Z/2)$:
\begin{equation}
    \Sq^1 U = 0, \quad \Sq^2 U = w_2(V)U\,.
\end{equation}
where $U\in H^0((B\O_2)^{V\oplus 3\det(V)};\Z/2)$ is the Thom class. Computing with this quickly gives a complete calculation of the $\cA(1)$-module structure in low degrees. To specify it, though, we must name a few commonly occurring $\cA(1)$-modules.
\begin{defn}\hfill
\begin{enumerate}
    \item The \emph{elephant}, denoted $R_2$, is the kernel of the unique nonzero $\cA(1)$-module homomorphism $\Sigma^{-1}\cA(1)\to\Sigma^{-1}\Z/2$.\footnote{The name ``elephant'' is due to Buchanan--McKean~\cite[Figure 1]{BM23}.}
    \item The \emph{upside-down question mark} is $Q\coloneqq \cA(1)/(\Sq^1, \Sq^2\Sq^3)$.
\end{enumerate}
\end{defn}
%the nontrivial homotopy groups in degrees less than 5. The red module is sometimes referred to as $R_2$, and the Adams chart for $\Ext^{s,t}_{\mathcal{A}(1)}(R_2;\Z/2)$ is given in \cite[Figure 29 (top)]{BC18}. The green module is dubbed the ``upside down" question mark $Q$, and its Ext is computed in \cite[Figure 29 (bottom)]{BC18}.
\begin{prop}
There is an $\cA(1)$-module isomorphism
\begin{equation}
    H^*((B\O_2)^{V + 3\det(V) - 5};\Z/2) \cong
        \textcolor{BrickRed}{R_2} \oplus
        \textcolor{Green}{\Sigma^3 \cA(1)} \oplus
        \textcolor{MidnightBlue}{\Sigma^4 Q} \oplus 
        \textcolor{Fuchsia}{\Sigma^5\cA(1)} \oplus P,
\end{equation}
where $P$ is concentrated in degrees $7$ and above.
\end{prop}
We draw a picture of this isomorphism in \cref{fig:spinO2}, left. Next, we look up the Ext groups of these summands to determine the $E_2$-page of the Adams spectral sequence: $\Ext_{\cA(1)}^{s,t}(\textcolor{BrickRed}{R_2}, \Z/2)$ and $\Ext_{\cA(1)}^{s,t}(\textcolor{MidnightBlue}{Q}, \Z/2)$ can be found in~\cite[Figure 29]{BC18}. Using this, we draw the $E_2$-page of the Adams spectral sequence in \cref{fig:spinO2}, right.

\begin{figure}[!ht]
\centering
\begin{subfigure}[c]{0.6\textwidth}
	\begin{tikzpicture}[scale=0.56, every node/.style = {font=\tiny}]
		\foreach \y in {0, 1, ...,11} {
			\node at (-2, \y) {$\y$};
		}
		\begin{scope}[BrickRed]
            \Rtwo{0}{0}{$U$}{$Uw_1$};
         %\tikzpt{0}{0}{$U$}{};
         %\tikzpt{0}{1}{$Uw_1$}{};
         %\tikzpt{0}{2}{}{};
         % \tikzpt{1.5}{2}{}{};
         % \tikzpt{1.5}{3}{}{};
         % \tikzpt{1.5}{4}{}{};
         % \tikzpt{1.5}{5}{}{};
         % \sqone(0, 1);
         % \sqone(1.5, 4);
         %  \sqone(1.5, 2);
         %  \sqtwoCR(0, 0);
         %  \sqtwoCR(0, 1);
         %  \sqtwoCR(0, 2);
         %   \sqtwoR(1.5, 3);
	\end{scope}
 \begin{scope}[Green]
		\Aone{3.5}{3}{$Uw_1^3$};
	\end{scope}
     \begin{scope}[MidnightBlue]
        \SpanishQnMark{6}{4}{$Uw_2^2$};
%		\tikzpt{6}{4}{}{};
         %\tikzpt{6}{5}{}{};
 %        \tikzpt{6}{6}{}{};
 %         \tikzpt{6}{7}{}{};
 %       \draw (6,3.7) node {$Uw_2^2$};
 %        \sqone(6, 6);
 %        \sqtwoL(6, 4);
	\end{scope}
    \begin{scope}[Fuchsia]
        \Aone{7.5}{5}{$Uw_1w_2^2$};
    \end{scope}
	\end{tikzpicture}
\end{subfigure}
\begin{subfigure}[!ht]{0.32\textwidth}
\includegraphics[width=46mm,scale=0.7]{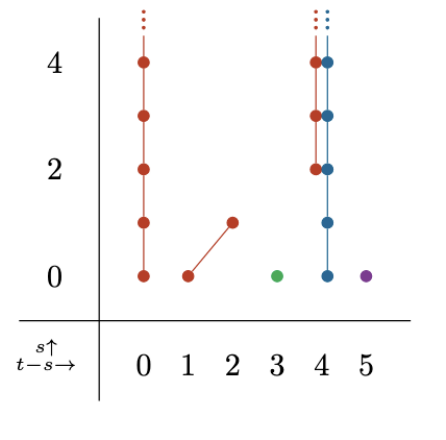}
%\begin{sseqdata}[name= spinO2, Adams grading, classes=fill, xrange={0}{5}, yrange={0}{4},xscale=0.5, yscale=0.6, x tick
%step=1, y tick step=2, class labels = { left = 0.07em, font=\small },
 %class pattern = linear, class placement transform = { rotate = 45 },
  %x label = {$\displaystyle{\substack{s\uparrow \\ t-s\rightarrow}}$},
  %x label style = {font = \small, xshift = -15.5ex, yshift=1.7ex}]]
%\begin{scope}[BrickRed]
 %   \class(0, 0)\AdamsTower{}
	%\class(4, 2)\AdamsTower{}
    %\class(1, 0)
    %\class(2, 1)
    %\draw (1,0) --(2,1);
%\end{scope}
%\begin{scope}[white]
 %   \class(4, 0)
  %  \class(4,1)
%\end{scope}
%\begin{scope}[Green]
	%\class(3, 0)
%\end{scope}
%\class[draw=none, fill=none](4, 0)
%\begin{scope}[MidnightBlue]
%\class(4, 0)\AdamsTower{}
%\end{scope}
%\class[Fuchsia](5, 0)
%\end{sseqdata}
%\printpage[name=spinO2]
\end{subfigure}
\caption{Left: The $\cA(1)$-module structure on $H^*((B\O_2)^{V\oplus 3\det(V)-5})$; this summand includes all classes in degree $6$ and below. Right: The $E_2$-page of the Adams spectral sequence computing $2$-completed spin-$\O_2$ bordism. We use this in \cref{lem:bordismspinO2}.}
    \label{fig:spinO2}
\end{figure}
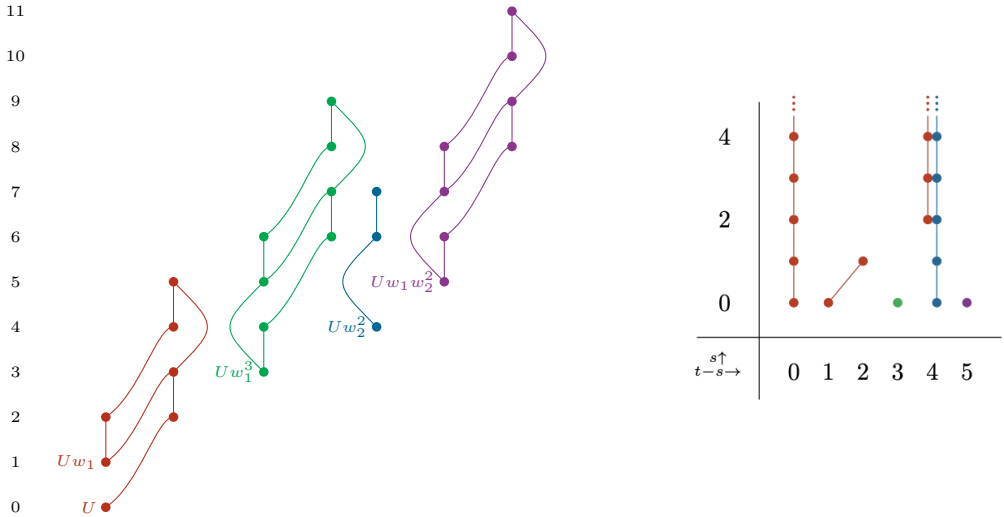

\begin{prop}\label{lem:bordismspinO2}
    There are isomorphisms 
     \begin{equation}
\begin{aligned}
    \Omega_0^{\Spin\!\text{-}\O_2} &\cong \Z\\
    \Omega_1^{\Spin\!\text{-}\O_2}&\cong \Z/2\\
    \Omega_2^{\Spin\!\text{-}\O_2} &\cong \Z/2\\
    \Omega_3^{\Spin\!\text{-}\O_2} &\cong \Z/2\\
    \Omega_4^{\Spin\!\text{-}\O_2} &\cong \Z \oplus \Z\\
    \Omega_5^{\Spin\!\text{-}\O_2} &\cong \Z/2\\
%    \Omega_6^{\Spin\!\text{-}\O_2} &\cong 0.
\end{aligned}
 \end{equation}
\end{prop}
In degrees $3$ and below, \cref{lem:bordismspinO2} is due to Stehouwer~\cite[\S 4.1]{Ste22}; the rest is new.

We can construct a Smith long exact sequence with $\Spin\text{-}\O_2$ using the notation in \S\ref{subsection:smithLESintro} in the following way: we take the manifold $X= B\O_2$ and the two bundles over $B\O_2$ to be $V$ and $\sigma=\det(V)$ as per \eqref{eq:smithLES}. This gives the sequence\footnote{In previous work \cite[Proof of Proposition A.25]{DYY23}, we also studied an additional Smith long exact sequence involving spin-$\O_2$ bordism.}
\begin{equation}\label{eq:smithSpinO2}
       \ldots \rightarrow \Omega^{\Spin}_{k-1}\left((B\O_2)^{\sigma-1}\right) \rightarrow \Omega^{\Spin}_k(B\O_1) \rightarrow \Omega^{\Spin\text{-}\O_2}_k\xlongrightarrow{S_V} \Omega^{\Spin}_{k-2}\left((B\O_2)^{\sigma-1}\right)\rightarrow \Omega^{\Spin}_{k-1}(B\O_1)\rightarrow \ldots\,,
    \end{equation}
where the sphere bundle of $V\rightarrow B\O_2$ is given by $B\O_1$ and the pullback of $V\oplus 3 \det(V)$ is given by 4 times the canonical bundle over $B\O_1$, which is spin. Applying the Smith map with respect to the bundle $V$ leads to a $(B\O_2,2V+3\sigma)$-twisted spin structure, but since $2V+2\sigma$ is spin, this is equivalent to a $(B\O_2,\sigma)$-twisted spin structure (see e.g.\ \cite[Theorem 1.29]{Deb21}), and gives the fourth entry in \eqref{eq:smithSpinO2}, and matches the target of the map in \eqref{eq:GM}.

By inspecting the groups $\Omega^{\Spin\text{-}\O_2}_k$ in Proposition \ref{lem:bordismspinO2}
we see why $\Omega^{\Spin\text{-}\O_2}_k$ only approximates $\Omega^{\mathrm{GM}}_k$. Nevertheless, we will be able to gain some information on the characteristic LES by studying $(B\O_2, \sigma)$-twisted spin bordism.
\begin{prop}\label{smith_spO2}
There are isomorphisms
\begin{equation}
\begin{aligned}
    \Omega_0^{\Spin}((B\O_2)^{\sigma-1}) &\cong \Z/2\\
    \Omega_1^{\Spin}((B\O_2)^{\sigma-1}) &\cong \Z/2\\
    \Omega_2^{\Spin}((B\O_2)^{\sigma-1}) &\cong \Z\oplus \Z/8\\
    \Omega_3^{\Spin}((B\O_2)^{\sigma-1}) &\cong \Z/2\\
    \Omega_4^{\Spin}((B\O_2)^{\sigma-1}) &\cong 0\\
    \Omega_5^{\Spin}((B\O_2)^{\sigma-1}) &\cong \Z/16.
\end{aligned}
\end{equation}
\end{prop}
\begin{proof}
By~\cite[Lemma 3.30]{Deb21}, there is a splitting
\begin{equation}
    (B\O_2)^{\sigma-1} \overset\simeq\longrightarrow (B\O_1)^{\sigma-1}\vee M
\end{equation}
for a spectrum $M$ whose mod $2$ cohomology is, as an $\cA(1)$-module, isomorphic to a complementary subspace of $H^*((B\O_2)^{\sigma-1};\Z/2)$ to the one spanned by $\{Uw_1^k: k\ge 0\}$. Since $\Omega_*^{\Spin}((B\O_1)^{\sigma-1})\cong\Omega_*^{\Pin^-}$~\cite[\S 7]{Pet68}, we will focus on computing $\Omega_*^{\Spin}(M)$, then at the end direct-sum on \pinm bordism groups, which are computed in~\cite{ABP69}.

There is no odd-primary torsion in $\Omega_*^{\Spin}(M)$, and the proof is similar to previous cases in this paper. At $p = 2$, we again use the Adams spectral sequence. Using the description of $H^*(M;\Z/2)$ in the previous paragraph, we obtain an $\cA(1)$-module isomorphism
\begin{equation}
    H^*(M;\Z/2) \cong \textcolor{BrickRed}{\Sigma^2 Q} \oplus \textcolor{MidnightBlue}{\Sigma^3\cA(1)} \oplus P,
\end{equation}
where $P$ is concentrated in degrees $6$ and above. We draw this in \cref{pinmbutO2}, left. We have already seen the Ext groups of $\textcolor{BrickRed}{Q}$ and $\textcolor{MidnightBlue}{\cA(1)}$, so we can draw the $E_2$-page of the Adams spectral sequence in \cref{pinmbutO2}, right. Margolis' theorem~\cite{Mar74} implies this spectral sequence collapses in the range depicted, and all extension problems are solved by the $h_0$-action on the $E_\infty$-page, finishing the proof.
\end{proof}
\begin{figure}[!ht]
\centering
\begin{subfigure}[c]{0.3\textwidth}
    \begin{tikzpicture}[scale=0.6, every node/.style = {font=\tiny}]
	\foreach \y in {2, 3, ...,9} {
		\node at (-2, \y) {$\y$};
    }
    \begin{scope}[BrickRed]
        \SpanishQnMark{0}{2}{$Uw_2$};
    \end{scope}
    \begin{scope}[MidnightBlue]
        \Aone{1.75}{3}{$Uw_1w_2$};
    \end{scope}
    \end{tikzpicture}
\end{subfigure}
\begin{subfigure}[c]{0.5\textwidth}
\includegraphics[width=60mm,scale=.9]{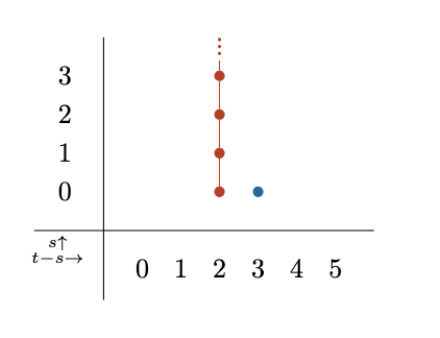}
 %   \begin{sseqdata}[name= smspinO2, Adams grading, classes=fill, xrange={0}{5}, yrange={0}{3},scale=0.5, class labels = { left = 0.07em, font=\small },
% class pattern = linear, class placement transform = { rotate = 45 },
 % x label = {$\displaystyle{\substack{s\uparrow \\ t-s\rightarrow}}$},
  %x label style = {font = \small, xshift = -15.5ex, yshift=1.7ex}]]
   %     \begin{scope}[BrickRed]
    %        \class(2, 0)\AdamsTower{}
     %   \end{scope}
     %   \class[MidnightBlue](3, 0)
    %\end{sseqdata}
   %\printpage[name=smspinO2]
\end{subfigure}
\caption{Left: The $\cA(1)$-module structure on $H^*((B\O_2)^{\sigma-1};\Z/2)/H^*((B\O_1)^{\sigma-1}; \Z/2)$; this summand includes all classes in
degree $5$ and below. Right: The $E_2$-page of the corresponding Adams spectral sequence, which computes a summand of the $2$-completion of $\Omega_*^{\Spin}((B\O_2)^{\sigma-1})$. We use this in \cref{smith_spO2}.}
\label{pinmbutO2}
\end{figure}
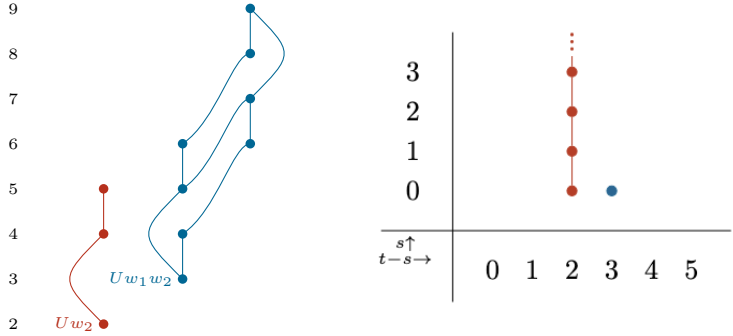
\begin{thm}
\label{GM_LES_thm}
Assume \cref{char_conj}, so that we may let $F_{\mathrm{GM}}$ denote the fiber of the characteristic map of spectra $\mathit{MTGM}\to\Sigma^2 \MTSpin\wedge (B\O_2)^{\sigma-1}$. Then there are isomorphisms
\begin{equation}
\begin{alignedat}{2}
    \pi_0(F_{\mathrm{GM}}) &\cong \Z\qquad\qquad & \pi_3(F_{\mathrm{GM}}) &\cong 0\\
    \pi_1(F_{\mathrm{GM}}) &\cong \Z/2\qquad\qquad & \pi_4(F_{\mathrm{GM}}) &\cong \Z\oplus\Z/2\\
    \pi_2(F_{\mathrm{GM}}) &\cong \Z/2\qquad\qquad & \pi_5(F_{\mathrm{GM}}) &\cong 0.
\end{alignedat}
\end{equation}
In addition, the characteristic long exact sequence is as given in \cref{GM_char_LES}.
\end{thm}
\begin{proof}
We will use the characteristic long exact sequence to compute $\pi_*(F_{\mathrm{GM}})$ in degrees $5$ and below. To do so, we need $\Omega_*^{\mathrm{GM}}$ in degrees $5$ and below, which we computed in \cref{prop:bordcompGM}, and $\Omega_*^\Spin((B\O_2)^{\sigma-1})$ in degrees $4$ and below, which we computed in \cref{smith_spO2}. Then, the theorem follows by exactness of the characteristic long exact sequence.
\end{proof}

\begin{figure}[!ht]
    \centering
\begin{tikzcd}
	{k} & {\pi_k(F_{\mathrm{GM}})} & {\Omega_k^{\mathrm{GM}}} & {\Omega_{k-2}^{\Spin}((B\O_2)^{\sigma-1})} \\
	0 & \Z \arrow[r,"\cong"]& {\Z} & 0 \\
	1 & {\Z/2} & 0 & 0 \\
	2 & {\Z/2}  & 0 & {\Z/2}\arrow[from=4-4,to=3-2,in=-150, out=30,"\cong",swap] \\
	3 & {0} & 0 & {\Z/2} \arrow[from=5-4,to=4-2,in=-150, out=30,"\cong",swap]\\
	4 & {\Z \oplus\Z/2}\arrow[r,"\phi"]
    & {\Z^2}\arrow[r, ->>] & {\Z\oplus\Z/8} \\
	5 & {0} & {0} & {\Z/2} \arrow[from=7-4,to=6-2,in=-150, out=30,"{(0, 1)}"]\\
	6 & {} & {} & 0\\
\end{tikzcd}
    \caption{The characteristic long exact sequence for Guillou--Marin's characteristic structure, which we prove in \cref{GM_LES_thm}. The map $\phi$ sends $(1,0)\mapsto (0,8)$ and $(0,1)\mapsto 0$.}
    \label{GM_char_LES}
\end{figure}
After applying Anderson duality to this long exact sequence, we obtain a long exact sequence of invertible field theories describing anomalies of field theories with Guillou--Marin structure and their relationships with anomalies on the respective defect theories.
\begin{cor}
\label{GM_physics}
Assuming \cref{char_conj}, consider a $k$-dimensional field theory on manifolds with Guillou--Marin structure. For $k = 0$, $1$, and $3$, the defect map $\mho_\Spin^{k-1}((B\O_2)^{\sigma-1})\to\mho_{\mathrm{GM}}^{k+1}$ is $0$. For $k= 2$, the map is nonzero and there is a $\Z$-valued obstruction to pulling the anomaly of the theory back to the defect.
\end{cor}
We would be interested in understanding this $2$-dimensional obstruction more explicitly in example field theories. Because the obstruction group is torsion-free, the effect on anomalies should be visible at the perturbative level, and thus may be easier to detect than in other examples.

\begin{rem}\label{KT_correction}
A version of a characteristic long exact sequence was proposed in \cite[Remark 6.15]{KT90} where it was stated to take the form
 \begin{equation}\label{eq:GMlongexact}
       \ldots \rightarrow \Omega^{\Spin}_{k-1}\left((B\O_2)^{\sigma-1}\right) \rightarrow \Omega^{\Pin^-}_k \rightarrow \Omega^{\mathrm{GM}}_k\xlongrightarrow{R} \Omega^{\Spin}_{k-2}\left((B\O_2)^{\sigma-1}\right)\rightarrow \Omega^{\Pin^-}_{k-1}\rightarrow \ldots\,
    \end{equation}
However, this is not exact as can be seen by tensoring with $\mathbb{Q}$. What the sequence in \eqref{eq:smithSpinO2} offers is an approximation to such a sequence.
\end{rem}
\begin{rem}[Spin-$\O_2$ bordism in higher degrees]
It is possible to completely solve spin-$\O_2$ bordism in all degrees by borrowing a strategy used by Buchanan--McKean~\cite{BM23} and Mills~\cite{Mil24} to compute spin\textsuperscript{$h$} bordism groups. As input, one shows that $H^*((B\O_2)^{V + 3\det(V) -5};\Z/2)$ is a direct sum of a copy of $\Sigma^{8k}R_2 \oplus \Sigma^{8k+4}Q$ for all $k\ge 0$, together with a free $\cA(1)$-module with degreewise finite rank. Combining this with Anderson--Brown--Peterson's determination of $H^*(\MTSpin;\Z/2)$~\cite{ABP67}, one has a complete description of $H^*(\mathit{MT}(\Spin\text{-}\O_2);\Z/2)$ and can apply a theorem of Stolz~\cite[Theorem 4.1]{Sto94} to decompose $\mathit{MT}(\Spin\text{-}\O_2)$ into well-understood spectra and therefore compute spin-$\O_2$ bordism groups in arbitrarily high degrees.

Following Anderson--Brown--Peterson~\cite{ABP66, ABP67} and Buchanan--McKean~\cite{BM23}, one can show that this implies there exists a collection of twisted $\mathit{KO}$- and $\mathit{KSp}$-cohomology characteristic classes such that the corresponding characteristic numbers, together with Stiefel-Whitney numbers, detect spin-$\O_2$ bordism. For example, for $k = 0,1,2,4$, $\Omega_k^{\Spin\text{-}\O_2}$ is detected by twisted $\mathit{KO}$- and $\mathit{KSp}$-characteristic numbers corresponding to the $\textcolor{BrickRed}{R_2}$ and $\textcolor{MidnightBlue}{\Sigma^4 Q}$ summands in \cref{fig:spinO2}, respectively; $\Omega_3^{\Spin\text{-}\O_2}$ is detected by the Stiefel-Whitney number $(M,E)\mapsto \int_M w_1(E)^3$, where $E\to M$ is the $\O_2$-bundle associated to the spin-$\O_2$ structure.
\end{rem}

%%%%%%%%%%%%%%%%%%%%%%%%%%%%%%%%%%%%%%%%%
\subsection{\texorpdfstring{Kirby--Taylor$^{-}$ Characteristic Bordism}{}}\label{subsection:KTcharBord}
%%%%%%%%%%%%%%%%%%%%%%%%%%%%%%%%%%%%%%%%%
In this subsection, we study the bordism groups with respect to the characteristic structure introduced in \cref{def:KT}. %\matt{Let me know what you think of this remark. We can take it out of the remark environment if that looks better.} \vic{I think the remark is great. But I want to put it at the very end of this sub-section}

The obstruction to finding a submanifold $i:F\hookrightarrow M$ with respect to the $\mathrm{KT}^{\pm}$ structures, by \cref{construction:PT}, corresponds to the obstruction to finding a lift $f$ in the following diagram,  
\begin{equation}
    \begin{tikzcd}
        & &M\O_2 \arrow[d, "U"]\\
        B\O \arrow[rr,"\mathcal{P}^{\pm}",swap ] \arrow[urr,dotted,"f"] & & K(\Z/2,2)\,.
    \end{tikzcd}
\end{equation}
where $\mathcal{P}^-=w_2(TM)+w_1(TM)^2$, $\cP^+=w_2(TM)$ and $U$ is the Thom class for $M\O_2$. Kirby--Taylor characteristic pairs fit into the following pullback square:
\begin{equation}\label{eq:KTpullback}
    \begin{tikzcd}[column sep=2cm, row sep=2cm]
        B\mathrm{KT}^{\pm} \arrow[r] \arrow[d] & M\O_2 \arrow[d,"U"]\\
        B\O \arrow[r,"\mathcal{P}^{\pm}", swap] & K(\Z/2,2)\,.
        \arrow["\lrcorner"{anchor=center, pos=0.125}, draw=none, from=1-1, to=2-2]
    \end{tikzcd}
\end{equation}
The data of the map consist of three pieces:
\begin{itemize}
\item The tangent bundle of $M$,
\item A map $M\rightarrow M\O_2$, which gives the submanifold $F$,
\item The identification $\mathcal{P}=U$, such that we can (and do) choose a pin$^\pm$ structure on $TM\oplus U$. 
\end{itemize}
This is the twisted pin$^\pm$ structure twisted by $U$. 
\begin{prop}\label{prop:KT-bord}
    The tangential structure of $\mathrm{KT}^{-}$-characteristic pairs is equivalent to an $(M\O_2,U)$-twisted pin$^{-}$-structure.\footnote{Analogously to \cref{nonVB}, an $(X, b)$-twisted \pinm structure on a vector bundle $E\to M$ is the data of a map $f\colon E\to X$ and a trivialization of $w_2(E) + w_1(E)^2 + f^*(b)$.}
\end{prop}

For the purpose of computation, we can translate these data into the following twisted spin structure. We first unfold the computations in the case where $\mathcal{P}=w_2(TM)+w_1(TM)^2$. As a twisted spin structure, the diagram in \eqref{eq:KTpullback} can equivalently be expressed as the pullback square:
\begin{equation}\label{eq:KTtwistedspin}
    \begin{tikzcd}[column sep=2cm, row sep=2cm]
 B\mathrm{KT}^{-} \arrow[r,"\sigma"] \arrow[d] \arrow[dr, phantom, "\lrcorner", very near start] & M\O_2 \times B\O_1 \arrow[d, "{(U+w_1^2(\sigma),w_1(\sigma))}"] \\
B\O \arrow[r, "{(w_2,w_1)}",swap] & K(\Z/2,2) \times K(\Z/2,1)\,,
\end{tikzcd}
\end{equation}
The data of KT$^-$ characteristic pairs thus consists of the following:
\begin{itemize}
    \item a map to $M\O_2$, and a principal $\Z/2$-bundle $\sigma$ which is the tautological bundle over $B\O_1 \cong \mathbb{R}P^\infty$,
    \item a map to $B\O$ given by a vector bundle $V$,
    \item an identification $w_1(V) = w_1(\sigma)$\,,
    \item an identification $U= w_2(V)+w_1(\sigma)^2$, i.e. $U=w_2(V)+w_1(V)^2$. 
\end{itemize}

Along the manifold $F$ we have the decomposition $TM|_F= TF \oplus \nu_F$  where the normal bundle $\nu_F$ is a rank $2$ vector bundle and $w_2(TM) = w_2(\nu_F)+w_1(\nu_F)^2$ by \cref{cor:atoN}, and $w_1(TM) = w_1(\sigma)$ since $M$ was unoriented. Then by the Whitney sum formula we obtain the tangential structure on $F$:
%We can further pull back the diagram in \eqref{eq:KTtwistedspin} to the manifold $F$ and obtain the tangential structure there as the following pullback:
%\begin{equation}
  %  \begin{tikzcd}[column sep=2cm, row sep=2cm]
 %BF \arrow[r,"{(V,\sigma)}"] \arrow[d] \arrow[dr, phantom, "\lrcorner", very near start] & B\O_2 \times B\O_1  \arrow[d, "{(z,\id)}"] \\
 %B\mathrm{KT}^{-} \arrow[r, "\sigma"] \arrow[d] & M\O_2 \times B\O_1  \arrow[d,"{(U+w_1^2(\sigma),w_1(\sigma))}"] \%\
%B\O \arrow[r,"{(w_2,w_1)}", swap]& K(\Z/2,2)\times K(\Z/2,1)\,.  
%\end{tikzcd}
%\end{equation}

%The full twisting on $(w_1(TF),w_2(TF))$ is computed by tracing the arrows along the top and right side of the diagram above.

\begin{lem}\label{KT_minus_F_str}
 Let $V$ be a $B\O_2$-bundle and $\sigma$ be a $B\O_1$-bundle.  The submanifold $F$ in a KT$^-$ characteristic pair has a twisted spin structure with twistings given by:
   \begin{equation}\label{eq:KTonFminus}
   \begin{aligned}
    w_1(TF) &= w_1(\sigma)+w_1(V)\,,\\
    w_2(TF) &= w_1(\sigma)^2+w_1(\sigma)w_1(V)+w_1(V)^2\,.
\end{aligned}
\end{equation}
\end{lem}

By unwinding the definitions of the twisted tangential structure in \cref{prop:KT-bord} allows us to write a KT$^-$ structure as a $(B\O_1\times M\O_2,w_1 ,U)$-twisted spin structure. Thus we want to compute the low-degree homotopy groups of the corresponding $\ko$-module Thom spectrum $M^\ko(\widehat A\circ f_{w_1, U})$.

The twist $\widehat A\circ f_{w_1, U}\colon B\O_1\times M\O_2\to B\GL_1(\ko)$ is an external sum of the twists $\widehat A\circ f_{w_1, 0}\colon B\O_1\to B\GL_1(\ko)$ and $\widehat A\circ f_{0,U}\colon M\O_2\to B\GL_1(\ko)$. Beardsley~\cite[Theorem 1]{Bea17} has shown that the Thom spectrum functor is symmetric monoidal, implying an equivalence
\begin{equation}\label{smash_split}
    M^\ko(\widehat A\circ f_{w_1, U}) \simeq M^\ko(\widehat A\circ f_{w_1,0})\wedge_\ko M^\ko(\widehat A\circ f_{0,U}).
\end{equation}
This is more helpful than it looks: we have a Künneth formula for $\ko$-modules $M$ and $N$ of the form
\begin{equation}
    H_\ko^*(M\wedge_\ko N) \cong H_\ko^*(M)\otimes_{\Z/2} H_\ko^*(N),
\end{equation}
which is an $\cA(1)$-module isomorphism if we give the right-hand side the $\cA(1)$-module structure satisfying the Cartan formula
\begin{equation}
    \Sq^k (x \otimes y) = \sum_{i+j=k} \Sq^i(x) \otimes \Sq^j(y)\,.
\end{equation}
\begin{sloppypar}
Moreover, both of the factors on the right-hand side of~\eqref{smash_split} are familiar: we computed $H_\ko^*(M^\ko(\widehat A\circ\allowbreak f_{0,U}))$ in \cref{GMA1}. If $\sigma\to B\O_1$ denotes the tautological bundle, then $(w_1, 0) = (w_1(\sigma), w_2(\sigma))$, so $M^\ko (\widehat A\circ\allowbreak f_{w_1,0})\simeq \ko\wedge (B\O_1)^{\sigma-1}$. Thus we can determine the $E_2$-page of the Baker--\allowbreak Lazarev Adams spectral sequence in the range we need by tensoring together these two $\cA(1)$-modules. We also compute the tensor product for $(B\O_1)^{3\sigma-3}$ in place of $(B\O_1)^{\sigma-1}$, as we will need it later.
\end{sloppypar}

We draw the most relevant submodules of $H_\ko^*(M^\ko(\widehat A\circ f_{0,U});\Z/2)$ and $H^*((B\O_1)^{\sigma-1};\Z/2)$ involved in the tensor product in Figure \ref{fig:tensorforKT-}.
\begin{figure}[!ht]
\hspace{5mm}
\begin{subfigure}[c]{0.5\textwidth}
	\begin{tikzpicture}[scale=0.8, every node/.style = {font=\tiny}]
		\foreach \y in {0, 1, ...,5} {
			\node at (-2, \y) {$\y$};
		}
 \begin{scope}[MidnightBlue]
 \tikzpt{0}{0}{$Q$}{};
         \tikzpt{0}{2}{}{};
          \tikzpt{0}{3}{}{};
          \tikzpt{0}{4}{}{};
        \tikzpt{0}{5}{}{};
          \sqtwoL(0,0);
              \sqtwoL(0,3);
              \sqone(0, 2);
               \sqone(0, 4);
               \begin{scope}
                \clip (-1, 4) rectangle (3, 5.6);
                \sqtwoR(0, 4);
                \end{scope}
		%\Aone{3}{2}{$Qt^2$};
	\end{scope}
	\end{tikzpicture}
\end{subfigure}
\qquad
\begin{subfigure}[c]{0.4\textwidth}
	\begin{tikzpicture}[scale=0.8, every node/.style = {font=\tiny}]
		\foreach \y in {0, 1, ...,5} {
			\node at (-2, \y) {$\y$};
		}
  \begin{scope}[Green]
		\tikzpt{0}{0}{}{};
         \tikzpt{0}{1}{}{};
         \tikzpt{0}{2}{}{};
         \tikzpt{0}{3}{}{};
         \tikzpt{0}{4}{}{};
          \tikzpt{0}{5}{}{};
         \sqone(0, 0);
         \sqtwoL(0, 1);
          \sqone(0, 2);
           \sqtwoR(0, 2);
           \sqone(0, 4);
           \begin{scope}
                \clip (-1, 4) rectangle (3, 5.6);
                \sqtwoL(0, 5);
                \end{scope}
	\end{scope}
	\end{tikzpicture}
\end{subfigure}
\caption{Left: The $\cA(1)$-module $M_1$ (\cref{M1defn}). In \cref{GMA1} we showed that, modulo classes of degree at least $7$, $M_1$ is a summand of $H_\ko^*(M^\ko (\widehat A\circ f_{0,U}))$. Right: the $\cA(1)$-module structure on $H^*((B\O_1)^{\sigma-1};\Z/2)$.}
    \label{fig:tensorforKT-}
\end{figure}
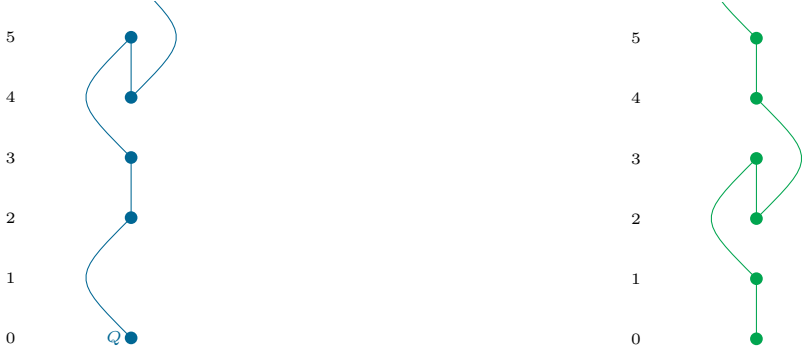

\begin{prop}
There are isomorphisms of $\cA(1)$-modules
\begin{equation}
    \begin{aligned}
        M_1\otimes H^*((B\O_1)^{\sigma-1};\Z/2) &\cong M_1 \otimes H^*((B\O_1)^{3\sigma-3}; \Z/2)\\
        &\cong \cA(1) \oplus \Sigma^2\cA(1) \oplus \Sigma^4 \cA(1) \oplus \Sigma^4\cA(1) \oplus P,
    \end{aligned}
\end{equation}
where $P$ is concentrated in degrees $5$ and above.
\end{prop}
Beware: the two isomorphisms in the $\sigma-1$ and $3\sigma-3$ cases are \emph{not} the same maps of underlying vector spaces!

Since $M_0\cong_{\le 3} M_1$, this also gives us information on $\Sigma^4M_0\otimes H^*((B\O_1)^{\sigma-1};\Z/2)$ (and similarly for $(B\O_1)^{3\sigma-3}$) in degrees $4$ and below. Thus:
\begin{prop}\label{just_free}
There are (different!) $\cA(1)$-module isomorphisms
\begin{subequations}
\begin{align}
    H_\ko^*(M^\ko(\widehat A\circ f_{0,U})\wedge (B\O_1)^{\sigma-1}) &\cong \textcolor{BrickRed}{\cA(1)} \oplus
    \textcolor{MidnightBlue}{\Sigma^2 \cA(1)} \oplus
    \textcolor{Green}{\Sigma^4\cA(1)} \oplus
    \textcolor{Fuchsia}{\Sigma^4\cA(1)} \oplus
    \textcolor{Orange}{\Sigma^4\cA(1)} \oplus P^+\\
    H_\ko^*(M^\ko(\widehat A\circ f_{0,U})\wedge (B\O_1)^{3\sigma-3}) &\cong \textcolor{BrickRed}{\cA(1)} \oplus
    \textcolor{MidnightBlue}{\Sigma^2 \cA(1)} \oplus
    \textcolor{Green}{\Sigma^4\cA(1)} \oplus
    \textcolor{Fuchsia}{\Sigma^4\cA(1)} \oplus
    \textcolor{Orange}{\Sigma^4\cA(1)} \oplus P^-,
\end{align}
\end{subequations}
for $\cA(1)$-modules $P^\pm$ concentrated in degrees $5$ and above.
\end{prop}
We draw these isomorphisms in \cref{fig:KT-} (left) and \cref{fig:KT+}.

Margolis' theorem~\cite{Mar74} implies that the respective Adams spectral sequences collapse without extension questions in the degrees we care about. Thus:
%The $\cA(1)$-modules of the tensor product of the two modules in Figure \ref{fig:tensorforKT-}, as well as the $E_2$-page of the Adams spectral sequence are computed in Figure \ref{fig:KT-}. From this we see:
\begin{prop}\label{prop:bordcompKT-}
    There are isomorphisms
 \begin{equation}
\begin{aligned}
    \Omega_0^{\mathrm{KT}^-} &\cong \Z/2\\
    \Omega_1^{\mathrm{KT}^-} &\cong 0\\
    \Omega_2^{\mathrm{KT}^-} &\cong \Z/2\\
    \Omega_3^{\mathrm{KT}^-} &\cong 0\\
    \Omega_4^{\mathrm{KT}^-} &\cong (\Z/2)^{\oplus 3}.
\end{aligned}
 \end{equation}
    %$\Omega^{\mathrm{KT}^-}_* = \{\Z/2, 0, \Z/2,0,\Z/2 \oplus \Z/2 \oplus \Z/2\}$ for $*= 0, \ldots,4$.
\end{prop}

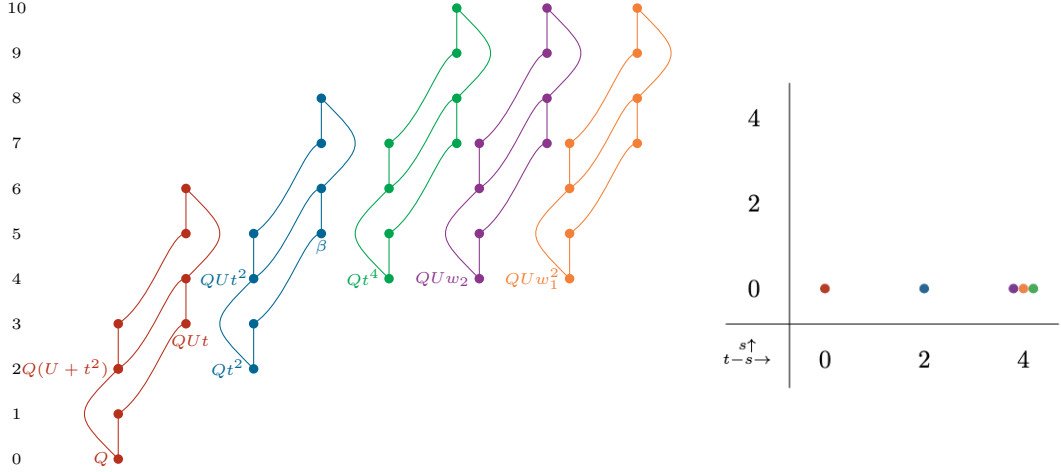
\begin{figure}[!ht]
\begin{subfigure}[c]{0.6\textwidth}
	\begin{tikzpicture}[scale=0.54, every node/.style = {font=\tiny}]
		\foreach \y in {0, 1, ...,10} {
			\node at (-2.25, \y) {$\y$};
		}
		\begin{scope}[BrickRed]
		\Aone{0}{0}{$Q$};
         \tikzpt{0}{2}{$Q(U+t^2)$}{};
          \node at (1.59,2.6) {$QUt$};
	\end{scope}
 \begin{scope}[MidnightBlue]
        \tikzpt{3}{2}{$Qt^2$}{};
         \tikzpt{3}{3}{}{};
          \tikzpt{3}{4}{}{};
          \tikzpt{3}{5}{}{};
        \tikzpt{4.5}{5}{}{};
        \tikzpt{4.5}{6}{}{};
        \tikzpt{4.5}{7}{}{};
        \tikzpt{4.5}{8}{}{};
         \sqone(3, 2);
         \sqone(3, 4);
         \sqone(4.5, 5);
         \sqone(4.5, 7);
        \sqtwoL(3,2)
        \sqtwoCR(3, 3);
         \sqtwoCR(3, 4);
         \sqtwoCR(3, 5);
         \sqtwoR(4.5, 6);
         \node[left=-1em] at (2.5,4) {$QUt^2$};
         \node at (4.5,4.7) {$\beta$};
		%\Aone{3}{2}{$Qt^2$};
	\end{scope}
  \begin{scope}[Green]
    \Aone{6}{4}{$Qt^4$};
%		\tikzpt{6}{4}{}{};
%         \tikzpt{6}{5}{}{};
%         \tikzpt{6}{6}{}{};
%        \draw (6,3.7) node {$Qt^4$};
%         \sqone(6, 4);
%         \sqtwoL(6, 4);
	\end{scope}
 \begin{scope}[Fuchsia]
    \Aone{8}{4}{$QUw_2$};
%		\tikzpt{7.5}{4}{}{};
%  \tikzpt{7.5}{5}{}{};
%  \tikzpt{7.5}{6}{}{};
%       \draw (7.5,3.7) node {$QUw_2$};
%        \sqone(7.5, 4);
%        \sqtwoL(7.5, 4);
	\end{scope}
     \begin{scope}[Orange]
        \Aone{10}{4}{$QUw_1^2$};
%     		\tikzpt{9}{4}{}{};
%  \tikzpt{9}{5}{}{};
%  \tikzpt{9}{6}{}{};
%       \draw (9,3.7) node {$QUw_1^2$};
%        \sqone(9, 4);
%        \sqtwoL(9, 4);
	\end{scope}
	\end{tikzpicture}
\end{subfigure}
\hfill
%\hspace{46mm}
\begin{subfigure}[!ht]{0.33\textwidth}
\includegraphics[width=45mm]{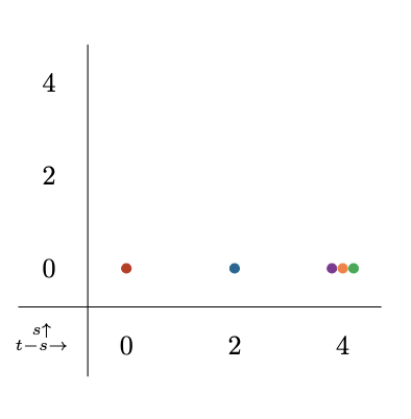}
%\begin{sseqdata}[name= KT, Adams grading, classes=fill, xrange={0}{4}, yrange={0}{4},xscale=0.7, yscale=0.6, x tick
%step=2, y tick step=2, class labels = { left = 0.07em, font=\small },
% class pattern = linear, class placement transform = { rotate = 45 },
 % x label = {$\displaystyle{\substack{s\uparrow \\ t-s\rightarrow}}$},
 % x label style = {font = \small, xshift = -16.5ex, yshift=1.7ex}]]
%\begin{scope}[BrickRed]
	%\class(0, 0)
%\end{scope}
%\begin{scope}[MidnightBlue]
	%\class(2, 0)
%\end{scope}
%\begin{scope}[Fuchsia]
	%\class(4, 0)
%\end{scope}
%\class[Orange](4, 0)
%\node at (4.2,0){\phantom{a}};
%\class[draw=none, fill=none](4, 0)
%\begin{scope}[Green]
%	\class(4, 0)
%\end{scope}
%\end{sseqdata}
%\printpage[name=KT]
\end{subfigure}
\caption{Left: The $\cA(1)$-modules that arise from tensoring the two $\cA(1)$-modules in Figure \ref{fig:tensorforKT-}. The class $\beta = Q(Ut^3+t^5)$. Right: The $E_2$-page of the Adams spectral sequence computing the nontrivial homotopy groups in degrees less than 5.} %\arun{We should fill in the rest of each of the summands we drew part of, so the reader is not confused whether they're free}}
    \label{fig:KT-}
\end{figure}

\begin{rem}
    \sloppy
    Analogous to GM-\allowbreak characteristic bordism in \S\ref{subsection:GMcharBord} the computations for KT$^-$-\allowbreak characteristic bordism agree with $\Omega^\O_k$ for $k=0,1,2,3$. Compare~\cite[Theorem 7.1]{KT90}.
\end{rem}

We now remark on the physics of the KT$^-$ pair. A similar interpretation works for the KT$^+$ pair. Upon inserting a defect along an unoriented submanifold $F$ in $M$, the manifold $M\backslash F$ has a pin$^-$ structure due to Lemma \ref{lem:restricta}. Therefore a theory outside of the defect, defined on $M\backslash F$, can support fermions and a time-reversal symmetry with ${T}^2=1$, but those fermions do not extend to $F$. In electron systems, we will interpret a defect as a local deformity of the background lattice, but there is nothing obstructing the electrons from existing on the defect. If the global structure of the lattice is proliferated with defects wrapped on unorientable manifolds such as the type of $F$, then the situation can be potentially modeled with a  KT$^-$ pair in the IR, where the defects collectively behave as a codimension 2-submanifold. If the global structure of the defect makes it Poincaré dual to $w_2(TM)+w_1(TM)^2$ then the electrons now cannot restrict onto the defects!

Now we will draw out the (conjectural) characteristic long exact sequence for KT$^-$ bordism, analogously to how we proved \cref{GM_LES_thm} for Guillou--Marin bordism. The first step is to compute the bordism groups for the symmetry type of the characteristic submanifold $F$; as we showed in \cref{KT_minus_F_str}, this is a $(B\O_1\times B\O_2; a+b, a^2+ab+b^2)$-twisted spin structure, where we let $a\coloneqq w_1\in H^1(B\O_1;\Z/2)$, $b\coloneqq w_1\in H^1(B\O_2;\Z/2)$, and $c \coloneqq w_2\in H^2(B\O_2;\Z/2)$. Let $\tau^-\coloneqq (a+b, a^2+ab+b^2)$.

For notational convenience, let $M\tau^-\coloneqq M^\MTSpin f_{a+b, a^2+ab+b^2}$ (which we defined, along with its approximation $M^\ko (\widehat A\circ f_{a+b, a^2+ab+b^2})$, in \cref{item:approx}). Thus the homotopy groups of $M\tau^-$ are isomorphic to the $(B\O_1\times B\O_2, \tau^-)$-twisted spin bordism groups.\footnote{It is possible to find a vector bundle $W\to B\O_1\times B\O_2$ realizing this twist, but we do not need that, so do not bother.}\textsuperscript{,}\footnote{See also Hertl~\cite{Her17} for additional calculations of twisted spin bordism over $B\O_1\times B\O_2$.}
\begin{lem}\label{klein_split_off}
Let $MV(a+b, a^2+ab+b^2)$ denote the $\ko$-module Thom spectrum defined in the same way as $M^\ko (\widehat A\circ f_{a+b, a^2+ab+b^2})$, but over $B\O_1\times B\O_1$ rather than $B\O_1\times B\O_2$. Then the standard inclusion $\O_1\times\O_1\hookrightarrow \O_1\times\O_2$ induces the inclusion of the first direct summand in a $\ko$-module splitting
\begin{equation}
    M^\ko(\widehat A\circ f_{a+b, a^2+ab+b^2}) \overset\simeq\longrightarrow MV(a+b, a^2+ab+b^2)\vee M
\end{equation}
for some $\ko$-module spectrum $M$ such that the induced map $H_\ko^*(M)\to H_\ko^*(M^\ko(\widehat A\circ f_{a+b, a^2+ab+b^2}))$ is the inclusion of a subspace complementary to that spanned by $\set{Ua^ib^j: i,j\ge 0}$.
\end{lem}
\begin{sloppypar}
The proof is completely analogous to that of~\cite[Lemma 3.30]{Deb21}: the map $M^\ko (\widehat A\circ\allowbreak f_{a+b, a^2+ab+b^2})\to\allowbreak MV(a+b, a^2+ab+b^2)$ induced by $(\id, \det)\colon B\O_1\times B\O_2\to\allowbreak B\O_1\times B\O_1$ provides a section to the inclusion map.
\end{sloppypar}

The matrix $\begin{pmatrix}1 & 0\\1 & 1\end{pmatrix}$ is an element of $\GL_2(\mathbb F_2) = \mathrm{Aut}(\O_1\times\O_1)$; the induced map on cohomology sends $a\mapsto a+b$ and $b\mapsto b$, so it sends $a+b\mapsto a$ and $a^2+ab+b^2\mapsto a^2+ab+b^2$. Thus we have a $\ko$-module equivalence $MV(a+b, a^2+ab+b^2) \simeq MV(a, a^2+ab+b^2)$. The low-degree homotopy groups of $MV(a, a^2+ab+b^2)$ are calculated in~\cite[Theorem 4.26, case $n = 2$]{Deb21}; see also~\cite{BottSpiral} for partial information in higher degrees. Thus, we will focus on the complementary summand $M$, and add on the $(B\O_1\times B\O_1, a+b, a^2+ab+b^2)$-twisted spin bordism groups later.
\begin{lem}\label{non_klein_calc}
There is an $\cA(1)$-module isomorphism
\begin{equation}
    H_\ko^*(M)\cong \Sigma^2\cA(1) \oplus \Sigma^3\cA(1) \oplus P,
\end{equation}
where $P$ is concentrated in degrees $4$ and above. The two listed free summands are spanned by $Uc$, resp.\ $Ubc$.
\end{lem}
\begin{cor}\label{KT_minus_2nd}
The low-degree $(B\O_1\times B\O_2, a+b, a^2+ab+b^2)$-twisted spin bordism groups are:
\begin{equation}
\begin{aligned}
    \Omega_0^\Spin(B\O_1\times B\O_2, \tau^-) &\cong \Z/2\\
    \Omega_1^\Spin(B\O_1\times B\O_2, \tau^-) &\cong \Z/2\\
    \Omega_2^\Spin(B\O_1\times B\O_2, \tau^-) &\cong \Z/2\oplus \Z/2\\
    \Omega_3^\Spin(B\O_1\times B\O_2, \tau^-) &\cong \Z/2 \oplus \Z/2.
\end{aligned}
\end{equation}
The following are complete invariants:
\begin{subequations}
\begin{align}
    \int b\colon &\Omega_1^\Spin(B\O_1\times B\O_2, \tau^-) \longrightarrow \Z/2\\
    \int c, \int b^2\colon &\Omega_2^\Spin(B\O_1\times B\O_2, \tau^-) \longrightarrow \Z/2\oplus \Z/2.
    %\int bc, \mathrm{sm}\colon &\Omega_3(B\O_1\times B\O_2, a+b, a^2+ab+b^2) \longrightarrow \Z/2\oplus \Z/2,
\end{align}
\end{subequations}
\end{cor}
\begin{proof}
This is a combination of the bordism groups from~\cite[Theorem 4.26]{Deb21} (for the $\O_1\times \O_1$ part) and Margolis' theorem, which tells us that $M$ is equivalent to $\Sigma^2 H\Z/2\vee \Sigma^3 H\Z/2\vee M'$ for a $3$-connected spectrum $M'$. Integrals of mod $2$ cohomology classes appearing as bordism invariants occur for classes whose images in the Adams spectral sequence are in filtration $0$: see~\cite[Figure 1, right]{Deb21}, and free $\cA(1)$-modules always meet this criterion.
\end{proof}
\begin{rem}
We do not need complete invariants in dimension $3$, but one can show that they are given by $\int bc$ together with the Smith homomorphism associated to the determinant line bundle of the $\O_2$-bundle, which lands in $\Omega_2^{\Pin^+}(B\O_2)$; then project onto $\Omega_2^{\Pin^+}\cong\Z/2$.
\end{rem}
Using \cref{prop:bordcompKT-,KT_minus_2nd}, together with the characteristic long exact sequence, we can compute some homotopy groups of the fiber $F_{\mathrm{KT}^-}$. Unfortunately, we learn less than we did for Guillou--Marin characteristic bordism.
\begin{thm}
\label{KT-_LES_thm}
Assume \cref{char_conj} is true, so that there is a characteristic map of spectra $\mathit{MTKT}^-\to \Sigma^2 M^\MTSpin \tau^-$, and let $F_{\mathrm{KT}^-}$ denote the fiber of this map. Then there is an abelian group $A_-$ of order $4$ and a nonzero abelian group $B_-$ such that
\begin{equation}
\begin{aligned}
    \pi_0(F_{\mathrm{KT}^-}) &\cong \Z/2\\
    \pi_1(F_{\mathrm{KT}^-}) &\cong \Z/2\\
    \pi_2(F_{\mathrm{KT}^-}) &\cong A_-\\
    \pi_3(F_{\mathrm{KT}^-}) &\cong 0\\
    \pi_4(F_{\mathrm{KT}^-}) &\cong B_-,
\end{aligned}
\end{equation}
and the characteristic long exact sequence is as given in \cref{KT-_char_LES}.
\end{thm}
\begin{proof}
Looking at \cref{KT-_char_LES}, most of the theorem follows from our prior bordism computations in \cref{prop:bordcompKT-,KT_minus_2nd} along with exactness -- \emph{as soon as we know $R\colon \Omega_2^{\mathrm{KT}^-}\to \Omega_0^\Spin(B\O_2\times B\O_1, \tau^-)$ vanishes}. We will show this by computing on a generator.

The pair $(\RP^2, \varnothing)$ has a $\mathrm{KT}^-$ characteristic structure, because $w_2(\RP^2) + w_1(\RP^2)^2 = 0$ (in fact this is true for all closed $2$-manifolds by the Wu formula), and $\varnothing$ is Poincaré dual to $0$. This pair represents the nonzero element of $\Omega_2^{\mathrm{KT}^-}$ because it is detected by $(M,F)\mapsto \int w_2(M)$. However, $R$ sends $(\RP^2, \varnothing)\mapsto\varnothing$, which (vacuously) bounds. Thus, as claimed, the characteristic map vanishes on $\Omega_2^{\mathrm{KT}^-}$.
\end{proof}

\begin{figure}[!ht]
    \centering
\begin{tikzcd}
	{k} & {\pi_k(F_{\mathrm{KT}^-})} & {\Omega_k^{\mathrm{KT}^-}} & {\Omega_{k-2}^{\Spin}(B\O_1\times B\O_2, \tau^-)} \\
	0 & \Z/2 \arrow[r,"\cong"]& {\Z/2} & 0 \\
	1 & {\Z/2} & 0 & 0 \\
	2 & {A_-}\ar[r, ->>]  & {\Z/2}\arrow[r, "0"] & {\Z/2}\arrow[from=4-4,to=3-2,in=-150, out=30,"\cong",swap] \\
	3 & {0} & 0 & {\Z/2} \arrow[from=5-4,to=4-2,in=-150, out=30,hookrightarrow,swap]\\
	4 & {B_-} \arrow[r]
    & {(\Z/2)^{\oplus 3}}\arrow[r, ->>] & {(\Z/2)^{\oplus 2}}
\end{tikzcd}
    \caption{The characteristic long exact sequence for the $\mathrm{KT}^-$ characteristic structure, which we prove in \cref{KT-_LES_thm}. Here $A_-$ is either $\Z/4$ or $\Z/2\oplus\Z/2$ and $B_-$ is some nonzero abelian group.}
    \label{KT-_char_LES}
\end{figure}

Just as we did for Guillou--Marin bordism, we can apply Anderson duality to interpret the characteristic long exact sequence in terms of anomalies. As all of the bordism groups we found are torsion, this is simpler, but means that all anomaly-theoretic information in this case is invisible to perturbative computations.
\begin{cor}
\label{KT_minus_physics}
Assuming \cref{char_conj}, consider a $k$-dimensional field theory on manifolds with $\mathrm{KT}^-$ structure. For $k = 0$, $1$, and $2$, the defect map $\mho_\Spin^{k-1}(B\O_1\times B\O_2, \tau^-)\to\mho_{\mathrm{KT}^-}^{k+1}$ is $0$. For $k= 3$, the map is nonzero and there is a $B_-$-valued obstruction to pulling the anomaly of the theory back to the defect.
\end{cor}
It would be interesting to determine the precise structure of $B_-$ and compute this obstruction in examples.

\begin{rem}\label{rem:differentfromKT}
    In the work of \cite{KT90} where the $\mathrm{KT}^-$ characteristic pairs first appeared, the authors demanded a map from the characteristic structures on $(M,F)$ to the space of \pinm structures on $F$. In particular, the authors were explicitly using the fact that $F$ was a pin$^-$ manifold in the computations of their characteristic bordism group denoted $\Omega^{!}_*$. The important distinction that we would like to draw here is that the condition of Definition \ref{def:KT} does not imply that $F$ has to have a \pinm structure necessarily. In fact, \cite[Lemma 6.7]{KT90} gives the full condition for when $F$ is \pinm. Since the physical context in which we interpret characteristic structures does not require the choice of tangential structure on $F$, we will continue forward using Definition \ref{def:KT}, and hence take a different approach than Kirby and Taylor. This is also reflected in the difference of our bordism computations versus Kirby and Taylor in degrees 3 and 4, see \cite[Theorem 7.2]{KT90} and
    \cite[Theorem 7.3]{KT90} respectively for the results. In this paper, we do not encode the condition that Kirby and Taylor requires to give $F$ a \pinm structure in a homotopical manner.
\end{rem}

%%%%%%%%%%%%%%%%%%%%%%%%%%%%%%%%%%%%%%%%%%%%%%%%%%%%%%%%
\subsection{\texorpdfstring{Kirby--Taylor$^{+}$ Characteristic Bordism}{}}\label{subsection:pin+KT}
%%%%%%%%%%%%%%%%%%%%%%%%%%%%%%%%%%%%%%%%%%%%%%%%%%%%%%%%
Similar to the case of $B\mathrm{KT}^-$ from \S\ref{subsection:KTcharBord}, we can construct the pullback for $B\mathrm{KT}^+$ where  the manifold $F$ is Poincaré dual to $w_2(TM)$. This gives
\begin{equation}\label{eq:KT+pullback}
    \begin{tikzcd}[column sep=2cm, row sep=2cm]
 B\mathrm{KT}^{+} \arrow[r,"\sigma"] \arrow[d] \arrow[dr, phantom, "\lrcorner", very near start] & M\O_2 \times B\O_1 \arrow[d, "{(U,w_1(\sigma))}"] \\
B\O \arrow[r, "{(w_2,w_1)}",swap] & K(\Z/2,2) \times K(\Z/2,1)\,,
\end{tikzcd}
\end{equation}
which implies that the data of a $\mathrm{KT}^+$ structure consists of an identification $w_2(V)=U$ and $w_1(V)=w_1(\sigma)$.
\begin{prop}\label{prop:KT+structures}
    The  structure of $\mathrm{KT}^{+}$-characteristic pairs is equivalent to $(M\O_2,U)$-twisted pin$^{+}$ structures.
\end{prop}

The structure on the submanifold $F$ can be determined analogous to the KT$^-$ case in \cref{KT_minus_F_str}.
\begin{lem}\label{KT_plus_F_str}
 Let $V$ be a $B\O_2$-bundle and $\sigma$ be a $B\O_1$-bundle.  The submanifold $F$ in a KT$^+$ characteristic pair has a twisted spin structure with twistings given by:
   \begin{align}\label{eq:KTonFplus}
    w_1(TF) &= w_1(\sigma)+w_1(V)\,,\\ \notag 
    w_2(TF) &= w_1(V)^2+w_1(\sigma)w_1(V)\,.
\end{align}
\end{lem}
We can express a $(M\O_2,U)$-twisted pin$^{+}$ structure as a $(B\O_1\times M\O_2,w_1,w_1^2 + U)$-twisted spin structure. The computation determining the $E_2$-page of the Baker--Lazarev Adams spectral sequence for the corresponding $\ko$-module twist is completely analogous to our argument for $\mathrm{KT}^-$, splitting the
Thom spectrum $M^\ko(\widehat A\circ f_{w_1, w_1^2+U})$ as a smash product of $M^\ko(\widehat A\circ f_{0,U})$ and $\ko\wedge(B\O_1)^{3\sigma-3}$.
%By the shearing construction of Lemma \ref{lem:shearing} we obtain in the  the following spectrum:
%\begin{equation}\label{eq:KT+}
%   (M\O_2,3\sigma,U)\text{--} \text{twisted} \,\,
%   \text{spin} \simeq \MTSpin \wedge (M\O_2)^{U-2} \wedge (B\O_1)^{3\sigma-3}\,,
%\end{equation}
%and we compute the right-hand-side with the Adams spectral sequence, with $E_2$-page
%\begin{equation}
%    E^{s,t}_2 = \Ext^{s,t}_{\cA(1)}(H^*((M\O_2)^U \wedge (B\O_1)^{3\sigma-3};\Z/2);\Z/2)\,,
%\end{equation}
 %Similar to the computation of $\mathrm{KT}^-$-bordism, we
Again we tensor together the corresponding two $\cA(1)$-modules in \cref{just_free}, then take Ext (which is trivial, because in the degrees we need, this $\cA(1)$-module is free).
%constructed from $H^*((M\O_2)^U;\Z/2)$ and $H^*((B\O_1)^{3\sigma-3};\Z/2)$ to compute the $\mathcal{A}(1)$-module structure for $H^*((M\O_2)^U \wedge (B\O_1)^{3\sigma};\Z/2)$.
The individual modules are drawn on the left in Figure \ref{fig:tensorforKT-} and in Figure \ref{fig:pin+module}, respectively. 
\begin{figure}[!ht]
\centering
	\begin{tikzpicture}[scale=0.8, every node/.style = {font=\tiny}]
		\foreach \y in {0, 1, ...,6,7} {
			\node at (-2, \y) {$\y$};
		}
  \begin{scope}[BrickRed]
		\tikzpt{0}{0}{$U$}{};
         \tikzpt{0}{1}{}{};
         \tikzpt{0}{2}{}{};
         \tikzpt{0}{3}{}{};
         \tikzpt{0}{4}{}{};
          \tikzpt{0}{6}{}{};
          \tikzpt{0}{5}{}{};
           \tikzpt{0}{7}{}{};
         \sqone(0, 0);
          \sqtwoL(0, 0);
           \sqtwoL(0, 3);
            \sqtwoR(0, 4);
          \sqone(0, 2);
           \sqone(0, 4);
            \sqone(0, 6);
            \begin{scope}
                \clip (-1, 6.5) rectangle (0.5, 7.5);
                \sqtwoL(0, 7);
            \end{scope}
	\end{scope}
	\end{tikzpicture}
\caption{The $\cA(1)$-module structure for $H^*((B\O_1)^{3\sigma-3};\Z/2)$. The corresponding twisted spin structure is equivalent to a \pinp structure. The Steenrod actions on the Thom class are given by $\Sq^1U = w_1(3\sigma)U$ and $\Sq^2 U = w_2(3\sigma) U$.}
    \label{fig:pin+module}
\end{figure}
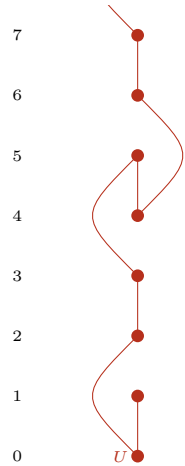
We compute the bordism groups in Figure \ref{fig:KT+}, which shows:
\begin{prop}\label{prop:bordcompKT+}
There are isomorphisms
\begin{equation}
\begin{aligned}
    \Omega_0^{\mathrm{KT}^+} &\cong \Z/2\\
    \Omega_1^{\mathrm{KT}^+} &\cong 0\\
    \Omega_2^{\mathrm{KT}^+} &\cong \Z/2\\
    \Omega_3^{\mathrm{KT}^+} &\cong 0\\
    \Omega_4^{\mathrm{KT}^+} &\cong (\Z/2)^{\oplus 3}.
\end{aligned}
 \end{equation}
\end{prop}

% note: I (AD) wanted to show that the A(1)-action does not stop in deg 6, while preserving the zoomed-in look we had here. Please feel free to undo my changes by uncommenting the stuff I commented out!
\begin{figure}[!ht]
\centering 
	\begin{tikzpicture}[scale=0.8, every node/.style = {font=\tiny}]
            \clip (-2.25, -0.25) rectangle (10.25, 6.5);
    
		\foreach \y in {0, 1, ...,6} {
			\node at (-2, \y) {$\y$};
		}
		\begin{scope}[BrickRed]
		\Aone{0}{0}{$Q$};
            \node[left] at (0, 2) {$QU$};
          \node at (1.51,2.7) {$\alpha$};
	\end{scope}
 \begin{scope}[MidnightBlue]
        \Aone{3}{2}{$Qt^2$};
         %\tikzpt{3}{3}{}{};
         % \tikzpt{3}{4}{}{};
         % \tikzpt{3}{5}{}{};
        %\tikzpt{4.5}{5}{}{};
        %\tikzpt{4.5}{6}{}{};
        % \sqone(3, 2);
        % \sqone(3, 4);
        % \sqone(4.5, 5);
        %\sqtwoL(3,2)
        %\sqtwoCR(3, 3);
        % \sqtwoCR(3, 4);
         \node at (2.7,4) {$\gamma$};
         \node at (4.5,4.7) {$Ut^3$};
		%\Aone{3}{2}{$Qt^2$};
	\end{scope}
  \begin{scope}[Green]
            \Aone{6}{4}{$Qt^4$};
		%\tikzpt{6}{4}{}{};
         %\tikzpt{6}{5}{}{};
         %\tikzpt{6}{6}{}{};
        %\draw (6,3.7) node {$Qt^4$};
        % \sqone(6, 4);
         %\sqtwoL(6, 4);
	\end{scope}
 \begin{scope}[Fuchsia]
        \Aone{7.5}{4}{$QUw_2$};
		%\tikzpt{7.5}{4}{}{};
 % \tikzpt{7.5}{5}{}{};
 % \tikzpt{7.5}{6}{}{};
   %    \draw (7.5,3.7) node {$QUw_2$};
     %   \sqone(7.5, 4);
     %   \sqtwoL(7.5, 4);
	\end{scope}
     \begin{scope}[Orange]
        \Aone{9}{4}{$QUw_1^2$};
		%\tikzpt{9}{4}{}{};
 % \tikzpt{9}{5}{}{};
 % \tikzpt{9}{6}{}{};
  %     \draw (9,3.7) node {$QUw_1^2$};
      %  \sqone(9, 4);
       % \sqtwoL(9, 4);
	\end{scope}
	\end{tikzpicture}
\caption{The $\cA(1)$-modules up to degree 4 for the tensor $H((M\O_2)^U;\Z/2) \otimes H((B\O_1)^{3\sigma};\Z/2)$. The classes $\alpha = Q(Ut+t^3)$, $\beta = Q(Ut^2+t^4)$.
This module is isomorphic to the one in \cref{fig:KT-}, but we indicate how the marked classes in \cref{fig:KT-} have changed in the modules above.}
%\arun{Are there really only two free summands in deg 4 here? I would have assumed the same number as in fig 4}}\matt{not quite: we are tensing fig 5 with the left of fig 3. but fig 5 is not the same as the green module in fig 3%
    \label{fig:KT+}
\end{figure}
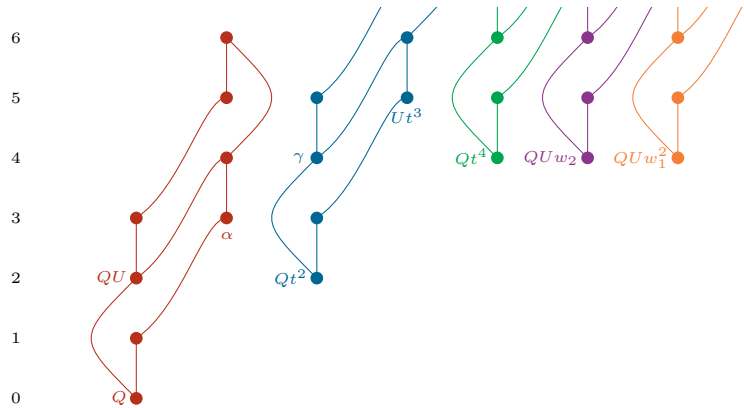

Now we will work out the characteristic long exact sequence for $\mathrm{KT}^+$ bordism in low degrees (as always, to do this we have to assume \cref{char_conj}). The first step is to compute the bordism groups for the structure present on $F$. In \cref{KT_plus_F_str}, we saw that this is a $(B\O_1\times B\O_2, a+b, ab+b^2)$-twisted spin structure, where we again let $a\coloneqq w_1\in H^1(B\O_1;\Z/2)$ and $b\coloneqq w_1\in H^1(B\O_2;\Z/2)$. We will refer to this twisting data as $\tau^+$.
\begin{thm}\label{2nd_part_KT+_thm}
The low-degree $(B\O_1\times B\O_2, a+b, ab+b^2)$-twisted spin bordism groups are
\begin{equation}
    \begin{aligned}
        \Omega_0^\Spin(B\O_1\times B\O_2;\tau^+) &\cong \Z/2\\
        \Omega_1^\Spin(B\O_1\times B\O_2;\tau^+) &\cong \Z/2\\
        \Omega_2^\Spin(B\O_1\times B\O_2;\tau^+) &\cong (\Z/2)^{\oplus 2}\\
        \Omega_3^\Spin(B\O_1\times B\O_2;\tau^+) &\cong \Z/8 \oplus \Z/2\\
        \Omega_4^\Spin(B\O_1\times B\O_2;\tau^+) &\cong (\Z/2)^{\oplus 4}.
    \end{aligned}
\end{equation}
\end{thm}
\begin{proof}
The outline of this proof is identical to that of \cref{KT_minus_2nd} and its constituent lemmas (\cref{klein_split_off,non_klein_calc}); only the specifics of the computation differ. Thus in this proof we focus on the differences.
\begin{itemize}
    \item Define $MV(a+b, ab+b^2)$ analogously to how we defined $MV(a+b, a^2+ab+b^2)$ in \cref{klein_split_off}; then, the analogue of that lemma for $\mathrm{KT}^+$ exhibits a $\ko$-module splitting
    \begin{equation}
        M^\ko(\widehat A\circ\tau^+)\overset\simeq\longrightarrow MV(a+b, ab+b^2)\vee M'
    \end{equation}
    for some $\ko$-module $M'$; the inclusion map on $H_\ko^*$ sends the cohomology of $M'$ to a subspace complementary to the span of $\set{Ua^ib^j: i,j\ge 0}$.
    \item The same element of $\mathrm{GL}_2(\mathbb F_2)$, interpreted as an automorphism of $B\O_1\times B\O_1$, identifies $MV(a+b, ab+b^2)$ with $MV(a, ab)$. This latter $\ko$-module is studied in~\cite[Appendix F]{KPMT20} for applications to ``dpin bordism,'' including a computation of its homotopy groups in degrees $6$ and below in (\textit{ibid.}, Theorem F.1). See also~\cite{WWZ20, SS24, BottSpiral} for related computations.\footnote{If $\sigma_i\to B\O_1\times B\O_1$ denotes the tautological real line bundle associated to the $i^{\mathrm{th}}$ copy of $B\O_1$, then $a+b = w_1(\sigma_1\boxplus 3\sigma_2)$ and $ab+b^2 = w_2(\sigma_1\boxplus 3\sigma_2)$, which implies by \cref{lem:shearing} that
    \begin{equation}
        MV(a+b, ab+b^2) \simeq \ko\wedge (B\O_1\times B\O_1)^{\sigma_1\boxplus 3\sigma_3 - 4}\simeq \ko\wedge (B\O_1)^{\sigma-1}\wedge (B\O_1)^{3\sigma-3}.
    \end{equation}
    The Smith isomorphism $(B\O_1)^{\sigma-1}\simeq \Sigma^{-1}B\O_1$ (the real analogue of \cref{cpx_smith_isom}; see for example~\cite[Lemma 2.6.5]{Koc96}) identifies these homotopy groups with $(\ko\wedge (B\Z/2)^{3\sigma-3})_{*+1}(B\O_1)$, and these groups have been calculated in low degrees by Guo-Ohmori-Putrov-Wan-Wang~\cite[Theorem 14]{GOPWW20} as the \pinp bordism groups of $B\O_1$, or in all degrees by Bruner~\cite[Remark B.5]{Bru14}.}
    \item Analogously to \cref{non_klein_calc}, we obtain an $\cA(1)$-module isomorphism
    \begin{equation}\label{M_prime_Hko}
        H_\ko^*(M') \cong \textcolor{BrickRed}{\Sigma^2\cA(1)} \oplus \textcolor{Green}{\Sigma^3 J} \oplus \textcolor{MidnightBlue}{\Sigma^4 \cA(1)} \oplus \textcolor{Fuchsia}{\Sigma^4 J} \oplus P,
    \end{equation}
    where $P$ is $4$-connected and $J$ is the \emph{Joker}, defined to be $\cA(1)/(\Sq^3)$. The name is due to Adams. In particular, in degrees $4$ and below, the Baker--Lazarev Adams spectral sequence collapses. We draw the decomposition~\eqref{M_prime_Hko} in \cref{fig:KT+pt2}, left, and give the $E_2$-page of the corresponding Baker--Lazarev Adams spectral sequence in \cref{fig:KT+pt2}, right.
    \qedhere
\end{itemize}
\end{proof}

\begin{figure}[!ht]
\begin{subfigure}[c]{0.45\textwidth}
	\begin{tikzpicture}[scale=0.6, every node/.style = {font=\tiny}]
		\foreach \y in {2, ...,10} {
			\node at (-2, \y) {$\y$};
		}
		\begin{scope}[BrickRed]
		      \Aone{0}{2}{$Uc$};
	    \end{scope}
            \begin{scope}[Green]
                \Joker{2.75}{3}{$Ubc$};
	    \end{scope}
            \begin{scope}[MidnightBlue]
                \Aone{4.75}{4}{$Ua^2c$};
            \end{scope}
            \begin{scope}[Fuchsia]
                \Joker{7.5}{5}{$U(a^2c+b^2c)$};
            \end{scope}
	\end{tikzpicture}
\end{subfigure}\qquad
\begin{subfigure}[!ht]{0.4\textwidth}
\includegraphics[width=51mm,scale=.9]{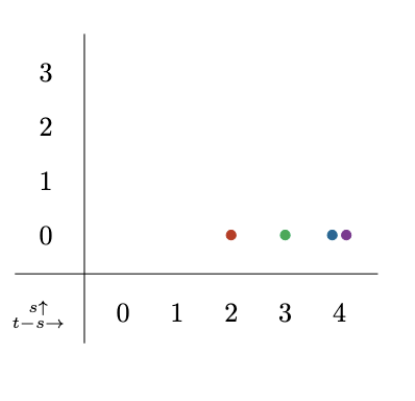}
%\begin{sseqdata}[name= KTplussecond, Adams grading, classes=fill, xrange={0}{4}, yrange={0}{3},scale=0.7, class labels = { left = 0.07em, font=\small },
% class pattern = linear, class placement transform = { rotate = 45 },
%  x label = {$\displaystyle{\substack{s\uparrow \\ t-s\rightarrow}}$},
 % x label style = {font = \small, xshift = -16.5ex, yshift=1.7ex}]]
%\class[BrickRed](2, 0)
%\class[Green](3, 0)
%\class[MidnightBlue](4, 0)
%\class[Fuchsia](4, 0)
%\end{sseqdata}
%\printpage[name=KTplussecond, page=2]
\end{subfigure}
\caption{Left: The $\cA(1)$-module structure on $H_\ko^*(M')$, as in~\eqref{M_prime_Hko}; this summand includes all classes in
degree $4$ and below. Right: The $E_2$-page of the corresponding Adams spectral sequence, which computes a summand of the $2$-completion of $(B\O_1\times B\O_2, a+b, ab+b^2)$-twisted spin bordism. We use this in \cref{2nd_part_KT+_thm}.}
    \label{fig:KT+pt2}
\end{figure}
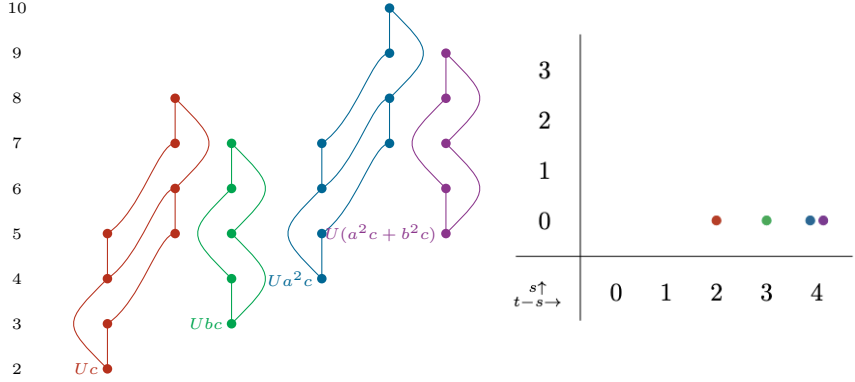

\begin{thm}
\label{KT+_LES_thm}
Assume \cref{char_conj} is true, so that there is a characteristic map of spectra $\mathit{MTKT}^+\to \Sigma^2 M^\MTSpin \tau^+$, and let $F_{\mathrm{KT}^+}$ denote the fiber of this map. Then there is an abelian group $A_+$ of order at least $8$ such that
\begin{equation}
\begin{aligned}
    \pi_0(F_{\mathrm{KT}^+}) &\cong \Z/2\\
    \pi_1(F_{\mathrm{KT}^+}) &\cong  0\\
    \pi_2(F_{\mathrm{KT}^+}) &\cong \Z/2\\
    \pi_3(F_{\mathrm{KT}^+}) &\cong 0\\
    \pi_4(F_{\mathrm{KT}^+}) &\cong A_+,
\end{aligned}
\end{equation}
and the characteristic long exact sequence is as given in \cref{KT+_char_LES}.
\end{thm}
\begin{proof}
Throughout this theorem, we will use the $\mathrm{KT}^+$ bordism groups and the $(B\O_1\times B\O_2, \tau^+)$-twisted spin bordism groups, which we computed in low degrees in \cref{prop:bordcompKT+,2nd_part_KT+_thm}, respectively.

We will first make the following three computations.
\begin{enumerate}
    \item\label{R2_KT+} The characteristic map $R_2\colon \Omega_2^{\mathrm{KT}^+}\to \Omega_0^\Spin(B\O_1\times B\O_2, \tau^+)$ is an isomorphism $\Z/2\to\Z/2$.
    \item\label{R4_KT+} The characteristic map $R_4\colon \Omega_4^{\mathrm{KT}^+}\to \Omega_2^\Spin(B\O_1\times B\O_2, \tau^+)$ , which is a map $(\Z/2)^{\oplus 3}\to (\Z/2)^{\oplus 2}$, is surjective.
    \item\label{Z4_5} The complement of the subgroup $4\Z/8\subset\Z/8\subset \Omega_3^\Spin(B\O_1\times B\O_2)$ is not in the image of the characteristic map.
\end{enumerate}
For~\eqref{R2_KT+}, we first show that $(\RP^2, \pt)$ is the nonzero element of $\Omega_2^{\mathrm{KT}^+}\cong\Z/2$. First we need to show this has a $\mathrm{KT}+$ structure, which amounts to the observation that $w_2(\RP^2)\ne 0$, so indeed $\pt$ is Poincaré dual to $w_2$. The $\Z/2$-valued bordism invariant $(M, F)\mapsto \int_M w_2$ is nonzero on $(\RP^2, \pt)$, so indeed the bordism class of $(\RP^2, \pt)$ is nonzero. The characteristic map $R_2$ thus sends this class to the bordism class of the point, which is also nonzero, so $R_2$ is a nonzero map $\Z/2\to\Z/2$, hence an isomorphism.

For~\eqref{R4_KT+}, we will show that $\Omega_4^{\mathrm{KT}^+}$ has a basis (as a $\Z/2$-vector space) given by $(\RP^4, \varnothing)$, $(\CP^2, \CP^1)$, and $(S^4, \RP^2)$ (the last two are the standard generators of Guillou--Marin bordism~\cite{GM}, and have the $\mathrm{KT}^+$-structures induced by their GM-structures). The reader can check that for all three of these pairs $(M, F)$, that $F$ is indeed Poincaré dual to $w_2(M)$. To show these three bordism classes form a basis, evaluate the $(\Z/2)^{\oplus 3}$-valued bordism invariant
\begin{equation}
    (M, F)\longmapsto \begin{pmatrix}
        \int_M w_1(M)^4\\
        \int_F w_2(M)\\
        \int_F w_1(F)^2
    \end{pmatrix}
\end{equation}
on each of the three bordism classes and observe that the resulting vectors are linearly independent over $\Z/2$. Now that we have the generators, we can evaluate $R_4$ on them to obtain the characteristic submanifolds $F$, which are $\varnothing$, $\CP^1$, and $\RP^2$, and the latter two are linearly independent in $\Omega_2^\Spin(B\O_1\times B\O_2, \tau^+)$, as can be shown by evaluating the bordism invariants $\int_F w_2(V)$ and $\int_F w_1(F)^2$.

To prove~\eqref{Z4_5}, it suffices to show $2\Omega_5^{\mathrm{KT}^+} = 0$, which can be shown in the same way as \cref{prop:bordcompKT+}.

With~\eqref{R2_KT+}--\eqref{Z4_5} established, the rest follows by exactness.
\end{proof}

\begin{figure}[!ht]
    \centering
\begin{tikzcd}
	{k} & {\pi_k(F_{\mathrm{KT}^+})} & {\Omega_k^{\mathrm{KT}^+}} & {\Omega_{k-2}^{\Spin}(B\O_1\times B\O_2, \tau^+)} \\
	0 & \Z/2 \arrow[r,"\cong"]& {\Z/2} & 0 \\
	1 & 0 & 0 & 0 \\
	2 & {\Z/2}\ar[r, "0"]  & {\Z/2}\arrow[r, "\cong"] & {\Z/2}\\%\arrow[from=4-4,to=3-2,in=-150, out=30,"\cong",swap] \\
	3 & {0} & 0 & {\Z/2} \arrow[from=5-4,to=4-2,in=-150, out=30,"\cong"]\\
	4 & {A_+}\ar[r]
    & {(\Z/2)^{\oplus 3}}\ar[r, ->>]& {(\Z/2)^{\oplus 2}}\\
    5 & & & {\Z/8\oplus\Z/2}\arrow[from=7-4,to=6-2,in=-150, out=30]
\end{tikzcd}
    \caption{The characteristic long exact sequence for the $\mathrm{KT}^+$ characteristic structure, which we prove in \cref{KT+_LES_thm}. Here $A_+$ has order at least $8$.}
    \label{KT+_char_LES}
\end{figure}

Finally, we interpret this long exact sequence physically.
\begin{cor}
\label{KT_plus_physics}
Assuming \cref{char_conj}, consider a $k$-dimensional field theory on manifolds with $\mathrm{KT}^+$ structure. For $k = 0$, and $2$, the defect map $\mho_\Spin^{k-1}(B\O_1\times B\O_2, \tau^+)\to\mho_{\mathrm{KT}^+}^{k+1}$ is $0$. For $k= 1$, the defect map is an isomorphism: the obstruction vanishes.
For $k= 3$, the map is nonzero and there is an $A_+$-valued obstruction to pulling the anomaly of the theory back to the defect.
\end{cor}
It would be interesting to determine the precise structure of $A_+$ and compute this obstruction in examples.

%%%%%%%%%%%%%%%%%%%%%%%%%%%%%%%%%%%%%%%%%
\subsubsection{Symmetry Approximations to KT}\label{subsubsection:approxKT}
%%%%%%%%%%%%%%%%%%%%%%%%%%%%%%%%%%%%%%%%%
Similar to the case of approximating the Guillou--Marin structure in \S\ref{subsubsection:approxGM}, we will define tangential structures $\Pin^{\pm}\!\text{-}\O_2 := (\Pin^{\pm}\times \O_2)/(\langle-1,x\rangle)$ to be used to approximate $\mathrm{KT}^\pm$-bordism. We will do everything explicitly in the \pinm case, but the \pinp case follows closely. 
\begin{lem}
    A $\Pin^-\!\text{-} \O_2$ tangential structure is a $(B\O_2,V)$-twisted pin$^-$ structure i.e. it
    fits into the following pullback square:
\begin{equation}
    \begin{tikzcd}[column sep=2.5cm, row sep=2.5cm]
B(\Pin^-\!\text{-} \O_2) \arrow[r,"V"] \arrow[d] \arrow[dr, phantom, "\lrcorner", very near start] & B\O_2 \arrow[d, "w_2"] \\
B\O \arrow[r, "w_2(TM)+w_1(TM)^2",swap] & K(\Z/2,2).
\end{tikzcd}
\end{equation}
\end{lem}

\begin{proof}
    We construct $\Pin^-\!\text{-} \O_2$ from a split  extension arising from the sequence 
    \begin{equation}\label{eq:pinO2seq}
    1 \rightarrow \Z/2 \rightarrow \Pin^-\!\text{-} \O_2 \rightarrow \O \times \O_2 \rightarrow 1\,.
\end{equation}
Therefore  $\Pin^-\!\text{-} \O_2$ arises from a class $\alpha \in H^2(B\O \times B\O_2 ; \Z/2 )$, which parametrizes the extension. The most general form of $\alpha$ is a linear combination of degree two terms built from Stiefel-Whitney classes of $V$ and $TM$, 
\[
    \alpha = \lambda_1 w_2(TM) + \lambda_2 w_1(TM)^2 + \lambda_3 w_2(V) + \lambda_4 w_1(V)^2 + \lambda_5 w_1(V)w_1(TM)
\]
for $\lambda_1,\dotsc,\lambda_5\in\Z/2$. We solve for these coefficients by pulling back to subgroups of the sequence in \eqref{eq:pinO2seq}. 
\begin{itemize}
    \item Pull the extension~\eqref{eq:pinO2seq} back to $\O$, where it becomes
    \begin{equation}
        \shortexact*{\Z/2}{\Pin^-}{\O}.
    \end{equation}
    Using this, we see that when pulling back to $H^*(B\O;\Z/2)$, $\alpha$ is sent to the extension class for $\Pin^-\to\O$, which is $w_2 + w_1^2$. Thus $\lambda_1 =1$ and $\lambda_2 = 1$.
    \item Pulling back along $\O_2\to\Pin^-\text{-}\O_2$, we get the extension
    \begin{equation}
        \shortexact*{\Z/2}{\Pin^+_2}{\O_2},
%        1 \rightarrow \Z/2 \rightarrow \Pin^+ \rightarrow \O_2 \rightarrow 1
    \end{equation}
    which is classified by $w_2(V)$. Thus $\lambda_3 = 1$ and $\lambda_4 = 0$.
 %   we see that $\alpha$ restricts to $w_2(V)$, since this extension is not split.
    \item To compute $\lambda_5$, pull back along the inclusion $\Pin^-\!\text{-}\O_1\hookrightarrow \Pin^-\!\text{-}\O_2$ to obtain the extension
     \begin{equation}
        \shortexact*{\Z/2}{\Pin^-\!\text{-}\O_1}{\O},
    %1 \rightarrow \Z/2 \rightarrow \Pin^-\!\text{-} \O_1 \rightarrow \O_1 \rightarrow 1
\end{equation}
which is classified by $w_2(TM)+w_1(TM)^2$, but $w_1(V)w_1(TM)$ pulls back to a nonzero class, so $\lambda_5 = 0$.
\end{itemize}
We thus see that $\alpha = w_2(V)+w_2(TM)+w_1(TM)^2$. The condition for the square being a pullback is that $w_2(V)= w_2(TM)+w_1(TM)^2$, which is indeed satisfied.
\end{proof}
 We can construct a Smith long exact sequence for the symmetry $\Pin^-\!\text{-} \O_2$ by interpreting it as a twisted pin$^-$ structure. By using the prescription in \S\ref{subsection:smithLESintro} and taking $\xi=\Pin^-$, $X= B\O_2$, and the two bundles over $B\O_2$ to both be $V$, we obtain the following Smith long exact sequence:
 \begin{equation}\label{eq:pinO2exactseq}
       \ldots  \rightarrow \Omega^{\Pin^-}_k((B\O_1)^{\sigma-1}) \xrightarrow{p^*} \Omega^{\Pin^-\text{-}\O_2}_k\xlongrightarrow{S_V} \Omega^{\Pin^-}_{k-2}\left((B\O_2)^{2V-4}\right)\xrightarrow{a} \Omega^{\Pin^-}_{k-1}((B\O_1)^{\sigma-1})\rightarrow \ldots\,.
    \end{equation}
The groups $\Omega_k^{\Pin^-}((B\O_1)^{\sigma-1})$ can be identified with the $(B\O_1\times B\O_1, a, a^2+ab+b^2)$-twisted spin bordism groups, where $\{a,b\}$ is a basis of $H^1(B\O_1\times B\O_1;\Z/2)$. These twisted spin bordism groups are computed in~\cite[Theorem 4.26]{Deb21} in degrees $4$ and below. %\arun{In fact, this has a Smith map to \pinp bordism which splits off as a direct summand, and the other summands are all shifts of $H\Z/2$! Probably not worth mentioning but it is cool. Similarly, $\MTSpin\wedge (B\O-2)^{V-2}$ can be shown to split similarly to spin$^h$ case, which could give a spectral-sequence-free calculation of $\MTSpin\wedge (B\O_2)^V\wedge (B\O_1)^\sigma$ in low degrees if we prefer. This would not be very elaborate, mostly citing other people's hard work}
The other two twisted \pinm bordism theories are new. We provide the computation of 
$\Omega^{\Pin^-\text{-}\O_2}_k$ for $k$ up to 4, and the groups $\Omega^{\Pin^-}_{k-2}\left((B\O_2)^{2V-4}\right)$ can be computed with the long exact sequence. 
\begin{prop}\label{pinmo2calc}
The low-degree \pinm-$\O_2$ bordism groups are:
\begin{equation}
\begin{aligned}
    \Omega_0^{\Pin^-\!\!\text{-}\O_2} &\cong \Z/2\\
    \Omega_1^{\Pin^-\!\!\text{-}\O_2} &\cong \Z/2\\
    \Omega_2^{\Pin^-\!\!\text{-}\O_2} &\cong (\Z/2)^{\oplus 2}\\
    \Omega_3^{\Pin^-\!\!\text{-}\O_2} &\cong \Z/2\\
    \Omega_4^{\Pin^-\!\!\text{-}\O_2} &\cong \Z/4\oplus \Z/2 \oplus \Z/2.
\end{aligned}
\end{equation}
\end{prop}
\begin{proof}
The groups  $\Omega^{\Pin^-\text{-}\O_2}_k$ can be computed using the techniques employed in \S\ref{subsection:KTcharBord} and \S\ref{subsection:pin+KT}. Using the shearing construction on $(B\O_2,V)$-twisted pin$^-$ bordism, allows us to express it as  the spectrum 
 \begin{equation}
     \MTSpin \wedge (B\O_2)^{V-2} \wedge (B\O_1)^{\sigma-1}\,.
 \end{equation}
We want to run the Adams spectral sequence to compute the $2$-completed homotopy groups of this spectrum. (The odd-primary argument is not much different than in the other bordism computations in this paper.) To do so, we need to determine the $\cA(1)$-module structure on $H^*((B\O_2)^{V-2}\wedge (B\O_1)^{\sigma-1};\Z/2)$, which amounts to tensoring together the $\cA(1)$-modules $H^*((B\O_2)^{V-2};\Z/2)$ and $H^*((B\O_1)^{\sigma-1};\Z/2)$. Freed--Hopkins~\cite[\S D.5]{FH21InvertibleFT} compute $H^*((B\O_2)^{V-2};\Z/2)$ as an $\cA(1)$-module in all degrees; in particular, there is an $\cA(1)$-module isomorphism
\begin{equation}\label{pincp}
    H^*((B\O_2)^{V-2};\Z/2) \cong \textcolor{Fuchsia}J \oplus
        \textcolor{Magenta}{\Sigma^2\cA(1)} \oplus
        \textcolor{amber}{\Sigma^4\cA(1)} \oplus P,
\end{equation}
where $P$ is concentrated in degrees $6$ and above. Here $\textcolor{Fuchsia}{J}$ is the Joker, first seen in \cref{M_prime_Hko}. We draw the isomorphism~\eqref{pincp} in \cref{fig:pin-O2module}; see~\cite[Figure 6.6]{Cam17} or~\cite[Figure 5, case $s = 2$]{FH21InvertibleFT} for more pictures continuing into higher degrees.

Now to tensor with $H^*((B\O_1)^{\sigma-1};\Z/2)$. For any $\cA(1)$-module $M$, $\cA(1)\otimes M$ is a free $\cA(1)$-module on a homogeneous vector space basis for $M$, which accounts for the $\textcolor{Magenta}{\Sigma^2\cA(1)}$ and $\textcolor{amber}{\Sigma^4\cA(1)}$ summands. For the Joker we only have to work a little harder.
\begin{lem}
\label{jokertensor}
There is an isomorphism of $\cA(1)$-modules
\begin{equation}
    J\otimes H^*((B\O_1)^{\sigma-1};\Z/2) \cong R_3 \oplus \Sigma\cA(1)
        \oplus \Sigma^2 \cA(1) \oplus P',
\end{equation}
where $P'$ is a free $\cA(1)$-module concentrated in degrees $5$ and above, and $R_3$ is the $\cA(1)$-module defined in~\cite[\S 5.2]{BC18}.
\end{lem}
\begin{proof}
Let $Q_0\coloneqq\Sq^1$ and $Q_1\coloneqq\Sq^1\Sq^2 + \Sq^2\Sq^1$. Then $Q_0^2 = Q_1^2 = 0$, so we may regard any $\cA(1)$-module $M$ as a chain complex in which $Q_0$ or $Q_1$ is the differential. The homology groups of this complex are called \emph{Margolis homology} and denoted $H_*(M;Q_0)$, resp.\ $H_*(M;Q_1)$. Both of these homology theories satisfy a Künneth formula~\cite[Chapter 19, Proposition 17(b)]{Mar83}: for $i = 0,1$,
\begin{equation}
    H_*(M\otimes N; Q_i) \cong H_*(M; Q_i)\otimes H_*(N; Q_i).
\end{equation}
In this proof, let $P\coloneqq H^*((B\O_1)^{\sigma-1};\Z/2)$. Then there are isomorphisms $H_*(P;Q_0) \cong 0$ and $H_*(P; Q_1) \cong \Z/2$ in degree $1$~\cite[Proposition 4.2]{Bru14}, and $H_*(J; Q_1)\cong \Z/2$ concentrated in degree $2$~\cite[\S 3]{AP76}. Thus $H_*(J\otimes P;Q_0)\cong 0$ and $H_*(J\otimes P; Q_1)\cong\Z/2$ in degree $3$. Yu~\cite[Proposition 2.3]{Yu95} shows that if $M$ is a bounded-below $\cA(1)$-module with vanishing $Q_0$-homology and one-dimensional $Q_1$-homology, then $M$ is the direct sum of a free $\cA(1)$-module and exactly one of $\Sigma^t P$, $\Sigma^t H^*((B\O_1)^{1-\sigma};\Z/2)$, $\Sigma^t R_3$, or $\Sigma^t R_5$ for some $t$. Here $R_5$ is the module defined in~\cite[\S 5.3]{BC18}; see also \textit{ibid.}, Figure 23 for a picture of $H^*((B\O_1)^{1-\sigma};\Z/2)$. The rest of the proof is a process of elimination.

First, we show $t = 0$. It suffices to show that the lowest-degree class in $J\otimes P$, which is the unique nonzero class $x$ in degree $0$ of $J\otimes P$ does not generate a free module, and indeed one can check using the Cartan formula that $\Sq^2\Sq^2\Sq^2(x) = 0$, which suffices (see~\cite[\S D.4]{FH21InvertibleFT}). Thus $x$ is the lowest-degree element in an $\cA(1)$-module $M$ isomorphic to one of $P$, $H^*((B\O_1)^{1-\sigma};\Z/2)$, $R_3$, or $R_5$. The Cartan formula implies $\Sq^2(x) \ne 0$ (which is not true for the lowest-degree element in $P$, so $M\not\cong P$), $\Sq^3(x) \ne 0$ (ruling out $R_5$), and $\Sq^2\Sq^1(x)\ne 0$ (ruling out $H^*((B\O_1)^{1-\sigma};\Z/2)$).

The degrees of the free summands in the proposition statement follow by comparing the dimensions of subspaces of elements of a given degree. For example, $J\otimes P$ is two-dimensional in degree $1$, but $R_3$ is only one-dimensional; therefore there must be a $\Sigma\cA(1)$ summand to account for the difference.
\end{proof}
It is possible to prove \cref{jokertensor} directly, without using Margolis homology, but it is a longer and more arduous computation.
%\begin{rem}
%It is possible to use the Adams-Margolis theorem~\cite[Theorem 4.2]{AM71}\footnote{In \textit{loc.\ cit.}, the part of the theorem we need is attributed to Wall.} to show that $P'$ is a free $\cA(1)$-module, similarly to the arguments in~\cite[Chapter 2]{Yu95}. We do not need this and so do not give the details.
%\end{rem}
\begin{cor}
There is an isomorphism of $\cA(1)$-modules
\begin{equation}\label{all_tensor_pieces}
    H^*((B\O_2)^{V-2}\wedge (B\O_1)^{\sigma-1}) \cong
        \textcolor{Fuchsia}{
            R_3 \oplus \Sigma\cA(1) \oplus \Sigma^2 \cA(1) 
        } \oplus \textcolor{Magenta}{
            \Sigma^2\cA(1) \oplus \Sigma^3\cA(1) \oplus \Sigma^4\cA(1)
        } \oplus \textcolor{amber}{\Sigma^4\cA(1)} \oplus P''
\end{equation}
for some $\cA(1)$-module $P''$ concentrated in degrees $5$ and above.
\end{cor}
We draw the decomposition~\eqref{all_tensor_pieces} in \cref{fig:E2pinO2}, left. To run the Adams spectral sequence, we need $\Ext_{\cA(1)}^{s,t}(\textcolor{Fuchsia}{R_3}, \Z/2)$, computed by Yu~\cite[Theorem 3.1]{Yu95} (note that Yu refers to this module as $\Sigma^3 N_3$ or $M(3,3)$). Thus we can draw the $E_2$-page of the Adams spectral sequence in \cref{fig:E2pinO2}, right. All differentials in the relevant range vanish for degree reasons, and the only extension question, in topological degree $4$, is solved by the $h_0$-action.
\end{proof}

%The action of $\cA(1)$ on $H^*((B\O_2)^V;\Z/2)$ is given by 
%\begin{equation}\label{eq:actiononBO2V}
%   \Sq^1 U = w_1(V) \,U, \quad \Sq^2 U = w_2(V) \,U,
%%\end{equation}
%from which we construct the $\cA(1)$-modules given in Figure \ref{fig:pin-O2module}.
\begin{figure}
\centering
	\begin{tikzpicture}[scale=0.6, every node/.style = {font=\tiny}]
		\foreach \y in {0, 1, ...,10} {
			\node at (-2, \y) {$\y$};
		}
  \begin{scope}[Purple]
    \Joker{0}{0}{$U$};
%		\tikzpt{0}{0}{$U$}{};
%         \tikzpt{0}{1}{}{};
%         \tikzpt{0}{2}{}{};
%             \tikzpt{1.5}{3}{}{};
%             \tikzpt{1.5}{4}{}{};
%         \sqone(0, 0);
%          \sqtwoL(0, 0);
%            \sqone(1.5, 3);
%            \sqtwoCR(0, 1);
%            \sqtwoCR(0, 2);
	\end{scope}
 \begin{scope}[Magenta]
 \Aone{2}{2}{$w_1^2 U$};
% \tikzpt{3}{2}{$w^2_1U$}{};
%      \tikzpt{3}{3}{}{};
 %     \tikzpt{3}{4}{}{};
 %     \tikzpt{3}{5}{}{};
 %     \tikzpt{4.5}{5}{}{};
 %     \tikzpt{4.5}{6}{}{};
 %     \sqone(3, 2);
 %     \sqone(3, 4);
 %     \sqone(4.5, 5);
 %     \sqtwoL(3, 2);
 %      \sqtwoCR(3, 3);
 %      \sqtwoCR(3, 4);
 \end{scope}
  \begin{scope}[amber]
  \Aone{5}{4}{$w_2^2U$};
%      \tikzpt{6}{4}{$w_1^2U$}{};
%       \tikzpt{6}{5}{}{};
%        \sqone(6, 4);
%          \tikzpt{6}{6}{}{};
%           \sqtwoL(6, 4);
 \end{scope}
	\end{tikzpicture}
\caption{The $\cA(1)$-module structure in low degrees for $H^*((B\O_2)^{V-2};\Z/2)$; the summand depicted here includes all classes in degrees $5$ and below. The Steenrod actions on the Thom class are given by $\Sq^1U = w_1(V)U$ and $\Sq^2 U = w_2(V) U$. This figure is adapted from~\cite[Figure 6.6]{Cam17}.}
%\arun{This figure is in Jonathan Campbell's bordism paper, fig 6.6. If we want, we could just direct the reader there and not replicate it here. I am happy either way}\matt{we can include a remark about campbell}
    \label{fig:pin-O2module}
\end{figure}
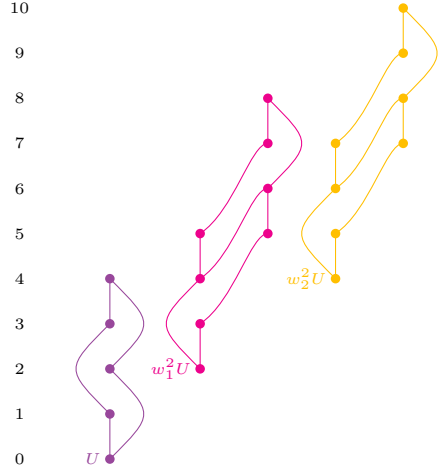
%We now compute the $\cA(1)$-module structure on the tensor product of the modules in \cref{fig:pin-O2module} with the module on the right of \cref{fig:tensorforKT-}. The resulting $\cA(1)$-modules and $E_2$-page are computed in \cref{fig:E2pinO2}. 

\begin{figure}[!ht]
\centering
\begin{subfigure}[c]{0.6\textwidth}
	\begin{tikzpicture}[scale=0.5, every node/.style = {font=\tiny}]
		\foreach \y in {0, 1, ..., 10} {
			\node at (-2, \y) {$\y$};
		}
		\begin{scope}[Purple]
  \tikzpt{0}{0}{$U$}{};
  \foreach \y in {1, ..., 10} {
    \tikzpt{0}{\y}{}{};
  }
    \tikzpt{1.5}{3}{}{};
    \tikzpt{1.5}{4}{}{};
         \sqone(0, 0);
          \sqtwoL(0, 0);
            \sqone(1.5, 3);
            \sqtwoCR(0, 1);
            \sqtwoCR(0, 2);
             \sqone(0, 2);
             \sqtwoL(0, 3);
              \sqone(0, 4);
               \sqtwoR(0, 4);
        \sqone(0, 6);
        \sqone(0, 8);
        \begin{scope}
            \clip (-1, 9) rectangle (2, 10.6);
            \sqone(0, 10);
        \end{scope}
        \sqtwoL(0, 7);
        \sqtwoR(0, 8);

        \Aone{2.25}{1}{$tU$};
        \Aone{4.25}{2}{$t^2 U$};
		\end{scope}
 \begin{scope}[Magenta]
 \Aone{6.25}{2}{$w_1^2U$};
 %\tikzpt{3}{2}{$w^2_1U$}{};
 %     \tikzpt{3}{3}{}{};
 %     \tikzpt{3}{4}{}{};
 %     \tikzpt{3}{5}{}{};
 %     \tikzpt{4.5}{5}{}{};
 %     \tikzpt{4.5}{6}{}{};
 %     \sqone(3, 2);
 %     \sqone(3, 4);
 %     \sqone(4.5, 5);
 %     \sqtwoL(3, 2);
 %      \sqtwoCR(3, 3);
 %      \sqtwoCR(3, 4);
 \end{scope}
 \begin{scope}[Magenta]
    \Aone{8.25}{3}{$tw_1^2U$};
 \end{scope}
 \begin{scope}[Magenta]
    \Aone{10.25}{4}{$t^2w_1^2U$};
 \end{scope}
   \begin{scope}[amber]
    \Aone{12.25}{4}{$w_2^2U$};
      %\tikzpt{6}{4}{$w_2^2U$}{};
      % \tikzpt{6}{5}{}{};
      %  \sqone(6, 4);
      %    \tikzpt{6}{6}{}{};
      %     \sqtwoL(6, 4);
 \end{scope}
	\end{tikzpicture}
\end{subfigure}
\hspace{8mm}
\begin{subfigure}[!ht]{0.3\textwidth}
\includegraphics[width=43mm,scale=.6]{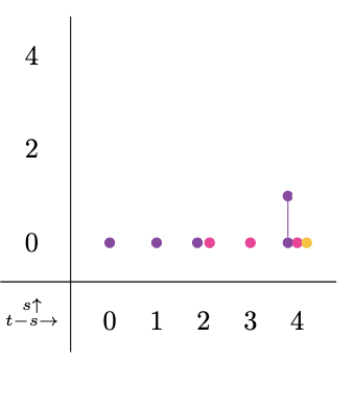}
%\begin{sseqdata}[name=pinO2, Adams grading, classes=fill, xrange={0}{4}, yrange={0}{4},xscale=0.6, yscale=0.6, x tick
%step=1, y tick step=2, class labels = { left = 0.07em, font=\small },
 % x label = {$\displaystyle{\substack{s\uparrow \\ t-s\rightarrow}}$},
 % x label style = {font = \small, xshift = -14.5ex, yshift=1.7ex}]]
%\begin{scope}[Purple]
%\class(0, 0)
%\class(4, 1)
%\class(4, 0)\structline
% \class(1, 0)
% \class(2, 0)
%\end{scope}
%\begin{scope}[Magenta]
%\class(2, 0)
%\class(3, 0)
%\class(4, 0);
%\end{scope}
%\begin{scope}[amber]
%	\class(4, 0)
%\end{scope}
%\begin{scope}[white]
%     \class(4, 1)
%     \class(4, 1)
%\end{scope}
%\end{sseqdata}
%\printpage[name=pinO2]
\end{subfigure}
\caption{Left: the $\cA(1)$-module structure on $H^*((B\O_2)^{V-2}\wedge (B\O_1)^{\sigma-1})$ in low degrees. The pictured submodule includes all classes in degrees $4$ and below. Right: The $E_2$-page of the Adams spectral sequence computing $\Omega_*^{\Pin^-\text{-}\O_2}$. We use this spectral sequence in the proof of \cref{pinmo2calc}.} %\arun{If the left side of this figure is the result of tensoring Figure $7$ with $H^*((B\O_1)^{\sigma-1});\Z/2)$, then I think there should be more $\cA(1)$ summands?}}
%s of $\Pin^-\text{-}\O_2$ up to degree 4. Right: The $E_2$-page of the Adams spectral sequence computing the nontrivial homotopy groups in degree less than 5.}
    \label{fig:E2pinO2}
\end{figure}
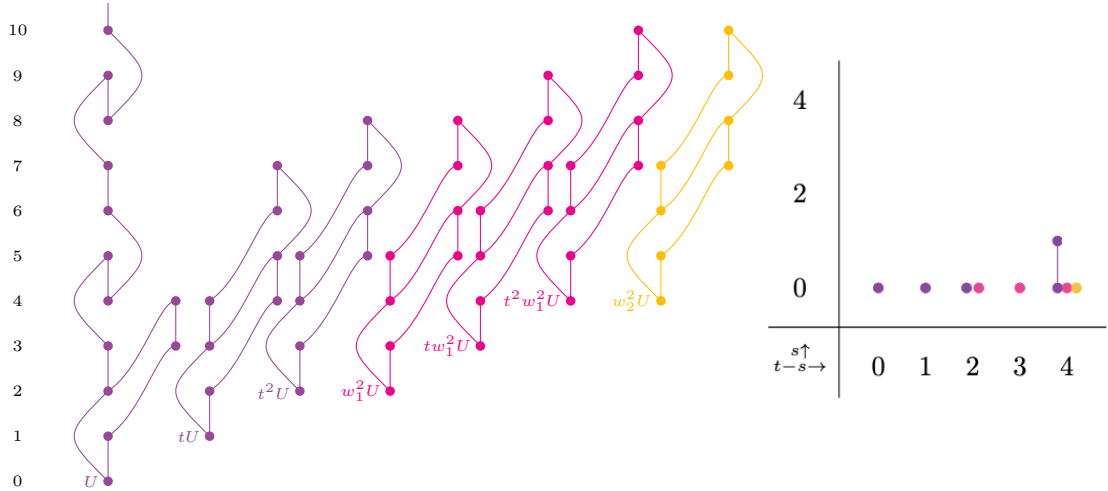

%%%%%%%%%%%%%%%%%%%%%%%%%%%%%%%%%%%%%
\section{Topological Obstructions to Spontaneous Symmetry Breaking}
\label{section:SSB}
%%%%%%%%%%%%%%%%%%%%%%%%%%%%%%%%%%%%%
In this section, we shift our focus and consider the situation where a theory with a finite (higher-form) global symmetry is placed on a $d$-dimensional closed manifold $M$, and study the topological properties that prevent the symmetry from spontaneously breaking. Following \cref{def:SSB}, the obstruction to spontaneous symmetry breaking on $M$ corresponds to the inability to define a domain wall to probe the higher-form symmetry. The support of this required domain wall is a submanifold $F$, defined as the Poincaré dual to the symmetry's cohomology class.

The standard way of presenting spontaneous symmetry breaking in QFT contexts is on non-compact spacetimes. In particular, it is only a feature that makes sense in the infinite volume limit. In condensed-matter systems, one typically considers an increasing family of finite spatial regions and takes the inductive limit over these regions to define the algebra of quasi-local observables. Although each finite system admits a unique, symmetry-invariant ground state, spontaneous symmetry breaking emerges in the thermodynamic limit from the appearance of macroscopically distinct states that are locally indistinguishable but globally inequivalent. 

Nevertheless, in both condensed-matter and modern high-energy contexts, it is often important to define the theory on a closed spatial manifold, possibly endowed with nontrivial global topology. This formulation captures the intrinsically finite extent of physical materials and allows global topological features to play a meaningful role in the characterization of the phase.
%Nevertheless, we will extrapolate the ideas of using a globally finite spacetime to more general relativistic quantum field theories where the notion of an infinite-volume limit is may not be sharply defined.
Our guiding assumption is that, at sufficiently large scales, physical observables are effectively governed by finite global structure, so that treating spacetime as infinite is, at best, an idealization rather than a fundamental requirement.

Under this assumption, we will show how the ideas of the Pontryagin--Thom theorem, which have already seen application in \S\ref{section:Defects}, can also be applied to finding obstructions to spontaneous symmetry breaking. This gives a more ``intrinsic'' way of realizing analogous characteristic structures which is only dependent on the theory placed on $M$, and not on inserting defects. %This point of view also reflects the fact that the dynamical problem of spontaneous symmetry breaking has some aspects that can be understood topologically. 

Choose a natural number $n$ and a finite group $G$, abelian if $n\ge 1$.
If an $n$-form $G$-symmetry were to spontaneously break, it would create a domain wall along a codimension $(n+1)$-manifold $F$. The fact that the domain wall comes from symmetry breaking means that the manifold where it is supported must be Poincaré dual to the background field of the $n$-form symmetry, see \cite{Hason:2020yqf}.  By Theorem \ref{construction:PT}, we can formulate the following definition: 
\begin{defn}\label{def:obstructionpoly}
The obstruction to breaking a finite $n$-form $G$-symmetry on a manifold $M$ is the obstruction to finding a lift of the map $f\colon M \to K(G,n+1)$ implementing the symmetry to a map $\widetilde f\colon M\SO_{n+1}$ such that $f = U\circ\widetilde f$:
\begin{equation}
    \begin{tikzcd}
        && M\SO_{n+1} \arrow[d, "U"]\\
        M \arrow[rr,"f",swap ] \arrow[urr,dotted,"\widetilde f"] & &K(G,n+1)\,.
    \end{tikzcd}
\end{equation}
\end{defn}
Here $U\colon M\SO_{n+1}\to K(G, n+1)$ is the $G$-cohomology Thom class. Like in the second part of \cref{construction:PT}, in \cref{def:obstructionpoly} we can replace $M\SO_{n+1}$ with $M\O_{n+1}$ if the theory has a time-reversal symmetry, permitting us to work on unorientable manifolds. By expanding the number of possible backgrounds that our theory can be placed on, this limits the possibilities for $G$. One can only take $G=\Z/2$ with time-reversal acting trivially on this group. We will use this setup as the  focus for the explicit computations in the remainder of this section.
%\begin{rem}
%In \cref{def:obstructionpoly}, if $G = \Z/2$, one could consider unorientable manifold $M$ and change from $M\SO_{n+1}$ to $M\O_{n+1}$ using the mod $2$ Thom class. However, the physical interpretation of such a lift is not clear in a unitary quantum field theory
%\end{rem}

The following standard result of obstruction theory can be used to find the obstruction class:
\begin{prop}\label{prop:obstr}
  For $E$ the fiber of the map $f_{n+1}\colon M\SO_{n+1} \to K(G,n+1)$, the primary obstruction to breaking a finite $n$-form $G$-symmetry is a class $\mathfrak o_{n+1}\in H^{n+2}(K(G,n+1);\pi_{n+2}(E))$. %where $i> n+1$.
\end{prop}
The primary obstruction is a necessary condition for the existence of a lift; in general, it is not a sufficient condition.
\begin{rem}
\Cref{prop:obstr} identifies the \emph{universal} primary obstruction, i.e.\ for choosing a lift on all manifolds. If one is interested in a specific manifold $M$ with $n$-form $G$-symmetry, the primary obstruction is the pullback of $\mathfrak o_{n+1}$ by the map $f\colon M\to K(G, n+1)$.
\end{rem}

We will give explicit obstructions to breaking 1-form and 2-form symmetries when $G=\Z/2$, which is relevant in many theories and perhaps most notably $\mathrm{SU}_2$ Yang-Mills theory in 4d resp.\ 6d supersymmetric theories \cite{Morrison:2012np,DelZotto:2015isa,Braun:2021sex,Apruzzi:2022dlm,DelZotto:2022fnw}.  Given that the only non-vanishing homotopy group for $K(G,{n+1})$ is $\pi_{n+1}$, the long exact sequence in homotopy groups implies $\pi_*(E)=\pi_*(M\O_{n+1})$ for $* > n+1$. One therefore just needs to understand the homotopy groups of $M\O_{n+1}$, and this can be done using the following theorem of Thom:
\begin{thm}[{\cite[Theorem II.7]{Thom}\label{thm:thom}}]
The map on homotopy groups induced by the map $\Sigma M\O_{k-1}\to M\O_k$ (which is the Thomification of the standard map $B\O_{k-1}\to B\O_k$) is an isomorphism in degrees $*<k$.
    %If $i < k$, then the  then the homotopy groups $\pi_{k+i}(M\O_k)$ do not depend on $k$.
\end{thm}
\noindent In particular, this means that for $i<k$ the homotopy groups $\pi_{k+i}(M\O_k)$ can be computed using the stable homotopy groups of $M\O$, which Thom computed~\cite{Thom}.
\begin{rem}
\Cref{thm:thom} can be generalized to the Thom spectra of more general tangential structures, including in particular $M\SO_k$. The stable homotopy groups of the limiting object $M\SO$ are also known, due to Wall~\cite{Wal60}.
\end{rem}
Starting off with a $\Z/2$ 1-form symmetry, we couple it to a background field $B \in H^2(M;\Z/2)$ where $M$ is a closed manifold. The failure for the map of spaces $M\O_{2} \to K(\Z/2,2)$ to be an equivalence is more pronounced in high dimensions. In particular, when the ambient manifold $M$ has dimension less than or equal to 4, then duals can always be constructed so there is no primary obstruction to defining a theory with spontaneously broken $\Z/2$ one-form symmetry on any such manifold. The first dimension when duals cannot always be found is 5. Some candidate theories in 5d with 1-form symmetry would include those in \cite{Bhardwaj:2020phs,Morrison:2020ool,Albertini:2020mdx,Genolini:2022mpi}. There are also many examples of theories with discrete 2-group symmetry \cite{Apruzzi:2021vcu,Apruzzi:2021mlh}, and understanding spontaneous symmetry breaking there would require understanding first breaking the 1-form symmetry subgroup \cite{Brennan:2023mmt}. For notational convenience, let $\mathcal{K} = K(\Z/2,2)$ and let $B\in H^2(\mathcal{K};\Z/2)$ denote the tautological class.
\begin{prop}[Thom~\cite{Thom}]
\label{pi4mo2}
 %       \item 
 $\pi_4(M\O_2)\cong\Z$, so that the primary obstruction $\mathfrak o_2$ from \cref{prop:obstr} is an element of $H^5(\mathcal K; \Z)$.
 \end{prop}
\begin{prop}[Suzuki~\cite{suzuki1958stiefel}]
\label{suzuki}
 %   
 %   \begin{enumerate}
 %       \item  
 $H^5(\mathcal{K};\Z) \cong \Z/4$ with the generator given by $ \frac{1}{4} \Box_{\Z/4} \mathfrak{P}(B)$ where  $\mathfrak{P}: H^2(;\Z/2)\to H^4(-;\Z/4)$ denotes the Pontryagin square and $\Box_{\Z/4}$ is the Bockstein for the sequence 
    \begin{equation}
        0 \longrightarrow \Z \xlongrightarrow{\cdot\, 4} \Z \longrightarrow \Z/4 \rightarrow 0\,.
    \end{equation}
\end{prop}

 Using \cref{pi4mo2,suzuki}, we can identify the primary obstruction for a one-form $\Z/2$ symmetry.
 \begin{prop}\label{prop:what_is_O2}
 $\mathfrak o_2 = \frac{1}{2} \Box_{\Z/4} \mathfrak{P}(B)\in  H^5(\mathcal{K};\Z) \cong \Z/4$, i.e.\ twice the generator.
%    \end{enumerate}
\end{prop}

\begin{proof}
Let $U\colon M\O_2\to \mathcal K$ denote the map corresponding to the Thom class and consider the pullback map $U^*\colon H^5(\mathcal{K};\Z)\to H^5(M\O_2;\Z)=\Z/2$. To understand this map, let us also consider its mod-2 reduction. The mod-2 reduction of the generator of $H^5(\mathcal{K};\Z)$ is $\Sq^2 \Sq^1 B$, while the mod-2 reduction of the generator of $H^5(M\O_2, \Z)$ is $(w_2w_1 + w_1^3)U$. By straightforward calculation, we see that $U^*(\Sq^2 \Sq^1 B) = (w_2w_1 + w_1^3)U $. Therefore, we see that the pullback map $U^*\colon H^5(\mathcal{K};\Z)\to H^5(M\O_2;\Z)$ maps the generator to the generator. Thus, $\frac{1}{2} \Box_{\Z/4} \mathfrak{P}(B)$, which is twice the generator of $H^5(\mathcal{K};\Z)$, will map to the trivial element in $H^5(M\O_2;\Z)=\Z/2$.
\end{proof}

\begin{lem}\label{cor:obstruction1form}
Let $\widetilde{\mathfrak o}_2$ denote the image of $\Sq^2\Sq^1 B$ in the quotient
\begin{equation}
    H^5(\mathcal K; \Z/2)/(\mathrm{Im}(\Sq^1\colon H^4(\mathcal K;\Z/2)\to H^5(\mathcal K;\Z/2)).
\end{equation}
Then the pullback of $\mathfrak o_2$ to a manifold with one-form $\Z/2$ symmetry 
vanishes if and only if the pullback of $\widetilde{\mathfrak o}_2$ vanishes.
\end{lem}
The upshot is that the obstruction $\frac{1}{2} \Box_{\Z/4} \mathfrak{P}(B)$ can equivalently be characterized as a class in  $H^5(\mathcal{K};\Z/2)$ modulo the image of $\Sq^1$, which will be a more convenient way for us to describe it.
\begin{proof}
   By \cref{prop:what_is_O2}, $\mathfrak{o}_2=\frac{1}{2} \Box_{\Z/4} \mathfrak{P}(B)$. Consider the cohomology of $\mathcal{K}$ with respect to the two sequences $0 \rightarrow \Z \rightarrow \Z \rightarrow \Z/2 \rightarrow 0$, and $0 \rightarrow \Z/2 \rightarrow \Z/4 \rightarrow \Z/2 \rightarrow 0$.  This gives the following map between long exact sequences:
\begin{equation}\label{eqn:Z4Z2LES}
\begin{gathered}
   \begin{tikzcd}
      \ldots \arrow[r] & H^4(\cK;\Z/2) \arrow[r,"\Box_{\Z/2}"] \arrow[d,"\cong"]&  H^5(\cK;\Z) \arrow[r,"\cdot 2"] \arrow[d,"\bmod 2"]& H^5(\mathcal{K};\Z) \arrow[r,"\bmod 2"] \arrow[d,"\bmod 4"]& H^5(\mathcal{K};\Z/2) \arrow[r,"\Box_{\Z/2}"] \arrow[d,"\cong"]& \ldots \\
       \ldots \arrow[r]& H^4(\cK;\Z/2) \arrow[r,"\Sq^1"]&  H^5(\cK;\Z/2) \arrow[r,"\cdot  2"]& H^5(\mathcal{K};\Z/4) \arrow[r,"\bmod 2"]& H^5(\mathcal{K};\Z/2) \arrow[r,"\Sq^1"]& \ldots
   \end{tikzcd}
   \end{gathered}
   \end{equation}
    Let $c\coloneqq \mathfrak o_2\bmod 4$, which is a class in $H^5(\mathcal K;\Z/4)$. Then $c\bmod 2 = 0$ in $H^5(\mathcal K;\Z/2)$, because the diagram commutes and $\mathfrak o_2\bmod 2 = 0$ (this because $\mathfrak o_2$ is twice the generator of $H^5(\mathcal K;\Z)$, then using exactness along the top row). Therefore by exactness along the bottom row, $c = 2\widetilde{\mathfrak o}_2$ for some class $\widetilde{\mathfrak o}_2\in H^5(\mathcal K;\Z/2)$, and  the ambiguity is the image of the previous map in the long exact sequence, which is $\Sq^1$. Thus $\widetilde{\mathfrak o}_2$ is uniquely defined in $H^5(\mathcal K;\Z/2)/(\mathrm{Im}(\Sq^1))$.

    The destiny bond between $\mathfrak o_2$ and $\widetilde{\mathfrak o}_2$ is a standard diagram chase.
    \begin{itemize}
        \item If $\mathfrak o_2 = 0$, then $c = 0$, so $\widetilde{\mathfrak o}_2 = 0\in H^5(\mathcal K;\Z/2)$ solves the equation $2\widetilde{\mathfrak o}_2 = c$. Since $\widetilde{\mathfrak o}_2$ is uniquely defined modulo the image of $\Sq^1$, this suffices to imply that it vanishes.
        \item Conversely, suppose $\widetilde{\mathfrak o}_2 = 0$, so that $c = 0$ too. The multiplication-by-$4$ map $H^5(\mathcal K;\Z)\to H^5(\mathcal K;\Z)$ is $0$, since $H^5(\mathcal K;\Z)\cong\Z/4$, so by exactness, the modulo $4$ reduction map $H^5(\mathcal K;\Z)\to H^5(\mathcal K;\Z/4)$ is injective. Thus, since $c =0$ and $\mathfrak o_2\bmod 4 = c$, then $\mathfrak o_2 = 0$.
    \end{itemize}
%   Let $\widetilde{\mathfrak o}_2\coloneqq\mathfrak o_2\bmod 2$, which is a class in $H^5(\mathcal K;\Z/2)$
%   Taking mod 2 of $\mathfrak{o}_2$ in $H^5(\mathcal{K};\Z)$ on the top row gives a class  $c \in H^5(\mathcal{K};\Z/2)$. The nontrivial class 2$\mathfrak{o}_2 \in H^5(\cK;\Z)$ is the image under multiplication by 2 and has an image in $H^5(\mathcal{K};\Z/4)$ under mod 4 reduction. Therefore $c$ cannot be in the image of $\Sq^1(x)$ for $x\in H^4(\mathcal{K};\Z/2)$ due to exactness and pulling back to $H^5(\cK;\Z)$.
%\end{proof}%
%\begin{cor}
%    As a class in $H^5(\mathcal{K};\Z/2)$, the obstruction is $\Sq^2 \Sq^1 B$.
%\end{cor}
%\begin{proof}
Now we identify $\widetilde{\mathfrak o}_2$ with $\Sq^2\Sq^1 B$. The construction of $\widetilde{\mathfrak o}_2$ in the previous argument implies that $\widetilde{\mathfrak o}_2$ is in the image of the mod $2$ reduction map $H^5(\mathcal K;\Z)\to H^5(\mathcal K;\Z/2)$. Since the mod $2$ reduction map $\Z\to\Z/2$ factors through $\Z/4$, the same is true for the induced map on cohomology: thus,
% Since the obstruction can equivalently be thought of as an element in $ H^5(\mathcal{K};\Z/4)$, we can
consider the long exact sequence 
 \begin{equation}\label{second_Z4_diagram}
     \ldots \rightarrow H^5(\mathcal{K};\Z/4) \xlongrightarrow{\bmod 2} H^5(\mathcal{K};\Z/2) \xlongrightarrow{\Sq^1}H^6(\mathcal{K};\Z/2) \rightarrow \ldots \,,
 \end{equation}
Being in the image of the mod $2$ reduction map in~\eqref{second_Z4_diagram} is thus equivalent to being in \\$\ker(\Sq^1\colon H^5(\mathcal K;\Z/2)\to H^6(\mathcal K;\Z/2))$; this kernel is $\set{0, \Sq^2\Sq^1B}$ (plus the indeterminacy of $\mathrm{Im}(\Sq^1)$). Since $\mathfrak o_2\ne 0$, then $\widetilde{\mathfrak o_2}\ne 0$ and therefore $\widetilde{\mathfrak o}_2$ must equal $\Sq^2\Sq^1 B$.
%    and therefore we need a nontrivial element that vanishes upon applying $\Sq^1$. The only class in $H^5(\mathcal{K};\Z/2)$ that satisfies this
%    property is $\Sq^2\Sq^1 B$.
\end{proof}

\begin{rem}
    In the case of a $\Z/2$ 0-form symmetry, the Thom class map $M\O_1 \to K(\Z/2,1)$ is a homotopy equivalence, so there is no topological obstruction to the symmetry spontaneously breaking.
    There is also no topological obstruction to breaking a $\mathrm U_1$ valued 0-form or equivalently a $\Z$ valued 1-form symmetry, as $B\mathrm U_1 = K(\Z,2)$ and the Thom class map $M\SO_2 \to K(\Z,2)$ is also a homotopy equivalence. The latter is in agreement with the Coleman--Mermin--Wagner (CMW) theorem for discrete symmetries \cite{Gaiotto:2014kfa} which states that a $p$-form symmetry in $D$-spacetime dimensions is obstructed from breaking if $p \geq D-1$.
    \end{rem}

    As dimension 5 is the edge case for which an obstruction is present, we address an important instance when the obstruction vanishes for 5-dimensional manifolds:
 \begin{prop}\label{prop:spin5}
    If $M$ is a spin $5$-manifold then the obstruction vanishes. If $M$ is only an oriented $5$-manifold then the obstruction does not vanish in general.
\end{prop}
\begin{proof}
  When $M$ is spin, then we compute $\Sq^2\Sq^1B$ on $M$ by using the fact that the second Wu class is $w_2(TM) + w_1(TM)^2=0$. This implies $\Sq^2\Sq^1B= (w_2(TM)+w_1(TM)^2)\Sq^1(B)=0$, and the obstruction vanishes. 
    In the oriented case, the generator of degree 5 oriented bordism  is the Wu manifold $W\coloneqq\SU_3/\SO_3$,\footnote{Barden~\cite[\S 1]{Bar65} named this manifold in reference to Wu's work~\cite{Wu50} (see also Dold~\cite{Dol56}) on a closely related manifold. See~\cite[Footnote 9]{LOT21} for more on the history of the Wu manifold.} which has $\Z/2$ cohomology $H^*(W;\Z/2) = \Z/2[z_2,z_3]/(z^2_2,z^2_3)$ (this follows from~\cite[Lemma 1.1]{Bar65} and Poincaré duality), with $w_2=z_2$ and $w_3=z_3$ (\textit{ibid.}), and $\Sq(z_2) = z_2 + z_3$ and $\Sq(z_3) = z_3 + z_2z_3$ (see~\cite[\S 3]{Flo73}, where this is attributed to Calabi). Thus $\Sq^2\Sq^1(z_2)$ is not in the image of $\Sq^1$.
\end{proof}
\begin{rem}
In fact, $\int \widetilde{\mathfrak o}_2$ is an oriented bordism invariant $\Omega_5^\SO(\mathcal K)\to\Z/2$. This amounts to asserting that the indeterminacy in $\widetilde{\mathfrak o}_2$ vanishes, which follows from the Wu formula on an oriented manifold.
\end{rem}

\begin{exm}\label{ex:5dtheory}
Bhardwaj--Schäfer-Nameki~\cite[\S 3.3.4]{Bhardwaj:2020phs} give an example of spontaneous symmetry breaking in a 5d Kaluza-Klein theory. Begin with the 6d SCFT with gauge group $\SU_3$ and matter in the defining representation, and compactify on $S^1$, where the monodromy around the circle acts by complex conjugation. The six-dimensional theory enjoys both a $\Z/3$ 2-form and $\Z/2$ 1-form global symmetry. 
The authors show that the  $\Z/3$ symmetry is unaffected by the monodromy and hence the $\Z/3$ symmetry descends to the 5d KK theory as a 1-form symmetry. However, there is no contribution to the 1-form symmetry of the 5d KK theory from the $\Z/2$ 1-form symmetry of the 6d SCFT. In particular, the $\Z/2$ 1-form symmetry breaks in 5d and the total 1-form symmetry is given solely by the group $\Z/3$. The authors are able to justify the symmetry breaking through studying the geometric properties used to define the 5d theory \cite[Equation 3.106]{Bhardwaj:2020phs}.
%preserved when the theory is compactified to five-dimensions. Furthermore, the contribution to the 1-form symmetry of the 5d KK theory from the 1-form symmetry of the 6d SCFT can also be computed via modifying the geometry for the 5d KK theory. 
%where it becomes a 1-form symmetry, and the $\Z/2$ 1-form symmetry breaks in five-dimensions. 

Since this 5d theory is supersymmetric, it contains fermions, hence must be placed on  manifolds with spin structure.  \Cref{prop:spin5} then says that there is no topological  obstruction for spontaneous symmetry breaking, which is in agreement with Bhardwaj--Schäfer-Nameki's results showing that symmetry breaking occurs.
\end{exm}

\begin{exm}\label{ex:4dYM}
    Pure $\SU_2$ gauge theory in four spacetime dimensions at zero temperature has a 1-form center $\Z/2$-symmetry which acts on Wilson lines. This is oftentimes referred to as the electric 1-form symmetry. At zero temperature, the theory is confining and hence the 1-form symmetry does not break. At finite temperature, with the $\theta$-angle not equal to 0 or $\pi$,  the theory can undergo a deconfining transition into a Coulomb phase where the 1-form $\Z/2$ symmetry spontaneously breaks \cite[Section 7]{Gaiotto:2017yup}. The fact that the deconfining transition exists which spontaneously breaks the 1-form $\Z/2$ symmetry is in line with \cref{prop:what_is_O2}, which says that there is no topological obstruction to spontaneous breaking of 1-form symmetries in theories below five dimensions.
\end{exm}

\subsection{\texorpdfstring{Obstruction to spontaneously breaking $\Z/2$ higher-form symmetry}{}}\label{subsection:breaking2form}
%%%%%%%%%%%%%%%%%%%%%%%%%%%%%%%%%%%%%%

We now examine the case of a $\Z/2$-valued 2-form symmetry, and consider the lowest dimension where the symmetry is obstructed from spontaneously breaking. Although we might expect this dimension to be $6$, we will show that this turns out not to be the case. Following the discussion in the previous section on probing obstructions, we consider the first dimension at which the map $U\colon M\O_{3}\rightarrow K(\Z/2,3)$ is not an equivalence. Using \cref{thm:thom} we know that $\pi_5(M\O_3) \cong \Z/2$; this is the first homotopy group of $M\O_3$ with degree higher than $3$ that is nontrivial.

\begin{lem}\label{prop:breaking2form}
Pullback by the Thom class $U\colon M\O_3\to K(\Z/2, 3)$ is an injective homomorphism $H^6(K(\Z/2, 3);\Z/2)\to H^6(M\O_3;\Z/2)$.
 %There is no class in  $H^6(K(\Z/2,3); \pi_5(M\O_3))$ that is trivial in $H^6(M\O_3;\Z/2)$.
\end{lem}
\begin{proof}
Let $C \in H^3(K(\Z/2,3);\Z/2)$ denote the tautological class and $V\to B\O_3$ be the tautological vector bundle. We will abuse notation mildly and let $U$ denote $U^*(C)$ (i.e.\ the Thom class); which $U$ we refer to will be clear from context.

$H^6(K(\Z/2, 3);\Z/2)$ is spanned by $\Sq^2\Sq^1C$ and $C^2$.
%The map $U\colon M\O_3 \rightarrow K(\Z/2,3)$ implies that the class $C$ pulls back to $U$ in $H^3(M\O_3;\Z/2)$.  
%The generators of $H^6(K(\Z/2,3);\Z/2)$ are given by $\Sq^2 \Sq^1 C$ and $ C^2$.  
Pulling back $\Sq^2 \Sq^1 C$ gives $(w_2(V)w_1(V)+w_1(V)^3)U$. Using the fact that $C^2 = \Sq^1 \Sq^2 U= \Sq^3 U$ means $C^2$ pulls back to $w_3(V)U$. 
\end{proof}
Therefore, one needs to look in dimensions even higher than $6$ to identify the primary obstruction, i.e.\ to find an element in %$\phantom{aa}$  
$H^{n+1}(K(\Z/2,3);\pi_n(M\O_3))$ which is trivial in $H^{n+1}(M\O_3;\pi_n(M\O_3))$, i.e.\ to observe the primary obstruction to symmetry breaking. However, this involves understanding $\pi_*(M\O_3)$ where $*\geq 6$ and is beyond the scope of Theorem \ref{thm:thom}, and hence using the stable homotopy groups of $M\O$ is not viable.
A general story persists for $\Z/2$ valued $n$-form symmetries when $n\geq 2$. Indeed, for an $\Z/2$ valued $n$-form symmetry with $n>2$, since $\pi_{n+1}(M\O_n)=0$ and  $\pi_{n+2}(M\O_n)=\Z/2$ the first cohomology that could support an obstruction to breaking this symmetry is $H^{n+3}(K(\Z/2, n); \Z/2)$. But this group is generated by $\Sq^2 \Sq^1 C$ and $\Sq^3 C$, and neither vanish when pulled back to $H^{n+3}(M\O_{n+1};\Z/2)$. This analysis can be straightforwardly generalized to any group of the form $K(\Z/2^k;n)$ using the results of Serre \cite{Serre}.

\begin{exm}\label{ex:6d2form}
Bhardwaj--Schäfer\text{-}Nameki \cite[\S 2.1]{Bhardwaj:2020phs} study 6d SCFTs that arise from F-theory. Every known F-theoretic 6d theory admits a tensor branch of vacua i.e.\ part of the moduli space obtained by giving vacuum expectation values (VEVs) to scalars in tensor multiplets. Along the tensor branch, the theory carries massive BPS string excitations and these strings are charged under the 2-form gauge fields with charge given by a positive definite matrix $\Omega^{ij}$ of rank $r$. The 2-form symmetry of 6d SCFTs in the presence of the charged strings is given by a group 
\begin{equation}
    \cT = \prod^r_i\, \Z/n_i\,,
\end{equation}
and it was pointed out that this group can also spontaneously break on the tensor branch. A direct consequence of \Cref{prop:breaking2form} is that there is no topological obstruction for breaking a finite 2-form symmetry on 6d manifolds. Hence, spontaneously breaking the 2-form symmetry $\cT$ in 6d SCFTs is not forbidden by topology, which would agree with the  claims of Bhardwaj--Schäfer\text{-}Nameki.
\end{exm}

\section{Discussion}

 We have set out to give an answer to Question \ref{question:main1} by introducing characteristic structures as a way of keeping track of the placement of defects along submanifolds, and the tangential structures that are induced from the defect. We furthermore presented a quantitative way of understanding when there is an obstruction for a finite symmetry to spontaneously break. Strides have also been made in our understanding of Question \ref{q:CharandSmith} by explicitly showing that there are maps that approximate the characteristic long exact sequence. Some information about the groups in the characteristic long exact sequence can also be gleaned from exactness, as was done for the GM case. 

 We conclude by offering some interesting future directions to explore. First of all, it will be very interesting if we can connect the bordism group calculations with concrete quantities in QFT. We believe that such connection can serve as another quantity that is RG invariant and provides nonperturbative arguments for phases and phase diagrams. Secondly, it will be important both mathematically and physically to write down the characteristic long exact sequence in an explicit way, by proving \cref{char_conj} and identifying the fiber of the map of spectra. This, among other things, serves as an ``improvement'' of the Smith long exact sequence.

 It would also be interesting to find some context for characteristic structures in situations that concern theories of quantum gravity. In particular, it would be interesting to tie this in with corner symmetries \cite{Freidel:2023bnj}, which involve picking a codimension 2 boundary in spacetime. Since characteristic pairs naturally give a hard boundary condition to the theory living outside the codimension 2 defect, it is not unreasonable to think that there is some connection. Moreover, with a clearer understanding of the dimensions where spontaneous symmetry breaking may encounter obstructions, it would be intriguing to explore whether certain gravity theories manifest such obstructions. Identifying such an example could lend further support to the idea that symmetries should be gauged, rather than broken when lifting to the UV theory.

\section{Conflict of Interest and Data}
All authors certify that they have no affiliations with or involvement in any organization or entity with any financial interest or non-financial interest in the subject matter or materials discussed in this manuscript.  There is no data available
for this article to declare.

\bibliographystyle{alpha}
\bibliography{anomaly}

\end{document}